\documentclass{article}
\usepackage[toc,page,header]{appendix}
\usepackage{minitoc}
\doparttoc % Tell to minitoc to generate a toc for the parts
\faketableofcontents % Run a fake tableofcontents command for the partocs
\usepackage[toc,page,header]{appendix}
\usepackage{minitoc}
\doparttoc % Tell to minitoc to generate a toc for the parts
\faketableofcontents % Run a fake tableofcontents command for the partocs
\newenvironment{smalleralign}[1][\small]
 {\par\nopagebreak\leavevmode\vspace*{-\baselineskip}%
  \skip0=\abovedisplayskip
  #1%
  \def\maketag@@@##1{\hbox{\m@th\normalfont\normalsize##1}}%
  \abovedisplayskip=\skip0
  \align}
 {\endalign\ignorespacesafterend}
\usepackage{balance}
\usepackage{url}            % simple URL typesetting
\usepackage{booktabs}       % professional-quality tables
\usepackage{amsfonts}       % blackboard math symbols
\usepackage{nicefrac}       % compact symbols for 1/2, etc.
\usepackage{mathrsfs}
\usepackage{pgf,tikz}
\usepackage{setspace}
\usetikzlibrary{arrows}
\usepackage{tikz-network}
\usepackage{multirow,array}

\usepackage{microtype}      % microtypography
\usepackage{algorithm}
\usepackage{graphicx}
\usepackage{subfigure}
\usepackage{mathtools,amsmath,amssymb,amsbsy}
\usepackage{amsthm}
\usepackage{color}
\usepackage{csquotes}
\usepackage{enumitem}
\usepackage{helvet}
\usepackage{courier}
\usepackage{hyperref}
\usepackage{tcolorbox}
\usepackage{tikz,pgfplots,tikz-3dplot}
\usetikzlibrary{arrows,shapes,positioning,calc,intersections,through}\tdplotsetmaincoords{55}{110} \usepackage{wrapfig}
\pgfplotsset{compat=newest} 
\pgfplotsset{plot coordinates/math parser=false} 
\usepackage{IEEEtrantools}
\usepackage[cal=boondoxo]{mathalfa}
\allowdisplaybreaks

\newtheorem{theorem}{Theorem}
\newtheorem{definition}{Definition}
\newtheorem{proposition}{Proposition}

\newtheorem{lemma}{Lemma}

\newtheorem*{remark-non}{Remark}

\newtheorem{innercustomgeneric}{\customgenericname}
\providecommand{\customgenericname}{}
\newcommand{\newcustomtheorem}[2]{%
  \newenvironment{#1}[1]
  {%
   \renewcommand\customgenericname{#2}%
   \renewcommand\theinnercustomgeneric{##1}%
   \innercustomgeneric
  }
  {\endinnercustomgeneric}
}
\usepackage{algorithm}
\usepackage{algorithmic}
\renewcommand{\algorithmicinput}{\textbf{Input:}}
\renewcommand{\algorithmicoutput}{\textbf{Output:}}
\newlength{\commentindent}
\makeatletter
\renewcommand{\algorithmiccomment}[1]{\unskip\hfill\makebox[\commentindent][r]{$\rhd$~#1}\par}
\newcustomtheorem{customtheorem}{Theorem}
\newcustomtheorem{customlemma}{Lemma}
\newcustomtheorem{customproposition}{Proposition}

\newcommand{\ourmethod}{\textbf{SP}\text{ot}$-$\textbf{AC}}

\usepackage[accepted]{icml2021}

\icmltitlerunning{Learning in Nonzero-Sum Stochastic Games with Potentials}

\begin{document}

\twocolumn[
\icmltitle{Learning in Nonzero-Sum Stochastic Games with Potentials}



\icmlsetsymbol{equal}{*}

\begin{icmlauthorlist}
\icmlauthor{David Mguni}{hu}
\icmlauthor{Yutong Wu}{ca}
\icmlauthor{Yali Du}{ucl}
\icmlauthor{Yaodong Yang}{hu,ucl}
\icmlauthor{Ziyi Wang}{ca}
\icmlauthor{Minne Li}{ucl}
\icmlauthor{Ying Wen}{sh}
\icmlauthor{Joel Jennings}{hu}
\icmlauthor{Jun Wang}{ucl}
\end{icmlauthorlist}

\icmlaffiliation{hu}{Huawei R\&D UK}
\icmlaffiliation{ca}{Institute of Automation, Chinese Academy of Sciences}
\icmlaffiliation{ucl}{University College London, UK}

\icmlaffiliation{sh}{Shanghai Jiao Tong University}
\icmlcorrespondingauthor{David Mguni}{davidmguni@hotmail.com}
\icmlcorrespondingauthor{Yaodong Yang}{yaodong.yang@outlook.com}

\icmlkeywords{Machine Learning, ICML}

\vskip 0.3in
]



\printAffiliationsAndNotice{} % otherwise use the standard text.

\begin{abstract}

% Multi-agent reinforcement learning (MARL) has produced many successes in tackling multi-agent systems (MAS). Although MARL has become effective in tackling discrete cooperative game scenarios, MARL has yet to penetrate settings in which agents seek to fulfil their own individual goals. This has confined MARL to a small subset of real-world scenarios and has left many MAS the beyond reach of MARL. In this paper, we take the first step towards enabling a new generation of MARL solvers that can handle MAS with \textit{nonzero-sum} payoff structures and in which the agents' use \textit{continuous actions}. In particular, we develop a first method that allows agents to learn how to solve MAS described by a subclass of nonzero-sum stochastic games known as  continuous-stochastic potential games. We prove our framework enables independent RL agents to learn their best-response Nash equilibrium strategies in polynomial time. We demonstrate our framework tackling previously unsolvable tasks in \textit{Coordination Navigation problems}, \textit{network routing} and an \textcolor{red}{\textit{optimal fishery problem}} and show state of the art MARL baselines fail in such scenarios. \\
% In particular, we develop a method that allows agents to learn how to solve a class of nonzero-sum stochastic games known as continuous-stochastic potential games (c-SPGs). Compared to discrete cooperative games where all agents share one common reward, c-SPGs model real-world scenarios in which agents seek to fulfil their individual goals. 

Multi-agent reinforcement learning (MARL) has become effective in tackling discrete cooperative game scenarios. However, MARL has yet to penetrate settings beyond those modelled  by team and zero-sum games, confining it to a small subset of multi-agent systems.  In this paper, we introduce a new generation of MARL learners that can handle \textit{nonzero-sum} payoff structures and continuous settings. In particular, we study the MARL problem in a class of games known as stochastic potential games (SPGs) with continuous state-action spaces. Unlike cooperative games, in which all agents share a common reward, SPGs are capable of modelling real-world scenarios where agents seek to fulfil their individual goals. 
We prove theoretically our learning method, $\ourmethod$, enables independent agents to learn Nash equilibrium strategies in \textit{polynomial time}. We demonstrate our framework tackles previously unsolvable tasks such as \textit{Coordination Navigation} and \textit{large selfish routing games} and that it outperforms the state of the art MARL baselines such as MADDPG and COMIX in such scenarios. 

% \yaodong{i think the first two sentences had better speak out the topic: potential game. So far, it is unclear what exactly game we can trying to solve. Can we start with something like: we study the multi-agent learning problem in stochastic potential games with continuous state-action space. Potential games describe multi-agent systems with nonzero-sum pay-off structure. Compared to discrete cooperative games where all agents share one common reward and thus confined to limited cases, potential games enjoy the benefit of modelling real-world scenarios in which agents seek to fulfil their individual goals.}
\end{abstract}

\section{Introduction}
% Moreover, systems such as complex organisms also rely on self-organising cells to perform tasks \cite{sherwood2015human}.
% \yaodong{the logic from the previous sentence to here should be, may also cause conjestion for other drivers so that other drivers get delayed? rather than avoid collision?}\DM{Yes, I think the sentence should be changed to better accommodate the example. I'll modify}\yaodong{i think this paragrah need some reference, also, not sure why we need to mention self-organization}\DM{Will add a reference, self-organising implies they have their own goals. not sure words like independent or autonomous would fit?}. 

Many real-world systems give rise to multi-agent systems
(MAS); traffic network systems with autonomous vehicles \cite{ye2015multi,zhou2020smarts}, network packet routing systems \cite{wiering2000multi} and financial trading \cite{mariano2001simulation}
are some examples. In these systems, self-interested
agents act in a shared environment, each seeking
to perform some pre-specified task. Each agent’s actions
affect the performance of other agents and may even prevent them from completing their tasks altogether. For example, autonomous vehicles seeking to arrive at their individual destinations must avoid colliding with other vehicles. Therefore to perform their task the agents must account for other agents’ behaviours. 

There is therefore a great need for reinforcement learning (RL) agents with their own goals to learn to perform in MAS. In these scenarios, agents are not required to behave as a team nor as perfect adversaries. 
% \yaodong{why we can rule out these two extreme cases? it is unclear to me} \DM{changed slightly}
These settings are modelled by \textit{nonzero-sum} stochastic games (SGs) whose solution concept is a fixed point known as Nash equilibrium (NE). An NE describes the stable point in which all agents respond optimally to the actions of other agent. Computing the NE is therefore central to solving MAS. 
% 
% When the interaction is one-off and when the action set is countable (and finite), the agents' rewards for each joint action can be described by a matrix or more generally in \textit{normal form game}~\cite{osborne1994course}. Game theory (GT), is the mathematical formalism that is used to extract predictions about outcomes in games. 
% For distributed systems, computing the MPE allows us to predict outcomes and help promote desirable outcomes when multiple actors act independently. 
% In general, the MPE enables making predictions about multi-agent systems whose behaviour is governed by the complex interaction between its agents.
% 
% 

Despite its fundamental importance for solving many MAS, computing the NE of \textit{any} SG with a general payoff structure remains an open challenge \cite{yang2020overview}. Presently, methods to compute NE in SGs that are neither zero-sum nor team settings are extremely scarce and impose limiting  assumptions. As such, the application of these methods is generally unsuitable for real world MAS \cite{shoham2008multiagent}. Moreover, finding NE even in the simple case of normal form games (where agents take only a single action) is generally intractable when the game is nonzero-sum \cite{chen2009settling}. 

Among multi-agent reinforcement learning (MARL) methods are a class of algorithms known as  independent learners e.g. independent Q learning \cite{tan1993multi}. These algorithms ignore actions of other agents and are ill-suited to tackle MAS and often fail to learn 
% due to the appearance of a \textit{nonstationary} environment during learning
\cite{hernandez2017survey}. In contrast, algorithms such as MADDPG \cite{lowe2017multi}, COMA \cite{foerster2018counterfactual} and QDPP \cite{yang2020multi} include a centralised critic that accounts for the actions of all agents. To date  none of these algorithms have been proven to converge in SGs  that are neither team nor zero-sum (adversarial). Additionally, these methods suffer from combinatorial growth in complexity with the number of agents \cite{yang2019alpha} leading to prohibitively expensive computations in some systems.
 There is also a noticeable lack of MARL methods that can handle continuous spaces which is required for tasks such as physical control \cite{bloembergen2015evolutionary}. This has left MARL largely unable to solve various practical tasks such as multi-agent Mujoco \cite{de2020deep} which remains an open challenge. This is in contrast to discrete counterpart settings e.g. Starcraft micro-management in which MARL has had notable success  \cite{peng2017multiagent}.

In this paper, we address the challenge of solving MAS with payoff structures beyond zero-sum and team game settings in continuous systems. In particular, we develop a MARL solver that computes the NE within a new subclass of continuous nonzero-sum SGs, namely continuous stochastic potential games (c-SPGs)  in which the agents' interaction at each stage has a potential game property. Lastly, our solver 
avoids combinatorial complexity with the number of agents.

% *******\yaodong{the paragraph elaborates what we do technically, however, a reader may be very thristy to know what kind of problem we try to solve before seeing the solution? maybe adding 1-2 sentences sumamarzing what is the prolem}\DM{fully agree, will modify}

% \yaodong{i feel the contirbutions should match the challenges we define in the former paragraphs }

Our framework is developed through theoretical results that enable the NE of some SGs to be found tractably. First, we formalise a construction of continuous SGs in which the interaction between agents at each stage can be described by a PG. Thereafter, we show that the NE of the SG can be computed by solving a \textit{dual} Markov decision process (MDP) whose solution \textit{exactly coincides} with the NE of the original SG. 
% In particular, for a class of SGs , there exists a function known as a dynamic potential function (DPF) which can be used to compute the MPE of SGs within an MDP.
This converts the problem of finding a fixed point NE of an (a priori unknown) nonzero-sum SG to solving an (unknown) MDP whose solution as we show, can be found tractably using a new distributed variant of actor-critic methods, which we call $\ourmethod$.

The paper is organised as follows: after the related work next,  we present our construction of c-SPGs in Sec.~\ref{sec_cspgs}. We continue in Sec.~\ref{sec:planning} to present a simple planner when the environment model is given and prove that c-SPGs have dual representations as MDPs. A polynomial-time fitted Q-learning solver (\textbf{SP}ot\textbf{Q}) is then given to find the NE in this setting. In Sec. \ref{sec:learning}, we extend the learning method and propose an actor-critic variant ($\ourmethod$) that solves c-SPGs in unknown environments. A fully distributed variant is also provided that scales with the number of agents. Robustness analysis is followed and we show that the method closely approximates the NE solution when the construction of the potential function has small estimation errors. Lastly in Sec.~\ref{sec:experiments}, we conduct detailed ablation studies and performance tests on various tasks and conclude the paper.

\section{Related Work} 
\label{sec_relwork}

MARL has been successful in zero-sum scenarios \cite{grau2018balancing} and settings of homogeneous agents with population sizes that approach infinity \cite{mguni2018decentralised,yang2018mean} and team game scenarios \cite{peng2017multiagent}. However, the restrictions on the payoffs therein means that these models are usually far away from many real-world scenarios, prohibiting the deployment of MARL therein. There have been few attempts at computing NE in settings outside of team and zero-sum SGs. Most notably is Nash Q-learning \cite{hu2003nash}; it however imposes stringent assumptions that force the SG to resemble a team game. For example, in \cite{hu2003nash} at each iteration a unique Pareto dominant NE must exist and be computed which is generally unachievable. `Friend or foe' learning \cite{littman2001friend} establishes convergence to NE in two-player \textit{coordination games} but requires known reward functions and solving a linear program at each time step. 
\citet{zhang2020bi} adopts the \emph{stackelberg equilibrium} as the learning target. 
More recently, \citet{lowe2017multi} suggests an actor-critic method (MADDPG) with centralised training on the critic. Nevertheless \citet{lowe2017multi} do not tackle SGs outside of the zero-sum or cooperative cases in either theoretical results or experiments. In particular, the experiments in \cite{lowe2017multi} are all aimed at either adversarial (zero-sum) or the fully cooperative settings.

Very recently \cite{zhang2021gradient} consider an SG setting in which all agents' value functions are assumed to satisfy a global PG condition, that is, the incentive of all agents to change their \textit{policies} can now be expressed using a single global function. As noted in their discussion, without further qualification, this assumption is rather strong and difficult to verify except in the case in which all agents share the same objective. In a later work, \cite{Leonardos2021GlobalCO} consider an SG setting with a PG game property while imposing conditions that either i) reduce the SG to a linear combination of normal form games and removes all planning aspects or ii) limit the agents' interaction to a term in their reward that does not depend on either the state or the agent's own actions. The latter condition (ii) results in an SG which is a restrictive case of our SG, in particular, our SG captures a richer, more general set of strategic interactions between agents (see Sec. \ref{sec:link_pgs} for a more detailed discussion).

% To date, learning in nonzero-sum games has eluded MARL in all but \textit{static} potential games (PGs) \cite{zhang2019multi,perkins2015mixed}\yaodong{the last sentecen is irrelevant to the idea of this paragraph}. 

% In SG theory, additive reward additive transitions (ARAT) SGs and separable rewards and separable in transitions (SR-SIT) SGs in which the agents do not affect the transitions of the game and the rewards are separable in the states and actions \cite{shoham2008multiagent,flesch2007stochastic}.
% These assumptions severely restrict the strategic interaction between agents and limit applicability. 
% In \cite{marden2012state} state based PGs were studied in which a notion of state is incorporated into the PG framework. This is a degenerate case of the current setting in which the agents do not seek strategies that maximise a future stream of rewards but rather the reward corresponding to a single action in a current state. For static games, a number of recent advances have led to efficient computation of NE in PGs \cite{leslie2006generalised,leslie2011equilibrium}.

We tackle a subclass of SGs which satisfy a PG condition at each stage game. We then show that with this construction, a potentiality property can be naturally extrapolated to the value functions of the SG without imposing restrictive assumptions. With this we prove that the NE of the game can be learned by independent agents without the need to impose restrictive assumptions as in \cite{hu2003nash,littman2001friend} and in a way that scales with the number of agents; in contrast to centralised critic methods that scale combinatorially.

\section{Continuous Stochastic Potential Games}\label{section:SPG_formalism}
% \yaodong{these margin between section and non-indexed subsections are wasting space. and there are many in this darft. you could just use a textbf at the beginning of each papragraph}
\textbf{Continuous Stochastic Games}\\
% \yaodong{the SG (also very common stuff) part takes too much space. I think wording wise it can be trimmed half before definition 1. }\DM{agreed, will comment out the fat}
% .
%
% In an SG, at a given time, the agents simultaneously play one of many possible \textit{stage games} $\mathcal{M}$ which are indexed by states that lie within some state space $S$. The outcome of each stage game $\mathcal{M}(s)$ depends on the joint action executed by the agents where $a^i_s\in \mathcal{A}_i$ is the action taken by agent $i\in \mathcal{N}$ and $s\in \mathcal{S}$ is the state of the world and $\mathcal{A}_i$ is the action set for agent $i$. 
% Given some stage game $\mathcal{M}(s)$ for $s\in \mathcal{S}$, the agents simultaneously execute a joint action $\boldsymbol{a}_s=(a^1_s,a^2_s,\ldots,a^N_s)\in \boldsymbol{\mathcal{A}}$ and immediately thereafter, each agent $i\in \mathcal{N}$ receives a payoff $R_i(s,\boldsymbol{a}_s)$, the state then transitions to $s'\in \mathcal{S}$ with probability $P(s';s,\boldsymbol{a}_s)$ where the game $\mathcal{M}(s')$ is played.
% 
MAS are modelled by SGs \cite{shoham2008multiagent,shapley1953stochastic}. An SG is an augmented MDP involving two or more agents $\{1,2,\ldots, N\}=:\mathcal{N}$ that simultaneously take actions over many (possibly infinite) rounds. Formally, a continuous SG is a tuple $\mathcal{G}=\langle \mathcal{N},\mathcal{S},\left(\mathcal{A}_{i}\right)_{i\in\mathcal{N}},P,\left(R_{i}\right)_{i\in\mathcal{N}},\gamma\rangle$ where $\mathcal{S}$ is the set of states, $\mathcal{A}_i\subset\mathbb{R}^q$ is an action set and $R_{i}:\mathcal{S}\times\boldsymbol{\mathcal{A}}\to\mathcal{P}(D)$ is the distribution reward function for agent $i\in\mathcal{N}$ where $D$ is a compact subset of $\mathbb{R}$ and lastly, $P:\mathcal{S} \times \boldsymbol{\mathcal{A}} \times \mathcal{S} \rightarrow [0, 1]$ is the probability function $P:\mathcal{S} \times \boldsymbol{\mathcal{A}} \times \mathcal{S} \rightarrow [0, 1]$ describing the system dynamics where $\boldsymbol{\mathcal{A}}:=\times_{i=1}^N\mathcal{A}_i$. 

In an SG, at each time $t\in 0,1,\ldots,$ the system is in state $s_t\in\mathcal{S}$ and each agent $i\in\mathcal{N}$ takes an action $a^i_t\in\mathcal{A}_i$. The \textit{joint action}\ $\boldsymbol{a}_t=(a^1_t,\ldots, a^N_t)\in\boldsymbol{\mathcal{A}}$  produces an immediate reward $R_i(s_t,\boldsymbol{a}_t)$ for agent $i\in\mathcal{N}$ and influences the next-state transition which is chosen according to $P$.  
Using a (parameterised) \textit{Markov strategy}\footnote{A Markov strategy requires as input only the current state (and not the game history or other agents' actions or strategies).} $\pi_{i,\boldsymbol{\eta}^i}: \mathcal{S} \times \mathcal{A}_i \rightarrow [0,1]$ to select its actions, each agent $i$ seeks to maximise its \textit{individual} expected returns as measured by its value function: $
v^{\pi^i,\pi^{j\neq i}}_i(s)=\mathbb{E}\left[\sum_{t=0}^\infty \gamma^tR_i(s_t,\boldsymbol{a}_t)\big|\boldsymbol{a}_t\sim(\pi^i_{\boldsymbol{\eta}^i},\pi^{j\neq i}_{\boldsymbol{\eta}^{-i}})\right]$ where $\boldsymbol{\eta}^i\in E_i\subset\mathbb{R}^l$ and $\Pi_i$ is a compact Markov strategy space. 
A \textit{pure strategy} (PS) is a map $\pi_i: \mathcal{S}\to\mathcal{A}_i $, for any $i\in\mathcal{N}$ that assigns to any state an action in $\mathcal{A}_i$. 
% 
% A \textit{behavioural strategy} (stochastic policy) is a probability distribution in $\mathcal{A}_i$ with density $\pi_i(\cdot|s): \mathcal{S}\times\mathcal{A}_i\to[0,1] $.
% 
% \footnote{By Kuhn's theorem \cite{ritzberger2002foundations}, if the agent retains a history of previous actions and states, each behavioural strategy has an equivalent \textit{mixed strategy}.} 

% We denote the space of pure strategies by $\{\Pi^{ps}_i\}_{i\in\mathcal{N}}$ and write 
% % Note that $\Pi^{ps}_i\subset \Pi_i$ since any PS assigns to any state the Dirac measure with probability mass concentrated at a point. 
% and $\boldsymbol{\Pi}:=\times_{i\in\mathcal{N}}\Pi_i$.
We denote the space of joint policies by $\boldsymbol{\Pi}:=\times_{i\in\mathcal{N}}\Pi_i$; where it will not cause confusion (and with a minor abuse of notation) we use the shorthands $\pi_{i}\equiv\pi^{i}_{\boldsymbol{\eta}_i}$ and  $f^{(\pi^{i},\pi^{-i})}(s)= f(s,\pi^i,\pi^{-i})\equiv\mathbb{E}_{\pi^i,\pi^{-i}}[f(s,a^i,a^{-i})]$.

SGs can be viewed as a sequence of \textit{stage games} $\{\mathcal{M}(s)\}_{s\in \mathcal{S}}$ that take place at each time step where $\mathcal{M}(s)=\langle \left(\mathcal{A}_i\right)_{i\in\mathcal{N}},\left(R_{i}(s)\right)_{i\in\mathcal{N}},\mathcal{N}\rangle$. Therefore, at each time step a stage game is played and then the game transitions to the next stage game which is selected according to $P$.

% \DM{More carefully, define Borel sigma-algebra on $\mathcal{A}_i$ then define mixed strategy as an element of the set of all prob. measures --- might be able to omit some of these technicalities.}
% 
% In SGs with discrete state space, the SG is constructed using a collection of \textit{stage games} each of which are normal form games which are played at each round after which the game transitions to another stage game. With a continuous state space, such a construction is not possible. 

% \subsection*{Potential Games}\label{sec:potential_games}
% We now give the definition of (state-based) PGs at the stage game level which is a central notion in our analysis:
% 

% The condition is fundamental to the structure of a state-based PG. In particular, it extends the notion of potentiality to the state input.

\begin{figure}\label{example_cSPGS_fig}
    \centering
    \subfigure{
		 \includegraphics[width=0.45\columnwidth]{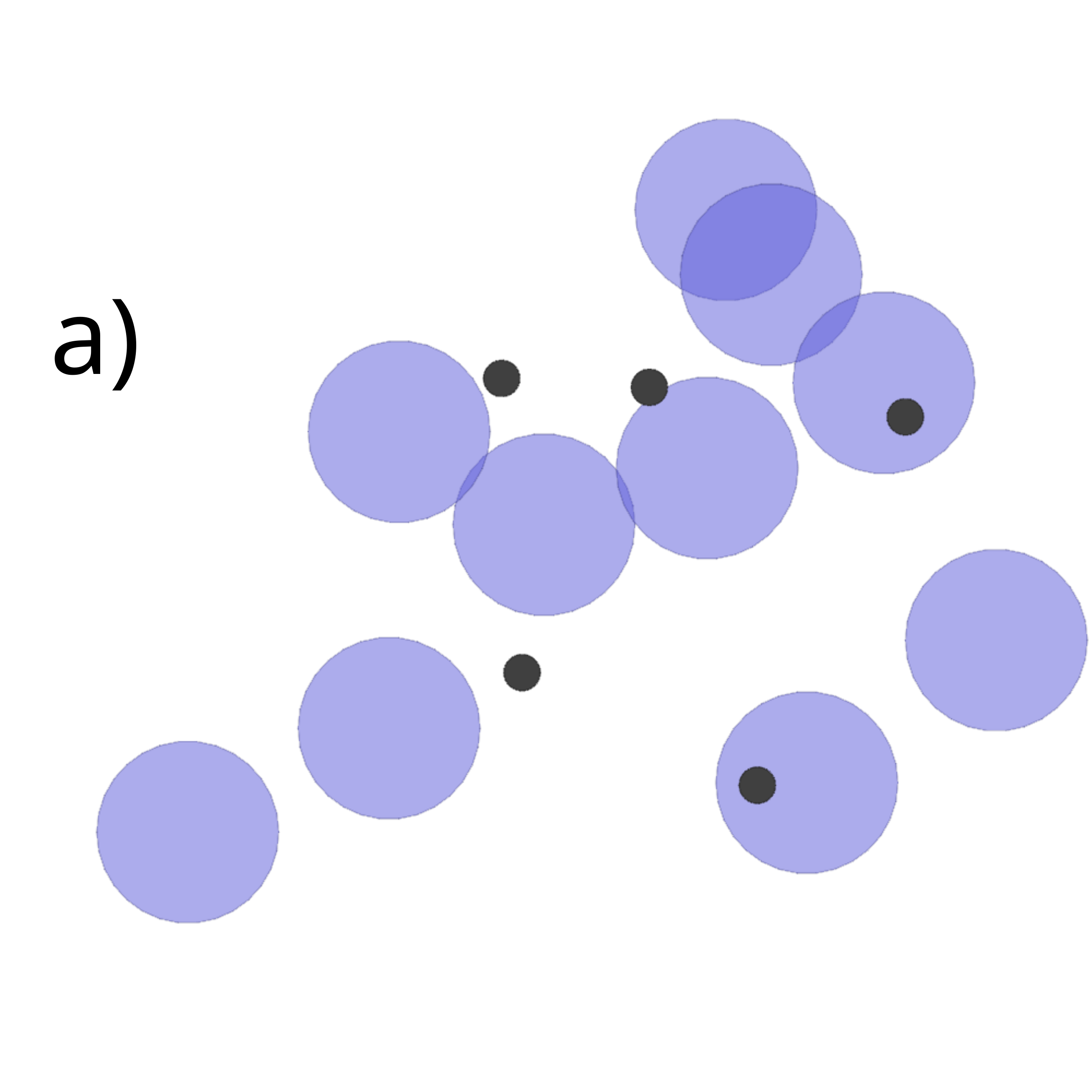}
    }
	\subfigure{
		 \includegraphics[width=0.45\columnwidth]{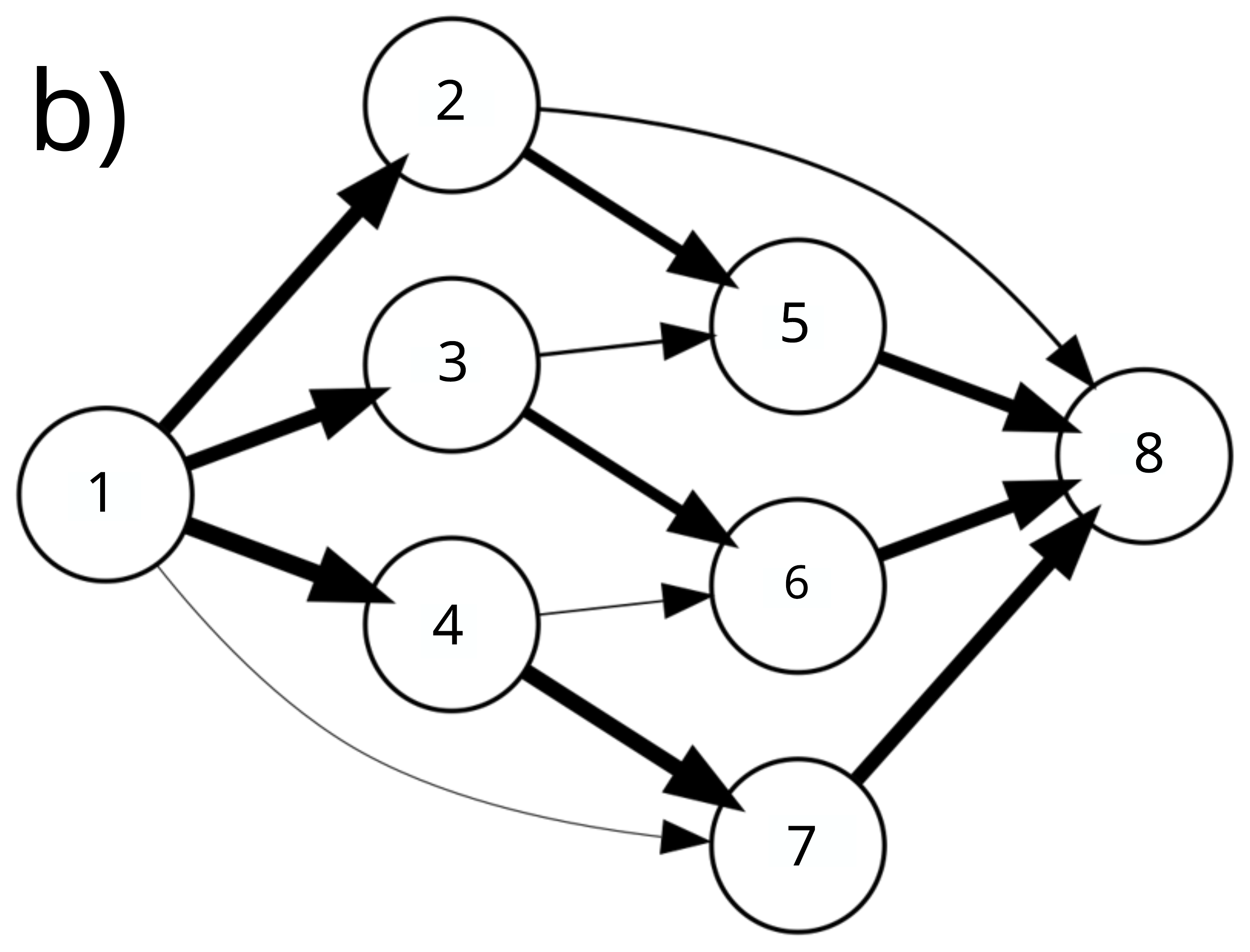}
	}
	\caption{\textbf{a)} Coordination Navigation: selfish agents (purple) seek to reach rewards (black) whilst minimising contact with each other. \textbf{b)} Routing Networks: agents split their own commodity flow between edges in a network over a sequence of time-steps. Starting at a source (node 1) and arriving at a target node (8), paths that have more commodity incur higher congestion costs.}
\vspace{-5 mm}\end{figure}

\textbf{Continuous Stochastic Potential Games}
\label{sec_cspgs}
% \yaodong{I feel this may not need a new section, an follow the previous one, even can merge SPG with state-based PG into one definiton. the detailed transformation process is not necessary to me}

We now introduce a new subset of SGs namely c-SPGs which is the framework of our approach. 
% The following defines a c-SPG: 
\begin{definition}
An SG is a c-SPG if for all states there exists a function $\phi:\mathcal{S}\times\boldsymbol{\mathcal{A}}\to\mathbb{R}$ such that the following holds for any 
$(a^i,a^{-i}),(a'^i,a^{-i})\in\boldsymbol{\mathcal{A}}$ where $a^{-i}_{t}:=(a^1,\ldots a^{i-1},a^{i+1},\ldots, a^N)$, $\forall i\in\mathcal{N}, \forall s\in \mathcal{S}$:
\begin{align}
\nonumber R_i(s,(a^i,a^{-i}))- R_i(s,(a'^i,a^{-i}))\\
=\phi(s,(a^i,a^{-i}))-\phi(s,(a'^i,a^{-i})). \label{potential_condition_static}
\end{align}
% $(\pi_i,\pi_{-i}),(\pi'_i,\pi_{-i})\in\boldsymbol{\Pi},\forall i\in\mathcal{N}, \forall s\in \mathcal{S}$:
% \begin{align}
% &\mathbb{E}_{\pi_i,\pi_{-i}}\left[R_i(s_t,(a^i_{t},a^{-i}_{t}))\right]- \mathbb{E}_{\pi'_i,\pi_{-i}}\left[R_i(s_t,(a'^i_{t},a^{-i}_{t}))\right]
% \\&=\mathbb{E}_{\pi_i,\pi_{-i}}\left[\phi(s_t,(a^i_{t},a^{-i}_{t}))\right]- \mathbb{E}_{\pi'_i,\pi_{-i}}\left[\phi(s_t,(a'^i_{t},a^{-i}_{t}))\right], \label{potential_condition_static}
% \end{align}
% where $a^i_{t}\sim \pi_i(\cdot|s_t),\; a'^i_{t}\sim \pi'_i(\cdot|s_t),\; a^{-i}_{t}:=(a^1_{t},\ldots a^{i-1}_{t},a^{i+1}_{t},\ldots, a^N_{t})$.
\end{definition}
Condition (\ref{potential_condition_static}) says that the difference in payoff from a deviation by one of the agents is exactly quantified by a global function $\phi$ that does not depend on the agent's identity. We call  $\phi$ the \textit{potential function} or potential for short. The condition extends the notion of static \textit{one-shot} potential games (PGs) \cite{monderer1996potential} to a continuous SG setting that now includes states and transition dynamics.

To complete the construction we introduce a condition which is a natural extension of PGs to state-based settings:
\begin{definition}
A stage game $\mathcal{M}(s)$ is state transitive if there exists a $ \phi:\mathcal{S}\times\boldsymbol{\mathcal{A}}\to\mathbb{R}$ s.th.
$\forall (a^i,a^{-i})\in\boldsymbol{\mathcal{A}}, \forall i\in\mathcal{N}, \forall s,s'\in \mathcal{S}$: \begin{align}
R_i(s,(a^i,a^{-i}))&- R_i(s',(a^i,a^{-i}))\nonumber
\\=\phi(s,(a^i,a^{-i}))&-\phi(s',(a^i,a^{-i})).
\end{align}
\end{definition}

 The intuition is that the difference in rewards for changing state is the same for each agent. 
%  The state transitivity condition naturally extends the notion of PGs to the SG setting which now includes states and transition dynamics.
% 
% 
%  Concretely, state transitivity asserts that the difference in reward for agent $i$ if it visits state $s^\star_i$ instead of $s^\star_i-\Delta$ is the same difference in reward received by agent $j$ if it visits state $s^\star_j$ instead of $s^\star_j-\Delta$. 
 Some classic examples of where state transitivity holds are anonymous games \cite{daskalakis2007computing}, symmetric SGs \cite{jaskiewicz2018symmetric}, team SGs \cite{cheng2017playing}. 
%  \yaodong{i feel too much explanation here is only making reader feel consufing. I think you can say in SPG, a natural extension of condition 1 is this, and leave all necessary explanation to condition 1 when it is introduced}
%  We give examples of SPGs in which state transitivity holds. 

% State transitivity is satisfied in numerous MAS e.g. multi-agent Coordination Navigation \cite{mylvaganam2017differential}, collaborative settings e.g. \cite{foerster2018counterfactual} and more generally, SGs with $R_i$ of the form $R_i(s,\boldsymbol{a})=g(s,\boldsymbol{a})+k(s)h_i(\boldsymbol{a})$ (c.f. Lemma \ref{trans_lemma}).\footnote{\textcolor{red}{This is  milder than conditions in \cite{marden2012state} which require $\phi$ to be non-decreasing along action invariant trajectories.}}

Our results are built under the  assumption\footnote{Statements of the technical assumptions are in the Appendix.} that state transitivity assumption holds.

Fig.~1 illustrates two examples of c-SPGs. For instance, in Coordination Navigation, one can verify that the state transitivity assumption is satisfied:
%\newline
%\textbf{Example: Coordination Navigation}\\
a collection of agents seeks to arrive at some  destination $x^\star\in\mathbb{R}^p$. Crucially, the agents must avoid colliding with other agents. Each agent's value function is given by:
\begin{align*}
V_i^{\boldsymbol{\pi}}&(\boldsymbol{x})=\frac{1}{2}\mathbb{E}\Bigg[\sum_{t=0}^\infty\gamma^t\Big\{K_i-\alpha\|x_{i,t}-x^\star\|^2
\\\textcolor{white}{xxx}&-\beta\sum_{j\in\mathcal{N}/\{i\}}(\|x_{i,t}-x_{j,t}\|^2+\epsilon)^{-1/2}-\|a_{i,t}^T-\rho\|_M^2\Big\}\Bigg]
\end{align*}
where $\|\cdot\|$ and $\|\cdot\|_M$ are Euclidean and Mahalanobis norms respectively, $x_{i,t}\in\mathbb{R}^p$ is the position of agent $i$ at time $t\in\mathbb{N}$ and $\boldsymbol{x}_t=(x_{1,t},\ldots,x_{N,t})\in\mathbb{R}^p$; $c,\rho, H, \alpha,\beta \{K_i\}_{i\in\mathbb{N}}$ are constants and $a_{i,t}$ is vector representing the action taken by  agent $i$. It can be readily verified that the game is potential with the following potential function: $
\phi^{\boldsymbol{\pi}}(\boldsymbol{x})=-\alpha\|x-x^\star\|^2
-\beta\sum_{j\in\mathcal{N}}\left(\|x-x_{j,t}\|^2+\epsilon\right)^{-1/2}+\beta c^{-2}-\|a^T-\rho\|^2$. Similarly, it can be readily verified that the game satisfies the state transitivity assumption.

% \textbf{Dynamic Demand Matching \cite{mguni2019coordinating}}\\
% Dynamic congestion games \cite{ahmad2010spectrum} model spectrum sharing in wireless communications, spatio-oligopolies e.g. competing firms locating their supply with varying demand and routing networks as shown in Fig. \ref{selfish_routing_network_fig}. An agent's expected return is $V_i(\boldsymbol{x}_t;m_{\boldsymbol{x}}) 
%     =\mathbb{E}
%     [K_i+
%         \sum_{t=0}^T
%  \Gamma(\boldsymbol{x}_t,m_{\boldsymbol{x}_t})-
%             \frac{1}{2} \|a_{i,t}\|_M^2
%     ],$
% where $\Gamma(\boldsymbol{x}_t,m_{\boldsymbol{x}_t})
% :=    
%     (   \boldsymbol{x}_t  - \tilde{\boldsymbol{x}}_t)^2 -    \alpha (m_{\boldsymbol{x}_t}
%     )^2$ and $m_{\boldsymbol{x}}$ is the agent density at $\boldsymbol{x}$. It can be similarly verified that the game is potential with the following potential function: $
% \phi^{\boldsymbol{\pi}}(\boldsymbol{x})= \Gamma(\boldsymbol{x}_t,m_{\boldsymbol{x}_t})-
%             \frac{1}{2} \|a_{t}\|_M^2$. Similarly, it can be readily verified that the game satisfies the state transitivity assumption.

In our Ablation experiments (see Sec. \ref{sec:experiments}) we show that our method is able to tackle settings in which the potentiality and state transitivity conditions are mildly violated. 

 C-SPGs also hold in SGs in which the agents have the same reward functions (identical interest games)\cite{monderer1996fictitious} such as anonymous games \cite{daskalakis2007computing}, team games \cite{wang2003reinforcement} and mean field games \cite{mguni2018decentralised}. Such SGs are widely used to study distributive systems and coordination problems in MAS such as Starcraft \cite{samvelyan2019starcraft} and Capture the Flag \cite{jaderberg2019human}.  MARL frameworks such as COMA \cite{foerster2018counterfactual}, QMIX \cite{rashid2018qmix} and QDPP \cite{yang2020multi} are fully cooperative settings and therefore fall within this category.  
% This broadens the modelling capabilities of PGs \cite{zazo2015dynamic,zazo2016dynamic}.
%\yaodong{the diffrence between the example you list for PGs, and the examples here for c-SPGs is not clear}\DM{Better?}

A key result we prove is that c-SPGs enjoy a dual representation as MDPs therefore enabling their solution to be computed by tackling MDPs. To construct a solution method for c-SPGs, we resolve a number of challenges: \textbf{i)} The first involves determining the dual MDP whose solution is to be learned through interaction with the environment. 
% The reward function of the dual MDP cannot be sampled directly and any approximation must generalise across states and actions. To this end, 
% To this end, we show that the MDP can be found by solving a differential equation using a sampling over the agents' policy spaces. 
 \textbf{ii)} The second involves developing a tractable learning procedure that ensures convergence to the game solution. To do this we develop a method that finds the solution of the dual MDP distributively, in doing so we also resolve the problem of combinatorial complexity that afflict MARL methods.  \textbf{iii)} The method of determining the dual MDP (\textbf{i)}) can incur small errors.  Our last challenge is to show that small errors in the construction of the dual MDP induce only small errors in the agents' best response actions. 
%  \yaodong{unclear what does this last point mean}.

\subsection{Link to Potential Games and Discussion}\label{sec:link_pgs}

We briefly continue the discussion on related works with a relevant review of PGs. The first systematic treatment of PGs appeared in \cite{monderer1996potential} in a static setting. PGs constitute a fundamental building block of general-sum games - any general-sum game can be decomposed into two (strategic) parts; PGs and \textit{harmonic games} \cite{candogan2011flows}. PGs model many real-world scenarios including traffic network scenarios, network resource allocation \cite{zazo2015dynamic}) social conflicts \cite{la2016potential} and consensus problems  \cite{marden2009cooperative}. PGs also  encompass all team games and some zero-sum games \cite{balduzzi2018mechanics}. 

% [Naive, relate to figure]

C-SPGs extend PGs to settings with dynamics and future uncertainty. This enables PGs to capture real-world scenarios that involve sequential decision-making and dynamics. Example of these settings traffic networks models, routing and packet delivery problems. 

% Static PGs have since been been applied to various problems within economics and engineering \cite{la2016potential,fabiani2018distributed}.

% 

% \DM{Will add statement on state-based PGs}.
% Moreover, with this our state-based PG construction does not impose the 
% % It is also used to prove that the potentiality property can be extended to the SG value functions. 
% monotonicity condition ($\phi(s,\cdot)\geq \phi(s',\cdot),\forall s\geq s'$) required in \cite{marden2012state}. \yaodong{the last sentence here is distracting, maybe put it into the related work. for people who do not know c-SPG that well, clarification only make them suspecisou of the contribution}

The analysis of dynamic PGs is extremely sparse and does not cover (reinforcement) learning settings in which the system is a priori unknown. 
% In GT, a small body of work covers dynamic PGs in known environments (with no learning). 
% Zazo et al. ~\cite{zazo2015dynamic} study deterministic dynamic PGs, here the agents' use polices that do not depend on the state (open-loop policies).
% Generalising to SGs is complicated by the fact that agents must execute policies that depend on the state.
In the direction of incorporating potentiality property within an SG, \cite{gonzalez2013discrete,macua2018learning} consider an SG in which the potentiality property is \textit{imposed} on the value functions which results in the need for highly restrictive assumptions. In \cite{gonzalez2013discrete} the SG is restricted to concave reward functions (in the state variable) and the transition function is required to be invertible (and known). These assumptions are generally incompatible with many MAS settings of interest.\footnote{The result in \cite{gonzalez2013discrete} also requires verifying the policy satisfies sufficiency conditions which is generally difficult given the size of the space of functions.} Similarly, \cite{macua2018learning} study a discrete Markov game in which the value function is \textit{assumed} to satisfy a PG property. Their construction requires that the agents' policies depend only on disjoint subcomponents of the state which prohibits non-local (strategic) interactions.  

Very recently \cite{zhang2021gradient} consider an SG setting in which all agents' value functions are assumed to satisfy a global PG property, that is, the incentive of all agents to change their \textit{policies} can now be expressed using a single global function. 
% As stated in \cite{zhang2021gradient}, without further qualification, this assumption is rather strong and difficult to verify except in the case in which all agents share the same objective in which case the analyses reduces to the familiar and well-investigated team game setting \cite{wang2003reinforcement,peng2017multiagent,foerster2018counterfactual}. 
% 
To construct this relationship using conditions on the stage game,  in a later work \cite{Leonardos2021GlobalCO} consider an SG setting and embed either of two properties into the game structure namely, an \textit{agent-independent transition assumption} (C.1) or an \textit{equality of individual dummy term assumption} (C.2). Using either of these conditions \textit{and} the stage game PG condition (Condition (\ref{potential_condition_static})), they show that the PG condition can be extrapolated to a global PG condition on the value functions. 

Conditions C.1. and C.2. in \cite{Leonardos2021GlobalCO} impose heavy restrictions since Condition C.1. reduces the SG to a linear combination of normal form games and removes all planning aspects (hence extrapolating the potentiality of stage games to the agents' value functions is deduced trivially). Condition C.2. restricts the noncooperative (strategic) interaction part of the game to a term that does not depend on the state or the agent's own action. Moreover imposing condition C.2. produces an SG  that is a special case of our SG (this can be seen using the equivalence expression in Lemma \ref{trans_lemma} (see Sec. \ref{sec:appendix_proof_results} in Appendix) by setting $k(s)\equiv 1$ and restricting $h_i$ to depend only on \textit{other agents'} actions in the reward functions of our SG). Therefore, the generalisation of the PG condition to SGs in \cite{Leonardos2021GlobalCO} requires strong limitations on the structure of the SG not present in our analysis.

With our new construction which has a PG at each stage game we show that the PG condition can be naturally extrapolated to the value functions of the SG. This provides verifiable assumptions on the game while imposing relatively weak assumptions on the SG in comparison to \cite{Leonardos2021GlobalCO}. With this we prove that the equilibrium of the SG can be found by merely solving an (unknown) MDP without imposing either state disjointness as in \cite{macua2018learning} or concavity as in \cite{gonzalez2013discrete}. 

\section{Planning in c-Stochastic Potential Games }\label{sec:planning}

We now show that the stable point solution of a c-SPG can be computed tractably by solving a \textit{dual MDP} with reward function $\phi$. This leads to a vast reduction in complexity for finding NEs in our c-SPG subclass of nonzero-sum. In what follows, we assume that the environment is known; in Sec. \ref{sec:learning} we extend the analysis of this section to unknown environments. We defer the proofs of the results of the following sections to the Appendix. 

As SGs are \textit{noncooperative settings}, the solution cannot be described as an optimisation of a single objective. The appropriate solution concept is the following NE variant  \cite{fudenberg1991tirole}: 

\begin{definition}
A strategy profile $\boldsymbol{\hat{\pi}}=(\hat{\pi}_i,\hat{\pi}_{-i})\in\boldsymbol{\Pi}$ is a {Markov perfect equilibrium (MPE)} if $\forall i\in\mathcal{N}$:

$
v_i^{(\hat{\pi}_i,\hat{\pi}_{-i})}(s)\geq v_i^{(\pi'_i,\hat{\pi}_{-i})}(s), \; \forall s\in \mathcal{S}, \;\forall \pi_i'\in\Pi_i$.
\end{definition}

The condition characterises a fixed point in strategies in which no agent can improve their expected payoff by unilaterally deviating from their current policy. We denote the set of NE of $\mathcal{G}$ by $NE\{\mathcal{G}\}$. Finding NE of nonzero-sum SGs in general involves using fixed point methods which are generally intractable \cite{chen2009settling}. Indeed, finding NE in SGs is PPAD complex  
(Polynomial Parity Arguments on Directed graphs) 
\cite{chen2009settling} for which brute force methods are intractable. Finding efficient solution methods for  nonzero-sum SGs is an open challenge \cite{shoham2008multiagent}.

We now show that c-SPGs exhibit special properties that enable their NE to be computed tractably. In particular, we show that computing the NE of c-SPGs can be achieved by solving an MDP. With this, solving c-SPGs can be approached with stochastic approximation tools. We then present a new Q-learning variant  that solves c-SPGs in polynomial time.

To begin, we construct the Bellman operator of $\mathcal{G}$. Let $g:\mathcal{S}\times\boldsymbol{\mathcal{A}}\to \mathbb{R}$ and $F:\mathcal{S}\to\mathbb{R}$, for any $s\in\mathcal{S}$ the Bellman operator of the game $\mathcal{G}$ is given by the following:
\begin{align*}
[T_g F](s):=\underset{\boldsymbol{a}\in\mathcal{A}}{\sup}\;[ g(s,\boldsymbol{a})+\gamma\int_{s'\in\mathcal{S}}ds'P(s';\boldsymbol{a},s)F[s']].
\end{align*}

We now state our first key result which reveals a striking property of the c-SPG class of games:
\begin{theorem}\label{reduction_theorem}
Let $V:\mathcal{S}\times\boldsymbol{\mathcal{A}}\to\mathbb{R}$ be a test function, then $\mathcal{G}$ possesses a fixed point NE in pure (deterministic) Markov strategies characterised by:
\begin{align*}\underset{k\to\infty}{\lim}T_\phi^kV^{\boldsymbol{\pi}}=\underset{{\boldsymbol{\hat{\pi}}}\in\boldsymbol{\Pi}}{\sup}V^{\boldsymbol{\hat{\pi}}},  
\end{align*}
%  and ${\boldsymbol{\pi^\star}}\in\underset{{\boldsymbol{\pi'}}\in\boldsymbol{\Pi}}{\arg\sup}\;F^{\boldsymbol{\pi'}}$.
where $\phi$ is the potential of $\mathcal{G}$.
\end{theorem}
The result states that the MPE of the game exist and in \textit{pure strategies} and correspond to solution of a (dual) MDP  $\mathscr{M}^\dagger:=\left\langle \phi,\times_{i\in\mathcal{N}}\mathcal{A}_i,P,\mathcal{S},\gamma \right\rangle$. In fact, it is shown that any MPE is a local optimum of the value function associated to $\mathscr{M}$. The value function of $\mathscr{M}$ which we call the \textit{dynamic potential function} (DPF), $B$, is constructed by $B^{\boldsymbol{\pi}}(s)=\mathbb{E}[\sum_{t=0}^\infty \gamma^t\phi(s_t,\boldsymbol{a}_t)|\boldsymbol{a}_t\sim\boldsymbol{\pi}]$, $\forall s\in\mathcal{S},\forall\boldsymbol{\pi}\in\boldsymbol{\Pi}$. 
% Hence, the MPE of the SG can be found by solving $\mathscr{M}$ instead of finding a fixed point solution of the game $\mathcal{G}$.  
% 

The theorem is proven inductively within a dynamic programming argument to extrapolate the potentiality property to the entire SG then showing $\mathcal{G}$ is \textit{continuous at infinity}. 

Theorem \ref{reduction_theorem} enables us to compute the MPE by solving an MDP, a task which can be performed in polynomial time.\footnote{The MDP lies is in a complexity class known as $\texttt{P-SPACE}$ which can be solved tractably \cite{papadimitriou1987complexity}.} 
% $\mathscr{M}=\langle \mathcal{S},\boldsymbol{\mathcal{A}},P,\phi,\gamma\rangle$. 
Moreover, Theorem \ref{reduction_theorem} enables a Q-learning approach \cite{bertsekas2012approximate} for finding the MPE of the game. The following fitted Q-learning method computes the approximate $B$ function and the corresponding optimal policy for each agent. 
% \yaodong{see my last point in the end of seciton 4, I think this value iteration had better be put after you introduce c-SPG, and then here, you can quote that, and say here is the FQI way to solve that in a model-free way. It is just maybe more natural to RL reviewers to see this in a logic flow}

First, let us define by
\begin{smalleralign}
Y_{l_k}(s_{l_k},\boldsymbol{a}_{l_k},s'_{l_k}):= \phi_{l_k,\hat{\rho}}(s_{l_k},\boldsymbol{a}_{l_k})+\gamma\underset{\boldsymbol{a'}}{\sup}\;\mathbb{E}_{\mathbb{P}}\left[\hat{B}_{l}\right](s'_{l_k},\boldsymbol{a'})\label{spot_q_target}
\end{smalleralign}
At each iteration $k=0,1,\ldots$ we solve the minimisation:
\begin{align}
\hspace{-5 mm}F_l\in \underset{\mathcal{F}\in\mathcal{H}}{\arg\inf} \sum_{l_k=1}^{n_k}\left(Y_{l_k}(s_{l_k},\boldsymbol{a}_{l_k},s'_{l_k})-\left[\mathcal{F}\right](s_{l_k},\boldsymbol{a}_{l_k})\right)^2
\label{critic_update}
\end{align}

 \begin{algorithm}[ht!]
%\vspace{-10pt}
\begin{algorithmic}[1]
\INPUT discount factor $\gamma$ and PF $\phi$.
    \FOR{$k \in {0,1,2,\cdots}$} 
        % 	\STATE Using $(\pi_i^k)_{i \in \mathcal{N}}$ to rollout, collect the trajectory data  and save it in the buffer  $\mathcal{D}$.
    	\FOR{$i \in \mathcal{N}$}
% \STATE  Solve Eq. \eqref{distrib_optimisation_obj} to find $(\hat{\rho},P_{\hat{\rho}})$.  \COMMENT{Approximate the potential function}
\STATE Set the local target $Y_{l_k}$ by \eqref{spot_q_target}
\STATE Update $F$ by minimizing Eq. \eqref{critic_update} 
\ENDFOR
    \ENDFOR
\OUTPUT $F, (\pi_i^k)_{i \in \mathcal{N}}$.
\end{algorithmic}
\caption{\textbf{SP}ot\textbf{Q}: \textbf{S}tochastic \textbf{POT}ential \textbf{Q}-Learning}
% \label{algo:main_algo1}
\end{algorithm}

The minimisation seeks to find the optimal action-value function $Q^\star$. Using this, we can construct our \textbf{SP}ot\textbf{Q} algorithm that works by mimicking value iteration. By Theorem \ref{reduction_theorem}, the algorithm converges to the MPE of the game.

Theorem \ref{reduction_theorem} does not establish uniqueness of $B$ which could lead to ambiguity in the solution. The following result reduces the set of candidates to a single family of functions:
\begin{lemma}\label{difference_constant_lemma}
If $B_1,B_2$ are value functions of the dual MDP $\mathcal{M}^\dagger$ then $
(B_1^{\boldsymbol{\pi}}-B_2^{\boldsymbol{\pi}})(s)=c,\forall \boldsymbol{\pi}\in \boldsymbol{\Pi}, \forall s \in \mathcal{S}$ where $c\in\mathbb{R}$
\end{lemma}
Therefore, the set of candidate functions are limited to a family of functions that differ only by a constant. 

\textbf{Computing the Potential Function $\phi$}

Theorem \ref{reduction_theorem} requires knowledge of  $\phi$. Existing methods to find $\phi$ in PGs e.g. MPD method \cite{candogan2013near} are \textit{combinatorial} in actions and agents. Indeed, directly applying \eqref{potential_condition_static} to compute $\phi$ requires checking all deviations over pure strategies (deterministic policies) which is expensive since it involves sweeping through the joint action space $\boldsymbol{\mathcal{A}}$. 
% Moreover, within RL settings, \textit{stochastic policies} are required in order to explore unknown environments.
% 
We now demonstrate how to compute $\phi$  while overcoming these issues by transforming \eqref{potential_condition_static} into a differential equation. To employ standard RL methods we require parameterised policies and, in anticipation of tackling an RL setting we extend our coverage to \textit{parameterised} stochastic policies.  
% Using the gradient theorem of vector calculus, it can be shown that the function $B$ is given by the following expression:
% \begin{lemma}\label{dynamic_potential_form}
% The function $B$ is given by the following expression for $s\in \mathcal{S}, \forall \pi_i\in\Pi_i,\forall \pi_{-i}\in\Pi_{-i}$:
% \begin{align}\nonumber
% v^{\boldsymbol{\pi}}(s)&=v^{\boldsymbol{\pi'}}(s)
% \\&+\mathbb{E}_{s_t\sim P}\left[\sum_{t=0}^\infty\sum_{i\in\mathcal{N}}\gamma^t\int_0^1\boldsymbol{\gamma}'(z)\frac{\partial R_i}{\partial \pi_i}(s_t,\boldsymbol{\gamma}(z))\right], \label{line_integral_dynamic_potential}
% \end{align}
% where $\boldsymbol{\gamma}(z)$ is a continuous differentiable path in $\boldsymbol{\Pi}$ connecting two strategy profiles $\boldsymbol{\pi}\in\boldsymbol{\Pi}$ and $\boldsymbol{\pi'}\in\boldsymbol{\Pi}$.
% \end{lemma}
% 

\begin{proposition}\label{reparam_pg_ode}
In any c-SPG the following result holds $\forall s\in \mathcal{S}, \forall\boldsymbol{a}\in \mathcal{A}$,  $ \forall(\boldsymbol{\eta}^i,\boldsymbol{\eta}^{-i}) \in \boldsymbol{E}^{ps}$:
\begin{align}
\hspace{-3 mm}\mathbb{E}_{\boldsymbol{\pi}(\boldsymbol{\eta})}\left[\frac{\partial\ln{\pi_{i,\boldsymbol{\eta}^i}}}{\partial \boldsymbol{\eta}^i}\left(\frac{\partial}{\partial a^{i}}+\frac{\partial}{\partial s}\right)(R_i- \phi)(s,\boldsymbol{a})\right]=\boldsymbol{0}. \label{potential_condition_pgt_parameterised}
\end{align}
% where $\boldsymbol{\eta}^j\in\mathbb{R}$ is the policy parameter for the agent $j\in\mathcal{N}$ policy $\pi_j\in\Pi_j$. 
\end{proposition}

The PDE serves as an analogue to the PG condition \eqref{potential_condition_static} which now exploits the continuity of the action space and the fact that they agents' actions are sampled from stochastic policies. Therefore Prop. \ref{reparam_pg_ode} reduces the problem of finding $\phi$ to solving a PDE. 
% The result exploits the continuity of the strategy set to convert \eqref{potential_condition_static} to a PDE.

So far we have considered a planning solution method that solves the game when the agents reward functions are known upfront. In Sec \ref{sec:learning}, we consider settings in which the reward functions and the transition function are a priori unknown but the agents observe their rewards with noisy feedback.

\section{Learning in c-Stochastic Potential Games} \label{sec:learning}
In RL, an agent learns to maximise its total rewards by repeated interaction with an unknown environment. The underlying problem is typically formalised as an MDP \cite{sutton2018reinforcement}.
% where $\mathcal{S}$ is a set of states, $\mathcal{A}$ is a set of actions, $P$ is a transition probability function describing the system dynamics, $R: \mathcal{S} \times \mathcal{A} \rightarrow \mathbb{R}$ is a reward function that measures the agent's performance and $\gamma \in [0, 1]$ quantifies how the agent's rewards are discounted over time \cite{sutton2018reinforcement}. \yaodong{explaining MDP here is unnecessary? because all the ingredients you need is already in SG}
% 
MARL extends RL to a multi-player setting \cite{yang2020overview}. The underlying problem is modelled as an SG in which the rewards of each agent and transition dynamics are a priori unknown. 

We have shown the MPE of a c-SPG can be computed by solving a \textit{Markov team game} $\mathscr{M}^\dagger$, an SG in which all agents share the same reward function $\phi$. We now discuss how to solve $\mathscr{M}^\dagger$ from observed data in unknown environments (i.e. if $P,\{R_i\}$ are not known). 
% The methodology involves showing that $\mathscr{M}^\dagger$ can be determined through supervised learning (SL) \cite{bishop2006pattern} on the agents' individual observed rewards. Thereafter, we demonstrate that $\mathscr{M}^\dagger$ can be solved in a distributed fashion. 
% 
% We then show how the agents can compute their MPE policy. 
% 
Additionally, we discuss our approach to enable easy scaling in the number of agents (and avoid combinatorial complexity) using distributive methods. 

The scheme can be summarised in the following steps:

\textbf{i)} Compute the potential estimate $\hat{\phi}$ by solving the PDE in Prop. \ref{reparam_pg_ode} using a distributed supervised learning method.

\textbf{ii)} Solve the team game $\hat{\mathscr{M}}^\dagger:=\langle \mathcal{N}, \mathcal{S},\boldsymbol{\mathcal{A}},P,\hat{\phi},\gamma\rangle$ with a distributed actor-critic method. The critic is updated with a distributed variant of the fitted Q-learning method in Sec. \ref{sec:planning}. 

% \yaodong{we need to say here why we need distribution to highlight your contribution on scalability}

\subsection{Learning the Potential Function $\phi$}\label{sec:compute-phi}

Though Prop. \ref{reparam_pg_ode} reveals that $\phi$ can be found by solving a PDE, it involves evaluations in pure strategies which can be costly. Moreover, the result cannot be applied directly to estimate $\phi$ since the agents sample their rewards but not $\phi$. 

We now show how each agent can construct an approximation of $\phi$ in a way that generalises across actions and states by sampling its rewards.  First, we demonstrate how the potential condition \eqref{potential_condition_static} can be  closely satisfied using \textit{almost pure} strategies. The usefulness of this will become apparent when we solve the PDE in Prop. \ref{reparam_pg_ode} to find $\phi$.  
% \yaodong{I feel the nascent delta function in definition 4 can be integrated with proposition 1, as this is one important component one would need to know the prop 1 result, break them is harming the flow. lemma 2 can be kept spearately to justify its existence} 
% 
% The following definition gives precise meaning to almost pure strategies:
% 
% \begin{definition}
% For any $\epsilon>0$, a function $\boldsymbol{\Gamma}_\epsilon$ is a nascent delta function if $\forall f$, $\underset{\epsilon \downarrow 0}{\lim}\int_{X}\boldsymbol{\Gamma}_\epsilon f(x)dx=f(0)$.
% \end{definition}
% 
%  
\begin{lemma}\label{bridge_lemma}
Let $F$ be a bounded and continuous function and let $\Delta F(s_t,(a^i_{t},a'^i_t),a^{-i}_{t}):= F(s_t,(a'^i_t,a^{-i}_{t}))- F(s_t,(a^i_{t},a^{-i}_{t}))$ 
% and $
% \Delta F(s_t,(\pi_i,\pi'_i),\pi_{-i})):=\mathbb{E}_{(\pi_i,\pi_{-i})}\left[F(s_t,(a^i_{t},a^{-i}_{t}))\right]- \mathbb{E}_{(\pi'_i,\pi_{-i})}\left[F(s_t,(a'^i_{t},a^{-i}_{t}))\right],$
then there exists $ c>0$ such that 
\begin{smalleralign}
\hspace{-2 mm}\left|\Delta F(s_t,(a^i_{t},a'^i_t),a^{-i}_{t})-\Delta F(s_t,(\pi_i,\pi'_i),\pi_{-i}))\right|\leq c\|F\|_\infty\hat{\sigma}_{\epsilon}^2\nonumber
\end{smalleralign}
where the policy $\pi_{i,\epsilon}$ is a \textit{nascent delta function}\footnote{
% For a function $f:\mathcal{F}\to\mathbb{R}^d$ a nascent function can be constructed by: $
% f_\epsilon(x)=\epsilon^df\left(\frac{x}{\epsilon}\right)
% $. 
A nascent delta function $g_{\epsilon}$ has the property $\underset{\epsilon \downarrow 0}{\lim}\int_{X}g_{\epsilon} f(x)dx=f(0)$ for any function $f$. They enable pure strategies to be approximated by stochastic policies with small variance. 
% E.g. for a Gaussian function $g_\epsilon(x):=(2\pi \epsilon)^{-1/2}e^{-x^2/(2\epsilon^2\sigma^2)}$ is nascent. 
We denote a nascent policy by $\pi_{\epsilon,i},\forall\epsilon>0$.} and $\hat{\sigma}^2_{\epsilon}:=\max\{\rm{Var}(\pi_i),\rm{Var}(\pi'_i)\}$.   
\end{lemma}
Since the bound approaches $0$ in the limit as policies become pure strategies, the potential condition \eqref{potential_condition_pgt_parameterised} is closely satisfied in nascent stochastic policies.

We now put Lemma \ref{bridge_lemma} to use with a method to compute $\hat{\phi}$ that inexpensively solves the PDE condition  \eqref{potential_condition_pgt_parameterised} over the policy parameter space $\boldsymbol{E}$. Indeed, thanks to Lemma \ref{bridge_lemma}, we can learn $\phi$ through an optimisation over $\boldsymbol{E}$.
The method uses a PDE solver over a set of randomly sampled points across $\boldsymbol{E}\times \mathcal{S}$ using the observed data  $\{(s_k,\boldsymbol{a}_t,(r_{1,k},\ldots,r_{N,k}))\}_{k\geq 0}$ where $r_{i,k}\sim R_i(s_k,\boldsymbol{a}_k)$.\footnote{As with methods with sharing networks (e.g. COMIX,
FacMADDPG \cite{de2020deep}), agents observe other agents' rewards. The method can be performed using only each agent's data $\{(s_k,\boldsymbol{a}_t,r_{i,k})\}$, however this requires more trajectory data.}  
% \yaodong{these equation need an entire line I think, put in text is not appropriate} 

Therefore, define by: 
\begin{smalleralign}
g^i(s,{\boldsymbol{\eta}},\hat{\phi}):=\nabla_{ \boldsymbol{\eta}^i} \ln{\pi_{\epsilon,i}(a^i|s;\boldsymbol{\eta})}\left(\frac{\partial}{\partial a^{i}}+\frac{\partial}{\partial s}\right) [R_i- \hat{\phi}](s,\boldsymbol{a})\nonumber
\end{smalleralign}
where $\boldsymbol{\pi}_{\epsilon}(\boldsymbol{a}|s;\boldsymbol{\eta}):=\pi_{\epsilon,i}(a^i|s;\boldsymbol{\eta}^i)\pi_{\epsilon,-i}(a^{-i}|s;\boldsymbol{\eta}^{-i})$. 

Following Prop. \ref{reparam_pg_ode} we consider the following problem to compute $\hat{\phi}$:
\begin{align}
\underset{\boldsymbol{\rho}\in\mathbb{R}^k}{\min}\|G(s,{\boldsymbol{\eta}};\hat{\phi}_{\boldsymbol{\rho}})    
\|^2_{\boldsymbol{E}\times \mathcal{S},\nu},\label{potential_estimation_obj}
\end{align}
% \\
% $\textcolor{white}{xxxxxxxxx}\underset{\boldsymbol{\rho}\in\mathbb{R}^k}{\min}\|G(s,{\boldsymbol{\eta}};\hat{\phi}_{\boldsymbol{\rho}})    
% \|^2_{\boldsymbol{E}\times \mathcal{S},\nu}$,
% 
where   $G(s,\boldsymbol{\eta};\hat{\phi}):=[g^1(s,\boldsymbol{\eta};\hat{\phi}),\ldots,
g^N(s,{\boldsymbol{\eta}};\hat{\phi})]^T$ and $\|f(y)\|^2:=\int_Y|f(y)|^2\nu(y)dy$ and $\nu(y)$ is a positive probability density on $Y$ and $\boldsymbol{\rho}\in\mathbb{R}^k$ are parameters.
% \footnote{In the deterministic case i.e. $a^i=h(s,\eta^i)$, $g^i$ takes the form: $g^i(s,{\boldsymbol{\eta}},\hat{\phi})=\left(\frac{\partial}{\partial a^{i}}+\frac{\partial}{\partial s}\right) [R_i(s,\boldsymbol{a})- \hat{\phi}(s,\boldsymbol{a})\frac{dh}{d\eta^i}]$.}
The optimisation performs evaluations in mixed strategies which is computationally inexpensive. Using a weighted exponential sum method \cite{CHANG2015325}, the objective reduces to a least squares problem on a single objective $G(s,\boldsymbol{\eta},\boldsymbol{\rho}):=N^{-1}\sum_{i\in\mathcal{N}}\left(g^i(s,\boldsymbol{\eta},\hat{\phi}_{\boldsymbol{\rho}^i})\right)^2 $. The optimisation can be solved with a function approximator on $\hat{\phi}_{\boldsymbol{\rho}}$ e.g. a deep neural network (NN). 
Under mild conditions \cite{bertsekas2000gradient} the method converges to a critical point of $G$, that is $
\underset{n\to\infty}{\lim}\nabla_{\boldsymbol{\rho}} G=0$.
% \footnote{For NNs convergence may only be to a local minimum. This is true for nonconvex optimisation in general.} 
We defer the details of the method to the the Appendix.

% The algorithm works by each agent solving the current stochastic team game approximation i.e. the solution to each RL problem is computed in a distributed fashion and at each step the estimates of $\phi$ calculated by each agent are shared in a consensus step.
% 

% Having acquired the elements of $\hat{\mathscr{M}}^\dagger=\langle \hat{\phi},\times_{i\in\mathcal{N}}\mathcal{A}_i,P,\mathcal{S},\gamma\rangle$, we are now in position to describe our Q-learning style algorithm \cite{bertsekas2012approximate} which computes the approximate $B$ function and the corresponding optimal policy for each agent. \yaodong{see my last point in the end of seciton 4, I think this value iteration had better be put after you introduce c-SPG, and then here, you can quote that, and say here is the FQI way to solve that in a model-free way. It is just maybe more natural to RL reviewers to see this in a logic flow} To this end, let us define by 
% $
% Y_{l_k}(s_{l_k},\boldsymbol{a}_{l_k},s'_{l_k}):= \hat{\phi}_{l_k,\hat{\rho}}(s_{l_k},\boldsymbol{a}_{l_k})+\gamma\underset{\boldsymbol{a'}}{\sup}\left[\hat{B}_{l}\right](s'_{l_k},\boldsymbol{a'})$

% At each iteration $k=0,1,\ldots$ we solve the minimisation: 
% \begin{align}
% \hspace{-5 mm}F_l\in \underset{\mathcal{F}\in\mathcal{H}}{\arg\inf} \sum_{l_k=1}^{n_k}\left(Y_{l_k}(s_{l_k},\boldsymbol{a}_{l_k},s'_{l_k})-\left[\mathcal{F}\right](s_{l_k},\boldsymbol{a}_{l_k})\right)^2
% \label{critic_update}
% \end{align}

\subsection*{Actor-Critic Method}
We now return to tackling the problem of solving the team game $\hat{\mathscr{M}}^\dagger:=\langle \mathcal{N}, \mathcal{S},\boldsymbol{\mathcal{A}},P,\hat{\phi},\gamma\rangle$. To enable the method to scale and handle continuous actions, we adapt the fitted Q-learning method in Sec. \ref{sec:planning} to an actor-critic method \cite{konda2000actor} for which each agent learns its own policy $\pi_{i,\eta_i}$ using the estimate of $B$. 
% where $\eta_i$ is the agent's policy parameter 
% \textcolor{red}{Note that the approximation of $\phi$ is identical for all agents following the consensus step.}
The policy parameter $\eta_i$  of the policy at the $k^{th}$ iteration is updated through sampled deterministic policy gradients (DPGs)  \cite{silver2014deterministic}:
\begin{align}\nonumber
&\nabla \eta_i \hat{B}^{(\pi^k_{i,\eta_i},\pi^k_{-i,\eta_{-i}})}\left(s_{l_k}\right)   
\\&\hspace{-4 mm}\approx \frac{1}{L} \sum_{l=1}^L \nabla_{\eta_i} \pi_{i,\eta_i}\left(\cdot|s_{l_k}\right) \nabla_{a^{i}_{l_k}} F_{k}\left(s_{l_k}, \boldsymbol{a}_{l_k}\right)\Big|_{a^{i}_{l_k}\sim\pi^k_{i,\eta_i}}
    \label{actor_update}
\end{align}
% \yaodong{pls check the new pseducode and the notation in this paragraph}\DM{updated some stuff. the indices still look like they need to be checked (the subscrpt $l_k$ and $l_t$... also in the algo what's the $k$ superscript on the policy?}\yaodong{k is the most outer loop, the optimal policy at the k-th step}

% This process \textit{solves} the sequence of stage games. 
Equation \eqref{actor_update} describes the actor update via a DPG. The complete process is described in Algorithm 1. 
% \yaodong{I would defetniely try putting the whole algo back here. since it is the key contribution. and also, the explaination you make here would make read the pseudocode in appexidix necessary, which break the protocol of using appendix, it could get complaint}
It
 involves two optimisations in sequence: the agents individually compute the approximation $\hat{\phi}$ which is then used for computing $\hat{B}$, which approximates the optimal value function $B$ by a Q-learning + decentralised DPG method and outputs each agent's MPE policy. Crucially the method avoids optimisations over the joint space $\times_{i\in\mathcal{N}}\mathcal{A}_i$  enabling easy scaling (in the number of agents) in this component of the algorithm.
 
%  and yields a method that scales well with the number of agents.

% The above process is performed distributively; each agent independently computes $\hat{\phi}$ using the data $\{(s_k,\boldsymbol{a}_t,(r_{1,k},\ldots,r_{N,k}))\}_{k\geq 0}$.

 \begin{algorithm}[ht!]
%\vspace{-10pt}
\begin{algorithmic}[1]
\INPUT discount factor $\gamma$, DPF and PF approximation maps $F,P_{\boldsymbol{\rho}}\in\mathcal{H}$ (resp.) ($\boldsymbol{\rho}\in\mathbb{R}^k$).
    \FOR{$k \in {0,1,2,\cdots}$} 
    	\STATE Using $(\pi_i^k)_{i \in \mathcal{N}}$ to rollout, collect the trajectory data  and save it in the buffer  $\mathcal{D}$. 
    	\FOR{$i \in \mathcal{N}$}
\STATE{\textbf{// Learn the potential function $\hat{\phi}$}}
    	    	\STATE Sample a random minibatch of $L$ samples $\{(s_{l_t},(a^i_{l_t})_{i \in \mathcal{N}},s_{l_{t+1}},(r^i_{l_t})_{i \in \mathcal{N}}\}$ from $\mathcal{D}$. 
\STATE  Compute $\hat{\phi}$ by solving Expression \eqref{potential_estimation_obj}
% \STATE  Solve Eq. \eqref{distrib_optimisation_obj} to find $(\hat{\rho},P_{\hat{\rho}})$.  \COMMENT{Approximate the potential function}
\STATE{\textbf{// Compute the value function $\hat{B}$}}
\STATE Set the local target $Y_{l_k}$ by \eqref{spot_q_target}
\STATE Update the shared critic $F$ by minimizing Eq. \eqref{critic_update}
\STATE{\textbf{// Learn the individual policy}}
\STATE Update the policy by minimizing Eq. \eqref{actor_update}
    	\ENDFOR
    \ENDFOR
\OUTPUT $F, (\pi_i^k)_{i \in \mathcal{N}}$.
\end{algorithmic}
\label{algo:main_algo1}
\caption{$\ourmethod$: \textbf{S}tochastic \textbf{POT}ential \textbf{A}ctor-\textbf{C}ritic}
\end{algorithm}

\textbf{Scaling in $N$ using Consensus Optimisation}\\
% \yaodong{this paragraph be put in the next subsection ?}
Although the above method represents progress for solving SGs, a scalability issue remains since estimating $\hat{\phi}$ involves a computation over the joint space $\boldsymbol{E}$. This becomes increasingly expensive with large numbers of agents. We now devise a fully distributed version of the method that scales with the number of agents. In this version, each agent $i$ constructs an independent estimate of $\hat{\phi}$ by sampling across $E_i\times\mathcal{S}$ at each step using only its own observed data $\{(s_k,\boldsymbol{a}_t,r_{i,k})\}_{k\geq 0}$. The method includes a consensus step that enables $\hat{\phi}$ (and hence $\hat{B})$ to be accurately computed efficiently in a fully distributed fashion \cite{tutunov2019distributed}.

To enable efficient scaling with the number of agents, we use distributed optimisation (DO) with consensus  \cite{nedic2009distributed} to find $\hat{\phi}$. Each agent produces its own estimate $\hat{B}$ based on its observed rewards. DO methods efficiently solve large scale optimisation problems \cite{macua2010consensus} and yields two major benefits:\newline\textbf{i) efficiency:} computing $\phi$ uses feedback from \textit{all} agents' reward samples.\\ \textbf{ii) consensus on Q:} agents learn $\phi$ distributively but have identical Q iterates (for computing $\hat{B}$). 

The common objective which each agent solves individually, is expressed with a set of local variables $\{\rho_i\}_{i\in\mathcal{N}}$ and a common global variable $z$:

${\rm minimise }\;\; \bar{G}(s,\boldsymbol{\rho})=N^{-1}\sum_{i\in\mathcal{N}}(g^i(s,\boldsymbol{\rho}^i))^2
$

s.t. $
\boldsymbol{\rho}^i-z=0, \quad i=1,\ldots, N,$\\
where the gradient descent is according to: $
\boldsymbol{\rho}^i_n=\boldsymbol{\rho}^i_{n-1}-\alpha\nabla g^i(s,\boldsymbol{\rho}^i)$ for some step size $\alpha>0$.
% 
% The agents' collective consensus objective is therefore given by:
% \begin{align}
% &\underset{\{\rho_i\}_{i\in\mathcal{N}},z}{\min}\int_{\boldsymbol{\zeta}\in \boldsymbol{C}} \boldsymbol{c}(\boldsymbol{a}|s;d\boldsymbol{\zeta})\left(\frac{\partial}{\partial \zeta^i}\ln{[\pi_{\epsilon,i}(a^i|s;\zeta^i)]}\left(R_i(s,\boldsymbol{a})-[P_{\rho_i}](s,\boldsymbol{a})\right)\right),
% \end{align}
% subject to
% \begin{align}
% \rho_1=\ldots\rho_N=z.
% \end{align}
% 
Note that the  constraint prevents convergence to any $R_i$.
% and the policy derivative alters the size of the contribution from different agents.  

The algorithm works by constructing an estimate $\hat{\phi}$ then solving $\mathscr{M}^\dagger$ in a \textit{distributed fashion} allowing the method to scale with the number of agents. 

% \yaodong{there are many shattered pieces introduced thus far, I think you need a summary paragraph here, maybe each to the pseducode, and link each subsection to the snippet in the pesducode. people can see immediatelty what each subsection is truing to do}
\textbf{Algorithm Analysis}

% \section{Algorithm Analysis}\label{sec:analysis_of_spot-q}
 Our $\ourmethod$ algorithm inherits many useful properties of Q-learning \cite{antos2008fitted}.\footnote{By Prop. 5 (see Appendix) any MPE is a \textit{local} optimum of $B$.} 
%  Integral to $\ourmethod$ is the team MDP $\mathscr{M}$ which is solved distributively. 
%  Team Markov games are well-known to exhibit strong convergence properties using various MARL methods \cite{littman2001value,wang2003reinforcement}. \yaodong{i think it is uninformative to mention team game here, just highligh the contribution on the error analysis of the learned Phi and its consequence on the learning algorithm.} 
 Nevertheless, it is necessary to ensure the output of the algorithm still yields good performance when the supervised learning approximation of $\mathcal{M}^\dagger$ has small errors. We now analyse the $\ourmethod$ algorithm and show that provided errors in approximating $\mathcal{M}^\dagger$ are small the error in the algorithm output is also small. 
 
 Our first result bounds the error on the estimate for the DPF from using the approximation method for $\hat{\phi}$.  

\begin{proposition}\label{SL_convergence}
Define by the following $
F_i(s,\boldsymbol{\eta},\rho):=\int_{\boldsymbol{\mathcal{A}}}\boldsymbol{\pi}_\epsilon(\boldsymbol{da},\boldsymbol{\eta},s)\frac{\partial}{\partial \eta_i}\pi_{i,\epsilon}(a^i,\eta_i,s)\nabla F_\rho(s,\boldsymbol{a})$ and $
U(s,\boldsymbol{\eta},\rho):=\int_{\boldsymbol{\mathcal{A}}}\boldsymbol{\pi}_\epsilon(\boldsymbol{da},\boldsymbol{\eta},s)\frac{\partial}{\partial \eta_i}\pi_{i,\epsilon}(a^i,\eta_i,s)\nabla R_i(s,\boldsymbol{a})$ then the following bound holds for some $c>0$:
% where we have used the shorthand: $\boldsymbol{\pi}_\epsilon(\boldsymbol{a},\boldsymbol{\eta},s):=\pi_i(a,s,\eta_i)\pi_{-i}(a_i,s,\eta_{-i})$
% 
\begin{align*}
\sum_{i\in\mathcal{N}}\left\|F_i(s,\boldsymbol{\eta},\rho)-U(s,\boldsymbol{\eta})\right\|\leq cN^2\epsilon^2,
\end{align*}
where $\epsilon$ is the approximation error from the SL procedure.
\end{proposition}

Our next result ensures that if the estimates of $\phi$ have only small errors, $\ourmethod$  generates policy performances that closely match that of the MPE policies.  

% The intuition is that changes in $\phi$ caused by policy deviations match changes to the values of $R_i$. 

\begin{proposition}\label{Lipschitz_potential_itself}
Define by $B_\epsilon^{\boldsymbol{\tilde{\pi}}}=\underset{k\to\infty}{\lim}T_{\phi_\epsilon}^kB^{\boldsymbol{\pi}}$ and let the policy  $\boldsymbol{\hat{\pi}}$ be an MPE strategy i.e. $\boldsymbol{\hat{\pi}}\in NE\{\mathcal{G}\}$  (so that  $\underset{k\to\infty}{\lim}T_\phi^kB^{\boldsymbol{\pi}}=B^{\boldsymbol{\hat{\pi}}}$) then for any $\epsilon>0$ the following holds:
\begin{align*}
    \| B^{\boldsymbol{\hat{\pi}}}-B_\epsilon^{\boldsymbol{\tilde{\pi}}}\|
\leq (2-\gamma)(1-\gamma)^{-1}\epsilon,
\end{align*}
whenever $\|\phi^\epsilon- \phi\|<\epsilon$.  
\end{proposition}

The result ensures that given close approximations of $\phi$ $\ourmethod$ in turn yields outputs close to $B^\star$. The result exploits the fact that the dual MDP $\mathcal{M}$ of Theorem \ref{reduction_theorem}  exhibits a continuity property so that small errors in the approximation of $\phi$ and $B$ incur only small changes to the MPE of $\mathcal{G}$.

\section{Experiments}  \label{sec:experiments}
We evaluate $\ourmethod$ in three popular multi-agent environments: the particle world \cite{lowe2017multi}, a network routing game \cite{roughgarden2007routing} and a Cournot duopoly problem \cite{agliari2016nonlinear}. 
These environments have continuous action and state spaces, and the agents seek to maximise their own interest e.g. reaching target without collisions on particle world and minimising the cost for transporting commodity on routing game. 
To solve these problems successfully, the agents must learn Markov perfect equilibrium policies in order to respond optimally to the actions of others.

We consider two groups of state-of-the-art MARL baselines that handle continuous actions. The first group use individual rewards for learning: MADDPG \cite{lowe2017multi} and DDPG \cite{lillicrap2015continuous}. The second group use the collective rewards of all agents: COMIX and COVDN \cite{de2020deep}. Further details are in the Appendix.

We use two evaluation metrics: \\\textbf{Exploitability} \cite{davis2014using} describes how much additional payoff players can achieve by playing a best-response policy $\operatorname{BR}$ (defined in the Appendix). It measures the proximity of the agents' policies to the MPE strategy, defined as $\delta = \frac{1}{N}\sum_i (u_i(\pi^{-i}, \operatorname{BR}(\pi^{-i}))- u_i(\pi))$. \\
\textbf{Social welfare} is the collective reward of all agents: $r=\sum_{i=1}^N r_i$, this is most relevant in tasks such as team games, where agents seek to maximise total reward.
\subsubsection*{Ablation Studies}
To test the robustness of $\ourmethod$ and the baselines, we perform a set of ablation studies within routing games. \newline\textbf{Ablation 1} analyses $\ourmethod$ in SGs that progressively deviate from c-SPGs, showing that $\ourmethod$ can handle SGs that mildly violate the c-SPG conditions (i.e. the potentiality requirement).\newline
\textbf{Ablation 2} analyses $\ourmethod$ in SGs that progressively deviate from team games but retain the potential game property. We demonstrate that, unlike other methods, $\ourmethod$ is able to converge to the Markov Perfect Equilibrium in non-cooperative SGs.
We also report results on the classic Cournot Duopoly and show convergence of $\ourmethod$ to NE. 

\textbf{Non-atomic Routing Games}
involve a set of $N$ selfish agents seeking to transport their commodity from a source node to a goal node in a network. This commodity can be divided arbitrarily and sent between nodes along edges.
 
At each time step, each agent has a distribution of commodity over the nodes of the graph. It assigns a fraction of its commodity in each node to travel along the edges emerging from those nodes. There are multiple agents (given by $N\in\{2,4,6,8\}$), using the same network (number of nodes $K\in\{20,40\}$) and agents pay a cost related to the total congestion of every edge at each time step. We design the game so that the MPE is socially efficient, i.e. playing an MPE strategy leads to high individual returns.  We repeat the experiments for 5 independent runs and report the mean and standard deviation of the rewards. Further details on the settings can be found in Appendix. 

\subsubsection*{Results}
\textbf{Exploitability:}
We test $\ourmethod$ in a simple Braess' paradox game. The exploitability of $\ourmethod$ (Fig. \ref{fig:exploit}) quickly converges to close to $0$, indicating it learns NE policies (negative values are due to the fact that we are approximating best-responses). In contrast, the high exploitability values of existing MARL methods indicate that they fail to converge to NE policies. The algorithms that involve reward sharing (COMIX, COVDN) attempt to maximise social welfare, which is incompatible with this non-cooperative setting, so can be exploited by a best-response strategy.

\begin{figure}[t!]
	\centering
	\subfigure[ 
	$\ourmethod$'s policy on Braess' paradox 
	]{
		 \includegraphics[width=0.7\columnwidth]{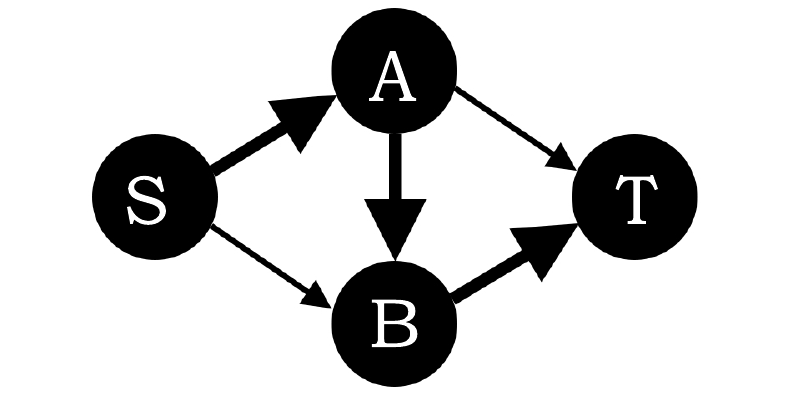}
	}
	\subfigure[Exploitability of different methods.]{
		 \includegraphics[width=0.8\columnwidth]{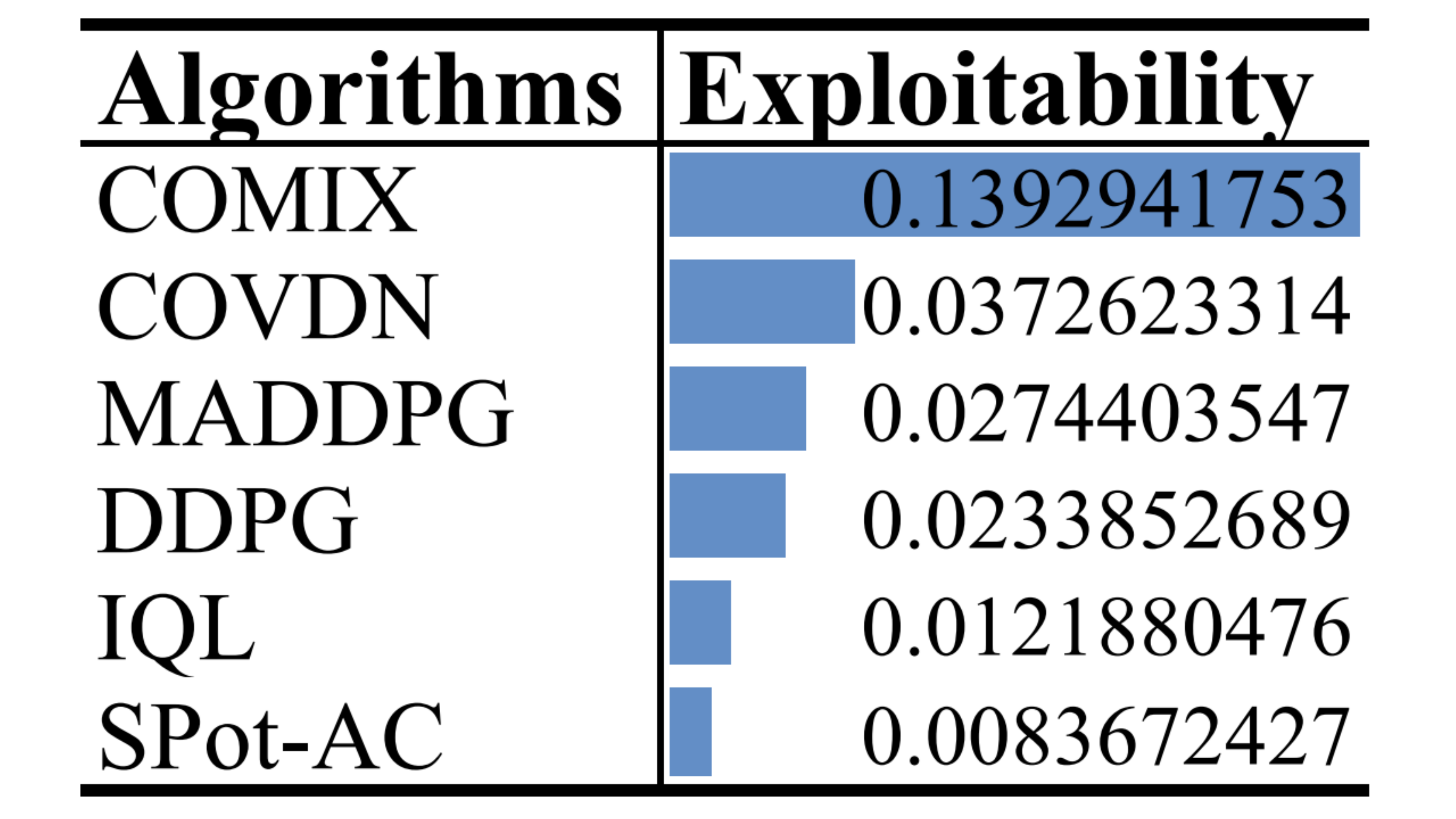}
	}
	\vspace{-5pt}
	\caption{Results of exploitability. (a) a visualization of learned policy flows by $\ourmethod$. (b) Exploitability results of all methods. }
	\label{fig:exploit}\vspace{-5pt}
\end{figure}

\textbf{Social welfare:} In the cooperative, non-atomic routing game environment, we see in Fig. \ref{fig:selfish-routing-res-main}, using $\ourmethod$ (orange), each agent learns how to split their commodity $G$ in a way that maximises rewards (minimises costs) and matches the shared reward baselines. Conversely, MADDPG (orange) and DDPG (blue) yield low rewards with high variance.

\begin{figure}[h]
\begin{center}
    \includegraphics[width=8.25cm, height=8cm]{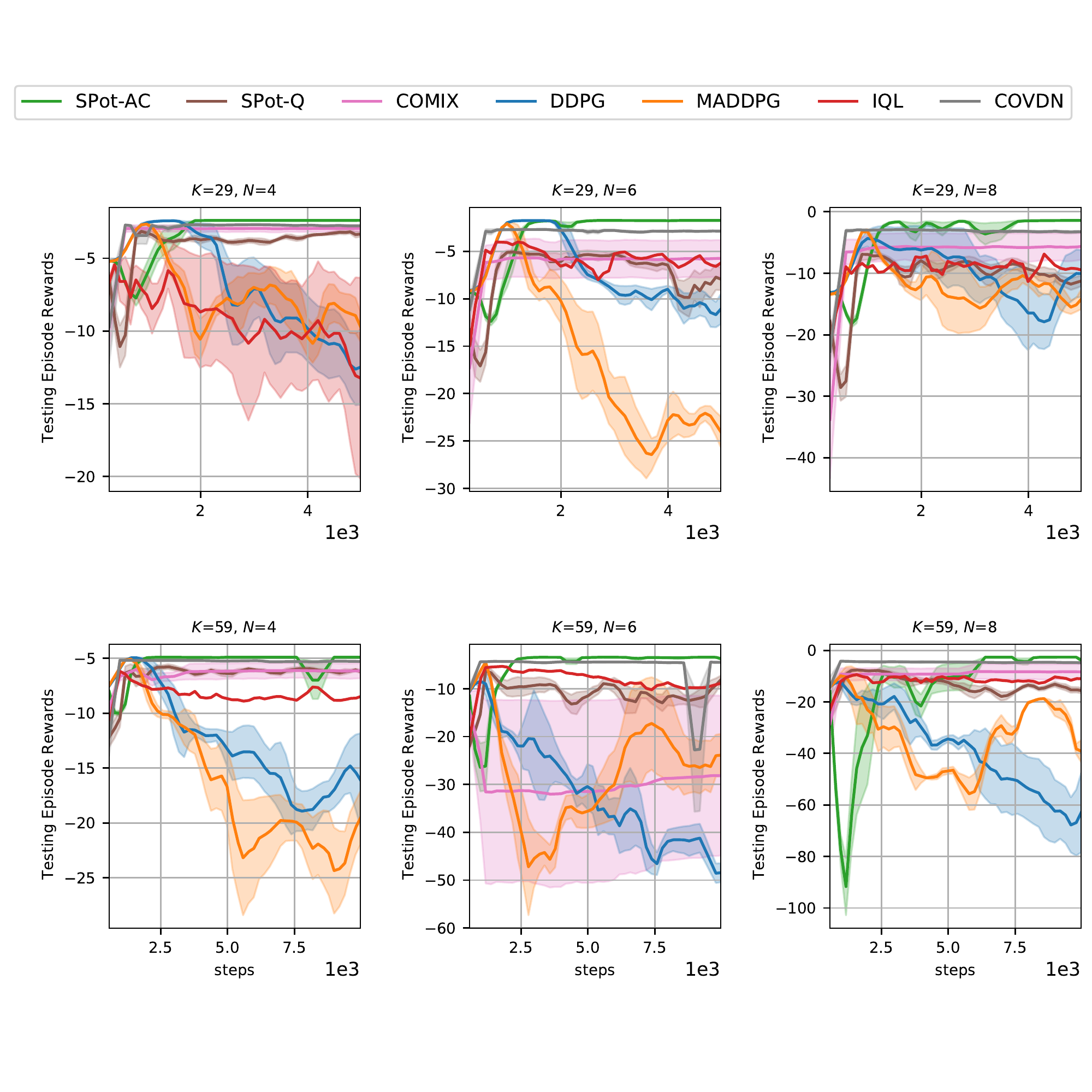}
\end{center}
\vspace{-25pt}
   \caption{\textbf{Top} Average agent returns in Network with $N=4, 6, 8$ agents, $K=20$ nodes. \textbf{Bottom} Average agent returns in Network with $N=4, 6, 8$ agents, $K=59$ nodes. 
   }
\label{fig:selfish-routing-res-main}\vspace{-2 mm}
\end{figure}

\textbf{Coordination Navigation}
An OpenAI Multi Agent Particle Environment task \cite{lowe2017multi} involves $n$ agents and $n$ landmarks.
Each agent must reach the target while avoiding collisions with other agents and fixed landmarks. Agents can observe the relative positions of other agents and landmarks,
and have five actions \{up, down, left, right, stay\}. 
The reward is calculated as the agent's distance to each landmark with penalties for collisions with other agents.

This is a non-cooperative SG, so we compare $\ourmethod$ to DDPG and MADDPG, algorithms that are able to learn policies in which agents can act selfishly. We perform the exploitability analysis as above. Fig. \ref{fig:spread_3agent} shows $\ourmethod$ achieves the best performance in terms of minimum distance to target \textit{and} number of collisions, demonstrating that $\ourmethod$ enables agents to learn to coordinate while pursuing their own goals.

\begin{figure}[h]
\begin{center}\vspace{-2 mm}
   \hspace{-5 mm}\includegraphics[width=0.35\textwidth]{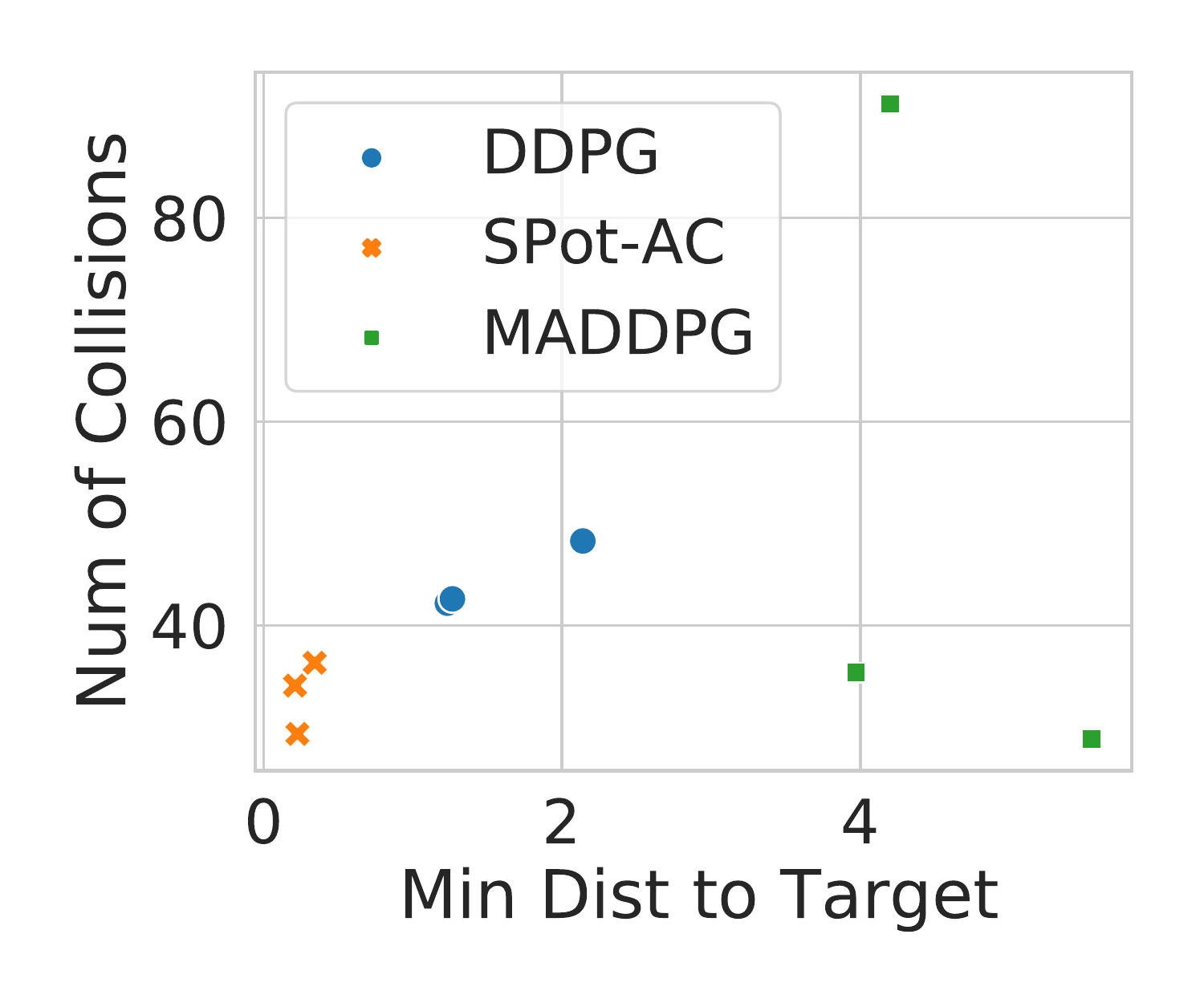}
\end{center}\vspace{-8 mm}
   \caption{Coordination Navigation problem.%   
   }
\label{fig:spread_3agent}
\vspace{-9pt}
\end{figure}
 
\section{Conclusion}
\vspace{-.25 mm}
In this paper, we describe the first MARL framework that tackles MAS with payoff structures beyond zero-sum or team games. In doing so, the results we establish pave the way for a new generation of solvers that are able to tackle classes of SGs beyond cases in which the payoff structures lie at extremes. Therefore, the results of this paper open the door for MARL techniques to address a wider range of multi-agent scenarios.  By developing theory that shows a class of SGs, namely c-SPGs have a dual representation as MDPs, we showed that c-SPGs can be solved by MARL agents using a novel distributed method which avoids the combinatorial explosion therefore allowing the solver to scale with the number of agents. We then validated our theory in experiments in previously unsolvable scenarios showing our method successfully learns MPE policies in contrast to existing MARL methods. 

% \section{Acknowledgement}
% \vspace{-.25 mm}

\section*{Acknowledgements}
YW and ZW are partly supported by the Strategic Priority Research Program of Chinese Academy of Sciences, Grant No.XDA27000000. DM is grateful to Mohammed Amin Abdullah for insightful discussions relating to the convergence proofs in this work.

\clearpage
\balance
\bibliography{main}
\bibliographystyle{icml2021}

% \appendix 

% \onecolumn 
% \icmltitle{ Supplementary Material for \\ 
% Learning in Nonzero-Sum Stochastic Games with Potentials }
% \clearpage
% \appendix

% \resetlinenumber
% \setcounter{page}{1}

% \addcontentsline{toc}{section}{Appendix} % Add the appendix text to the document TOC
% \part{{\Large{Appendix}}} % Start the appendix part
% \parttoc

%\yaodong{i put a draft here. the strcture of the appendix need to be tuned. section/subsection, and also we need the algo 1}

% \clearpage

% {sec:appendix_zero_order}

% \clearpage
% \Huge
% \begin{center}
% \textbf{Supplementary Material}    
% \end{center}
% \normalsize

\appendix 

\onecolumn 
\icmltitle{ Supplementary Material for \\ 
Learning in Nonzero-Sum Stochastic Games with Potentials }

The Supplementary material is arranged as follows: first, in Sec. \ref{sec:appendix_Experiment_ Details} we give a description of the experimental details and report the hyperparameter values used in our experiments. In Sec. \ref{sec:ablation_experiments}, we give a detailed discussion of our ablation studies. In Sec. \ref{sec:cournot_duopoly}, we give results of our analysis on a static noncooperative game namely Cournot duopoly problem and in Sec. \ref{sec:appendix:cournot_analytic} we perform an study of the problem and verify our solution analytically. In Sec. \ref{sec:algorithm_2}, we give additional details on our supervised learning method to compute the potential function. In Sec. \ref{sec:Consensus Optimisation}, we give additional details on our distributed learning method using consensus optimisation required to compute the potential function distributively.  In Sec. \ref{sec:append_assumptions}, we outline some of the additional notation and detail the technical assumptions used in the proofs of our results which are contained in Sec. \ref{sec:appendix_proof_results}  which concludes the supplementary material.

\section{Experiment Details \& Hyperparameter Settings}
\label{sec:appendix_Experiment_ Details}

The settings for all methods are the same, except for the stated cases that use a shared critic. The optimiser is set to Adam for all methods reported. The learning rates for actors and critics are 1e-4 and 1e-3 respectively.
Both actors and critics consist of four fully connected layers with dimensions of [64,64,64,$n_{act}$].

In the table below we report all hyperparameters used in our experiments. Hyperparameter values in square brackets indicate range of values that were used for performance tuning.

\begin{center}
    \begin{tabular}{c|c} 
        \toprule
        \textbf{Setting} & \textbf{Value} \\\hline
        Clip Gradient Norm & 1\\
        Discount factor $\gamma_{E}$ & 0.99\\
        $\lambda$ & 0.95\\
        Learning rate & $1$x$10^{-4}$ for actor and $1$x$10^{-3}$  for critic \\
        % Number of minibatches & 256\\
        % Number of optimization epochs & 4\\
        Batch size & 256\\
        Buffer size & 4096\\
        Policy architecture & MLP \\
        Number of parallel actors & 1\\
        Optimization algorithm & Adam\\
        Rollout length & 1000*[10, 20]\\
        \bottomrule
    \end{tabular}
\end{center}

\section{Ablation Studies}\label{sec:ablation_experiments}
Our method allows MARL agents to solve noncooperative SGs within the SPG subclass. In this section, we analyse the behaviour of our method compared against existing MARL baselines in scenarios that range from team (cooperative) SG settings to noncooperative games outside of SPGs. In doing so, we examine their performance when the SPG assumptions are violated and show that $\ourmethod$ is still able to perform well when the PG condition (Equation \eqref{potential_condition_static}) is mildly violated.  Additionally, in these settings we also compare the performance of $\ourmethod$ in cooperative settings which are the degenerate case of SPGs.

As in Section 6 (within the main body), we consider a stochastic network routing game which has both continuous action and state spaces and is a nonzero sum game (neither team-based nor zero-sum) which represents
a challenge for current MARL methods. As in the Network routing games considered in Section 6, we restrict our attention to networks that have efficient NE. In such network structures, playing an NE (best-response) strategy leads to a higher total return for the (self-interested) agent. For these networks, the average return for an agent serves as a measure the performance of the different algorithms.

% % \DM{I think here we can cover any doubts that might be raised in the experiment section of the main body. First, why we have the average return in one of our experiments - we can say something like for one of our games (shown in experiment section) the NE is efficient. Then secondly, why exploitability is needed as a metric for our non-cooperative settings and thirdly, why the shared reward baselines don't make sense.}

\tikzstyle{vertex}=[circle,fill=black,minimum size=16pt,inner sep=0pt,text=white]
\tikzstyle{ellipsis}=[circle,fill=white,minimum size=16pt,inner sep=0pt,text=black]
\tikzstyle{arrow} = [draw,thick,line width=1.5pt,->,>=stealth,black!75]
\tikzstyle{dots} = [draw,dotted,line width=1.5pt,->,>=stealth,black!40]
\tikzstyle{weight} = [font=\footnotesize]

\begin{figure}
	\begin{center}
		\begin{tikzpicture}[scale=1.5, auto,swap]
		\foreach \pos/\name in {
		{(0,0)/0}, {(1.5,1)/1}, {(1.5,-1)/3}, {(1.5,0)/2},
		{(3,1)/4}, {(3,-1)/6}, {(3,0)/5},
		{(5,1)/193}, {(5,-1)/195}, {(5,0)/194},
		{(6.5,1)/196}, {(6.5,-1)/198}, {(6.5,0)/197}, {(8,0)/199}}
		\node[vertex] (\name) at \pos {$\name$};
		\node[ellipsis] (u) at ($(4)!.5!(193)$) {$\ldots$};
		\node[ellipsis] (m) at ($(0)!.5!(199)$) {$\ddots$};
		\node[ellipsis] (d) at ($(6)!.5!(195)$) {$\ldots$};
		\path[arrow] (0) -- node[below] {} (1);
% 		\path[arrow] (193) -- node[right] {} (198);
% % 		\path[arrow] (1) -- node[right] {} (4);
		\path[arrow] (4) -- node[above] {} (u);
		\path[arrow] (u) -- node[above] {} (193);
		\path[arrow] (4) -- node[right] {} (m);
		\path[arrow] (m) -- node[right] {} (195);
		\foreach \source/ \dest /\weight in { 0/3/{}, 1/4/{}, 3/6/{}, 193/196/{}, 195/198/, 196/199/{}, 198/199/{}, 0/2/, 2/5/,5/m/, m/194/, 194/197/, 197/199/, 1/5/,193/197/,2/6/, 195/197/, 2/4/,194/196/,3/5/}
		\path[arrow] (\source) -- node[above] {$\weight$} (\dest);
		\foreach \source/ \dest /\weight in { 6/d/, d/195/}
		\path[dots] (\source) -- node[above] {$\weight$} (\dest);
		\end{tikzpicture}
	\end{center}
	\caption{Selfish routing network with 200 nodes.}
		\label{fig:large-routing-network}
\end{figure}
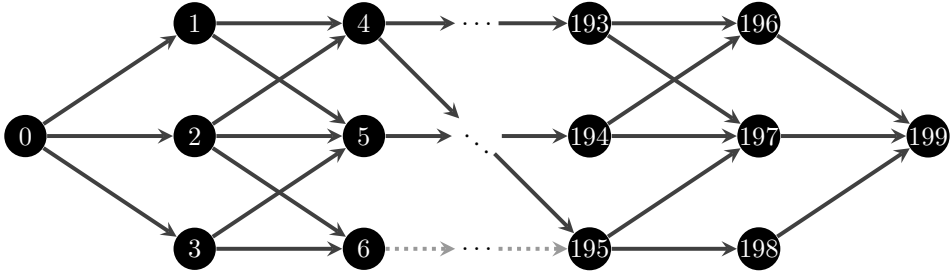

\subsection{Non-Cooperative, Stochastic Potential Game Ablation Study}
% \DM{@Yali, can we include image of the network here}
Fig. \ref{fig:large-routing-network} demonstrates the large network used in experiments.
% \DM{We can write in the explanation for the exploitability for this ablation study}
The class of potential games includes all team-games as a subclass, i.e. every team game is a potential game, and some but not all non-cooperative games are potential. In this ablation study, we examine the performance of $\ourmethod$ against other baselines in stochastic potential but non-cooperative games. 

To do this, using Lemma \ref{dummy_coord_prop},  we know that, in any potential game, the reward function $R_i:\mathcal{S}\times\left(\times\mathcal{A}_{i\in\mathcal{N}}\right)\to\mathbb{R}$ for any agent $i\in\mathcal{N}$ can be decomposed into two components: the team game component $J:\mathcal{S}\times\left(\times\mathcal{A}_{i\in\mathcal{N}}\right)\to\mathbb{R}$ (i.e.  function that all agents seek to maximise) and a strategic (non-cooperative) component which is specific to each agent $L_i:\mathcal{S}\times\left(\times\mathcal{A}_{i\in\mathcal{N}/{\{i\}}}\right)\to\mathbb{R}$. We now study a set of games in which each agent's reward function has the following form:
\begin{align}
R_i(s,(a_i,a_{-i})) = \underbrace{J(s,(a_i,a_{-i}))}_{\text{Team game reward}} +  c \underbrace{L_i(s,a_1,\ldots \textcolor{red}{a_{i-1},a_{i+1}},\ldots,a_N)}_{\text{Non-cooperative part}} \label{dummy_coord_sep_o}
\end{align}

The value of the constant $c\in\mathbb{R}$ determines the contribution of the non-cooperative, strategic component. For $c=0$, the game is a team game and as $c\to\infty$ the non-cooperative component of the game dominates.

As can be seen in the plots, $\ourmethod$ has better average return compared to the DDPG-based algorithms, whose performance degrades most due to their team-game requirement. COMIX and COVDN achieve similar levels of performance.

\begin{figure}[h]
	\begin{center}
		% \fbox{\rule{0pt}{1.in} \rule{0.7\linewidth}{0pt}}
		\includegraphics[width=0.6\linewidth]{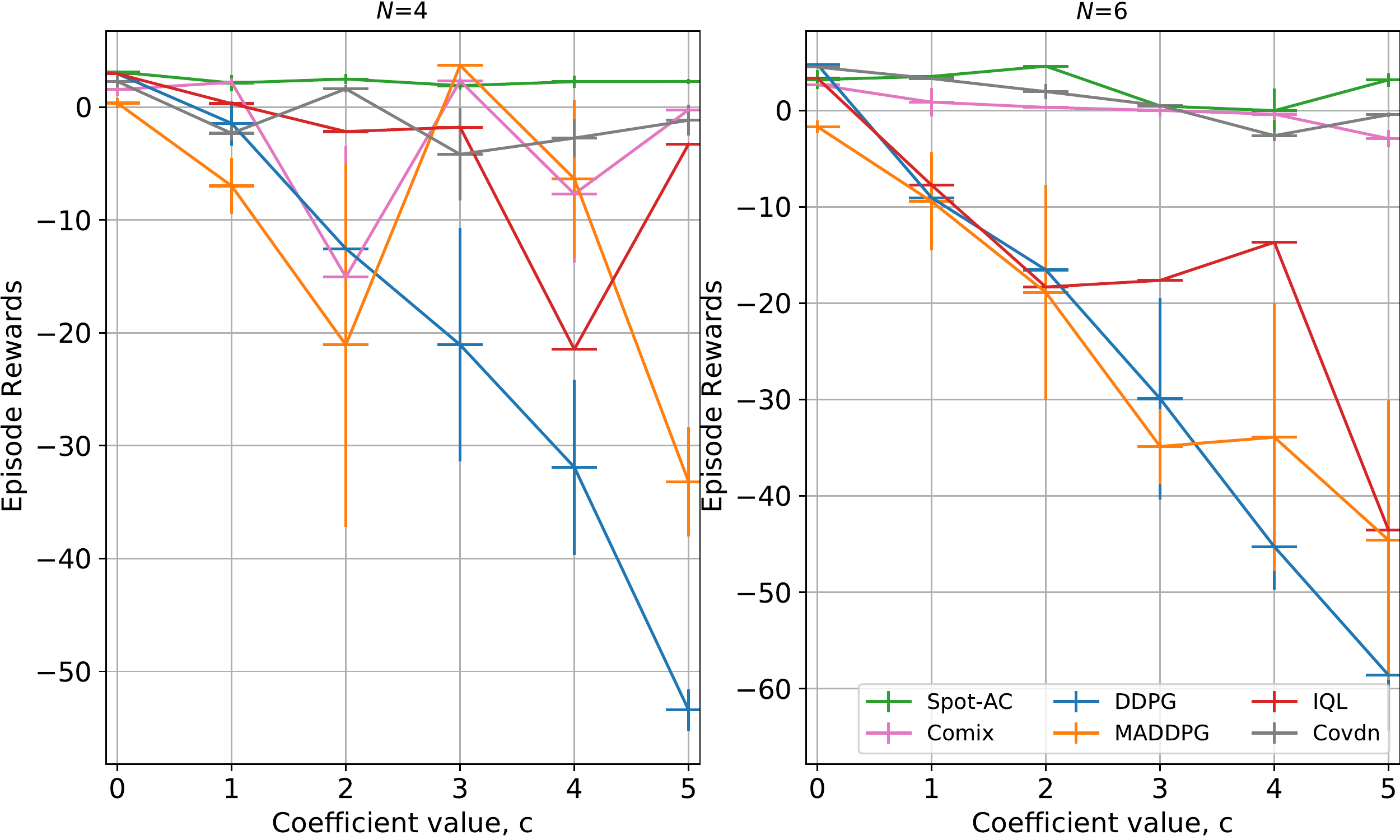}
	\end{center}
	\caption{Results of the training curves for the non-cooperative, potential non-atomic routing game when the number of agents $N= 4,6$, with $c=0,1,2,3,4,5$. 
}
	
	\label{fig:routing-ablation2-curve}
\end{figure}

\subsection{Non-cooperative, Non-Potential Stochastic Game Ablation Study}

Extending the ablation studies of the previous section, we examine the performance of $\ourmethod$ in games that are both non-cooperative and not potential. We parameterise the agents' reward function for the congestion game as follows: 
\[
R_i(s,(a^i,a^{-i})) = \underbrace{r_i(s,(a^i,a^{-i}))}_{\text{potential reward function}} + \hspace{2mm}  c \underbrace{J(s,(a^i,a^{-i}))}_{\text{non-potential contribution}}, \qquad \forall s\in\mathcal{S},\forall (a^i,a^{-i})\in\boldsymbol{\mathcal{A}}. \label{ablation_2_reward}
\]

The functions $r_i$ are those from the original (potential) congestion game \eqref{potential_condition_static}. $J$ is a generic non-potential reward function.
$c=0$ corresponds to a potential game, and as $c\to\infty$ the non-potential component of the game dominates.

Fig. \ref{fig:routing-ablation1} shows the results of this ablation study in a network routing game with 4 agents. We see that \ourmethod is able to handle small deviations (small $c$) from the potential requirements, but performance degrades for larger values. It again outperforms the DDPG baselines, whose performance degrades rapidly with increasing values of the ablation parameter.

\begin{figure*}[h]
	\begin{center}
		% \fbox{\rule{0pt}{1.in} \rule{0.7\linewidth}{0pt}}
		%\includegraphics[width=0.99\linewidth]{figs/routing_all_env2_ablation2-curve.pdf}
		\includegraphics[width=0.99\linewidth]{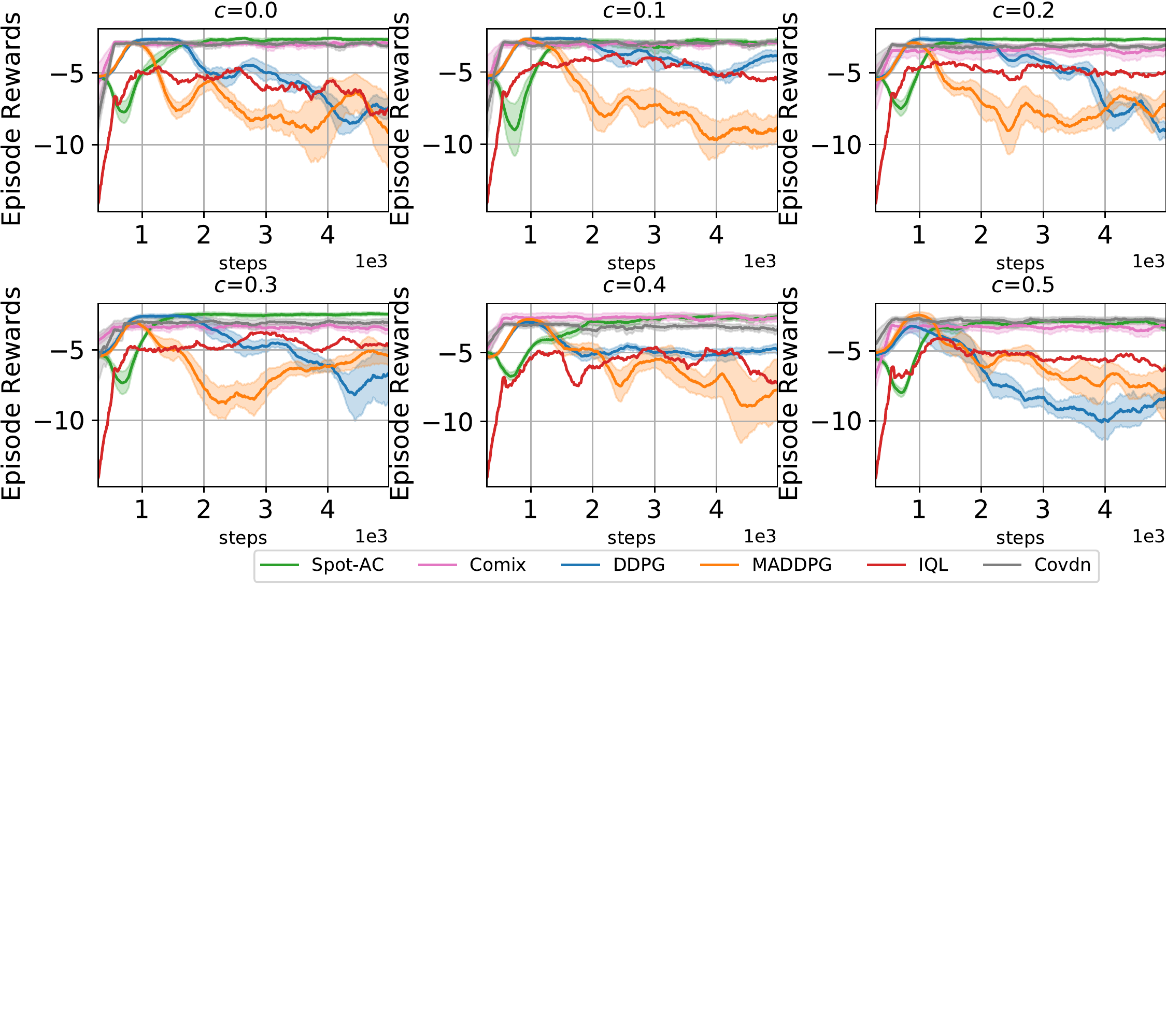}
		
	\end{center}
	\caption{Results of the training curves for non-cooperative, non-potenial, non-atomic routing game when number of agents $N= 4$, coefficients $c=0.0 - 0.5$. }
	\label{fig:routing-ablation1}
\end{figure*}

\clearpage

\section{Cournot Duopoly Problem}\label{sec:cournot_duopoly}

% ========================================
\textbf{Cournot Duopoly} is a classic static game \cite{monderer1996potential} that models the  imperfect competition in which multiple firms compete in price and production to capture market share. Since the firms' actions are continuous variables, the game is a continuous action setting. It is a nonzero sum game (neither team-based nor zero-sum) which represents a challenge for current MARL methods.   
Let $a_i\in [-A_i,A_i]$ where $A_i\in \mathbb{R}_{>0}$ which represent the set of actions for Firm $i\in\{1,2\ldots, N\}:=\mathcal{N}$.
Let $\alpha,\beta,\gamma\in\mathbb{R}_{>0}$ be given constants, each firm i's reward (profit) is 
$R_i(a_i,a_{-i})=a_i(\alpha -\beta\sum_{i\in\mathcal{N}}a_i)-\gamma a_i$. We set $A_i=1, \forall i\in\mathcal{N}$ and $\alpha=\gamma=1, \beta=-1$.  

\section{Analytic Example: Cournot Duopoly}\label{sec:appendix:cournot_analytic}
Reward functions

Let $a_1\in [-A_1,A_1]$ and $a_2\in [-A_2,A_2]$ where $A_1,A_2\in \mathbb{R}_{>0}$ which represent the actions for Firm 1 and Firm 2 respectively.

Also let $\alpha,\beta,\gamma\in\mathbb{R}_{>0}$ be given constants.
\begin{align}
R_i(a_1,a_2)=a_i(\alpha -\beta(a_1+a_2))-\gamma a_i    \label{cournot_reward_function}
\end{align}

\textbf{Cournot Potential Function  ($N= 2$ Agents)}
\begin{align}
\phi(a_1,a_2)=\alpha(a_1+a_2)-\beta(a_1^2+a_2^2) -\beta a_1 a_2 -\gamma (a_1+a_2)+k   \label{cournot_potential}
\end{align}
where $k\in\mathbb{R}$ is an arbitrary constant.

\subsection{Cournot Duopoly with $N>2$ Agents}
Reward functions

Let $a_i\in [-A_i,A_i]$ where $A_i\in \mathbb{R}_{>0}$ which represent the actions for Firm $i$, $i\in\{1,2\ldots, N\}$.

Also let $\alpha,\beta,\gamma\in\mathbb{R}_{>0}$ be given constants.
\begin{align}
R_i(a_i,a_{-i})=a_i(\alpha -\beta\sum_{i\in\mathcal{N}}a_i)-\gamma a_i    \label{cournot_reward_function_many}
\end{align}

\textbf{Cournot Potential Function ($N\geq 2$ Agents)}
\begin{align}
\phi(a_i,a_{-i})=\alpha\sum_{i\in\mathcal{N}}a_i-\beta\sum_{i\in\mathcal{N}}a_i^2 -\beta a_i\sum_{j\in\mathcal{N}/\{i\}}a_j -\gamma\sum_{i\in\mathcal{N}}a_i+k   \label{cournot_potential_many}
\end{align}
where $k\in\mathbb{R}$ is an arbitrary constant.

\textbf{Derivatives}
\begin{align}
\frac{\partial R_i(a_i,a_{-i})}{\partial a_i}=\alpha -\beta\sum_{i\in\mathcal{N}}a_i-\beta a_i-\gamma     \label{cournot_reward_function_many_derivatives}
\end{align}
\begin{align}
\frac{\partial \phi(a_i,a_{-i})}{\partial a_i}=\alpha-2\beta a_i -\beta\sum_{j\in\mathcal{N}/\{i\}}a_j-\gamma   \label{cournot_potential_many_derivative}
\end{align}

\subsection{Analytic Verification of our Method}
In this section, we validate that the optimisation in Sec. \ref{sec:compute-phi} yields the correct results. To do this, we derive analytic expressions for $\phi$ and show that the solution of the optimisation in Sec. \ref{sec:compute-phi} coincides with this solution.
Recall that our proposition says that:
\begin{align}
\mathbb{E}_{(a^i,a^{-i})\sim(\pi_i(\eta^i),\pi_{-i}(\eta^{-i}))}\left[\frac{\partial}{\partial \eta^i}\ln{[\pi_i(a^i|s;\eta^i)]}\left(\frac{\partial}{\partial a^i} R_i(s,a^i,a^{-i})-\frac{\partial}{\partial a^i}\phi(s,a^i,a^{-i})\right)\right]=\boldsymbol{0}, \label{proposition_within_cournot_example}
\end{align}

We first want to check that any PF in \eqref{cournot_potential} solves \eqref{proposition_within_cournot_example}, indeed:

For $\implies$ we find that
\begin{align*}
    \frac{\partial}{\partial a^i}R_i(\cdot,a^i,a^{-i})&=\alpha -\beta(a_1+a_2)-\beta a_i -\gamma
\end{align*}
and
\begin{align*}
    \frac{\partial}{\partial a^i}\phi(\cdot,a^i,a^{-i})&=\alpha -2\beta a_i-\beta a_j-\gamma 
    \\&=\alpha - \beta(a_1+a_2)-\beta a_i-\gamma
\end{align*}
and hence verify:
\begin{align*}
    \frac{\partial}{\partial a^i}\phi(\cdot,a^i,a^{-i})-    \frac{\partial}{\partial a^i}R_i(\cdot,a^i,a^{-i})=0
\end{align*}
so that any $\phi$ in \eqref{cournot_potential} is a candidate solution to \eqref{proposition_within_cournot_example}. Indeed, we observe that
\begin{align*}
      & \frac{\partial}{\partial a^i}\phi(\cdot,a^i,a^{-i})-    \frac{\partial}{\partial a^i}R_i(\cdot,a^i,a^{-i})=0
    \\&\implies
\mathbb{E}_{(a^i,a^{-i})\sim(\pi_i(\eta^i),\pi_{-i}(\eta^{-i}))}\left[\frac{\partial}{\partial \eta^i}\ln{[\pi_i(a^i|s;\eta^i)]}\left(\frac{\partial}{\partial a^i} R_i(s,a^i,a^{-i})-\frac{\partial}{\partial a^i}\phi(s,a^i,a^{-i})\right)\right]=\boldsymbol{0},
\end{align*}
and hence the forward implication is verified.
$\impliedby$

To check the reverse we perform the following optimisation:

\begin{align*}
{\rm minimise} \;\;\; \int \left(\frac{\partial}{\partial \eta^i}\ln{[\pi_i(a^i|\cdot;\eta^i)]}\left(\frac{\partial}{\partial a^i} R_i(\cdot,a^i,a^{-i})-\frac{\partial}{\partial a^i}P_\rho(\cdot,a^i,a^{-i})\right)\right), 
\end{align*}

Consider candidate functions of the following form 
\begin{align*}
P_\rho(\cdot,a^i,a^{-i})=\rho_0+\rho_{a_1,1}a_1+\rho_{a_2,1}a_2+\rho_{a_1,2}a^2_1+\rho_{a_2,2}a^2_2+\rho_{a}a_1a_2+c
\end{align*}
and Gaussian policies: $\pi_i(a^i|\cdot;\eta^i)=\frac{1}{\sqrt{2\pi\sigma^2}}e^{\frac{1}{2}(\frac{a^i-\eta}{\sigma})^2}$

Then
\begin{align*}
\frac{\partial}{\partial a^i}P_\rho(\cdot,a^i,a^{-i})  = \rho_{a_i,1}+2\rho_{a_i,2}a_i +\rho_a a_j
\end{align*}

\begin{align*}
{\rm minimise} \;\;\; -\frac{1}{\sigma^2}\int (a^i-\eta)\left(\alpha -\beta(a_i+a_j)-\beta a_i -\gamma-(\rho_{a_i,1}+2\rho_{a_i,2}a_i +\rho_a a_j)\right), 
\end{align*}
After matching like terms we find that:
\begin{align*}
\alpha-\gamma &=\rho_{a_i,1}
\\-2\beta&= 2\rho_{a_i,2}
\\-\beta&=\rho_a
\end{align*}
Hence, we find that
\begin{align*}
P_\rho(\cdot,a^i,a^{-i})&=-(\gamma-\alpha)a_1-(\gamma-\alpha)a_2-\beta a^2_1-\beta a^2_2-\beta a_1a_2+c
\\&=\alpha(a_1+a_2)-\beta(a_1^2+a_2^2)-\beta a_1a_2-\gamma(a_1+a_2)=\phi+c
\end{align*}
which verifies the reverse.

\section{Estimating the Potential Function: Algorithm 2}\label{sec:algorithm_2}
The following algorithm computes the potential function of the SPG using the supervised learning method described in Sec. 5.1. We illustrate the convergence of the method in Sec. \ref{sec:appendix_convergence}.   
\begin{algorithm}
\textbf{Estimating the Potential Function}
\begin{algorithmic}[1]
    \STATE 
Generate set of random points $((\boldsymbol{\eta}^k_i,\boldsymbol{\eta}^k_{-i}),s^k)\in \boldsymbol{E}\times \mathcal{S}$ for $k=1,2,\ldots$ according to the probability density $\nu$.

\STATE  For each $(\boldsymbol{\eta}^k,s^k)\equiv ((\eta^k_i,\eta^k_{-i}),s^k)$, 
% denote approximation by $\bar{g}^i(s^k,{\boldsymbol{\eta}}^k;\rho^k)$. For each $(\boldsymbol{\eta}^k,s^k)\equiv ((\eta^k_i,\eta^k_{-i}),s^k)$ 
evaluate $g^i(s^k,{\boldsymbol{\eta}}^k;\boldsymbol{\rho}^k)$ (or for stochastic policies approximate expectation $\mathbb{E}_{(a_i,a_{-i})\sim(\pi_i,\pi_{-i})}\left[g^i(s^k,{\boldsymbol{\eta}}^k;\boldsymbol{\rho}^k)\right]$ by MC sampling of actions). 
% Evaluate the derivatives using either \eqref{sampling_potential_optimisation_g_i_estimate} or by another differentiation method.

\STATE Calculate the squared error $g^2(s^k,\boldsymbol{\eta}^k;\boldsymbol{\rho}^k)= \sum_{i\in\mathcal{N}}\left[g^i(s,\boldsymbol{\eta}^k;\boldsymbol{\rho}^k)\right]^2$ using step 2.

\STATE Take a descent step at $((\boldsymbol{\eta}^k_i,\boldsymbol{\eta}^k_{-i}),s^k)$, compute $
\boldsymbol{\rho}^k=\boldsymbol{\rho}^{k-1}-\alpha\nabla_{\boldsymbol{\rho}} g^2(s^k,\boldsymbol{\eta}^k,\boldsymbol{\rho})|^{\boldsymbol{\rho}=\boldsymbol{\rho}^{k-1}}$.
\STATE Repeat until convergence criterion is satisfied.
\end{algorithmic}
\end{algorithm}
\subsection{Convergence of The Potential Function}\label{sec:appendix_convergence}
% Our first demonstration shows that the optimisation  \ref{sec:compute-phi} indeed converges to the true potential function $\phi$, for the selfish routing games (Fig. \ref{fig:pot_approx}).

Fig. \ref{fig:pot_approx} gives the learning curves for computing the potential function for the selfish routing games using Algorithm 2, corresponding to the method in Section  \ref{sec:compute-phi}.
The potential function defines the team game which agents jointly seek to maximise.
For the training of potential function $\phi$, we use a batch size of 2. The learning converges after around 200 iterations, and can easily handle settings with large numbers of agents.

% % \subsection{Results of potential approximation }
% % We report for potential approximation in Fig. \ref{fig:potential-approx-routing}.

% % The training error decreases to zero, validating that the potential approximation is accurate.

\begin{figure}[h]
	\begin{center}
		% \fbox{\rule{0pt}{1.in} \rule{0.7\linewidth}{0pt}}
		\includegraphics[width=0.99\linewidth]{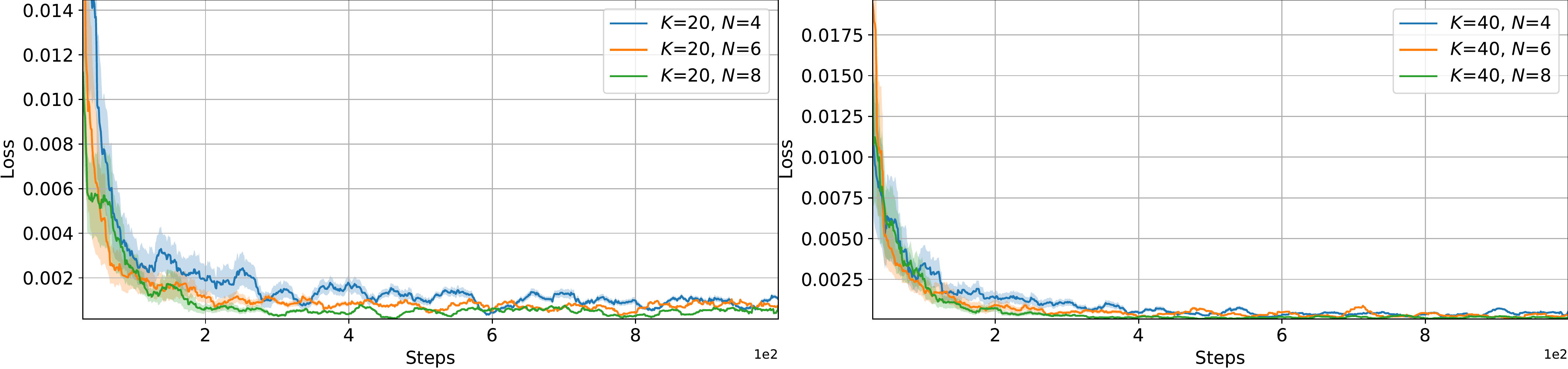}
	\end{center}
	\caption{Results of the training curves for potential approximation. Non atomic routing when number of agents $N= 4,6,8$, number of nodes $K=20,40$. The y-axis is $\|\Delta \phi - \Delta R_i\|^2$ (minimising this quantity gives a candidate  function for the potential function c.f. Prop. \ref{reparam_pg_ode}). 
% 	\textbf{TBA: a description of what is being trained, how (supervised learning) and the conclusion that we successfully learn the potential function using the method in sec 5.1. then say that this is a function that defines the team game which agents jointly seek to maximise. You can also relate this to the network image below and make the point that the method easily handles settings with large numbers of agents.}
	}
	\label{fig:pot_approx}
\end{figure}

\section{Consensus Optimisation}
\label{sec:Consensus Optimisation}
In what follows, we denote by $\mathcal{C}$ the set of continuously differentiable functions and by $\mathcal{H}$ the set of measurable functions.  

To perform the consensus optimisation step, we use the following update processes:
\begin{align}
&\boldsymbol{\rho}^i_{l+1}=\sum_{j\in\mathcal{N}}c_l(i,j)\rho^j_l-\alpha\cdot\kappa^i_l \label{consensus_update_rho}
\\&\kappa^i_{l+1}=\sum_{j\in\mathcal{N}}c_l(i,j)\cdot\kappa^j_l+\nabla g^i(s,\boldsymbol{\rho}^i_{l+1})-\nabla g^i(s,\boldsymbol{\rho}^i_{l}) \label{consensus_update_kappa}
\end{align}
where $\alpha>0$ is the stepsize and $\boldsymbol{C}_l=[c_l(i,j)]_{N\times N}$ is the consensus matrix at iteration $l\in\{0,1,\ldots\}$.

\begin{algorithm}
\textbf{Consensus update}
\label{algo:consensus}
\begin{algorithmic}[1]
    \REQUIRE Parametric function class $\mathcal{H}$, stepsize $\alpha>0$, initial consensus matrix $\boldsymbol{C}_0=[c_0(i,j)]_{N\times N}$ , initial parameter $\boldsymbol{\rho}^i_0\in \mathrm{T}, \kappa^i_0=\nabla g^i(s,\boldsymbol{\rho}^i_0)$ for all $i\in\mathcal{N}$. 
\FOR{$ l\in\{0,\ldots, L-1\}$ }
\FOR{agent $i\in\mathcal{N}$}     
        \STATE $\boldsymbol{\rho}^i_{l+1}=\sum_{j\in\mathcal{N}}c_l(i,j)\rho^j_l-\alpha\cdot\kappa^i_l$
 \STATE $\kappa^i_{l+1}=\sum_{j\in\mathcal{N}}c_l(i,j)\cdot\kappa^j_l+\nabla g^i(s,\boldsymbol{\rho}^i_{l+1})-\nabla g^i(s,\boldsymbol{\rho}^i_{l})$
\ENDFOR
\ENDFOR
\STATE \textbf{Output:} The vector of functions $[P_{\rho_i}]_{i\in\mathcal{N}}$ for all agent $i\in\mathcal{N}$.
\end{algorithmic}
\end{algorithm}

\section{Assumptions}
\label{sec:append_assumptions}
Given a pair of metric spaces $\left(X_1,d_1\right)$ and $\left(X_1,d_2\right)$, we say that a function $f:X_1\to X_2$ is Lipschitz if the constant defined by $
L_f:=\underset{x_1\in X_1,x_2\in X_2}{\sup}\frac{d_2(f(x_1),f(x_2))}{d_1(x_1,x_2)} $
is finite.
The constant $L_f$ is called the Lipschitz constant. We denote by $L_R{_\infty}:=\max\{L_{R_1},\ldots L_{R_N}\}$ and similarly for the L-Lipschitz gradients $L_{\frac{\partial R{_\infty}}{\partial \pi}}:=\max\left\{L_{\frac{\partial R{_1}}{\partial \pi}},\ldots L_{\frac{\partial R{_N}}{\partial \pi}}\right\}$.

Given a pair of metric spaces $\left(X_1,d_1\right)$ and $\left(X_1,d_2\right)$, we say that a function $f:X_1\to X_2$ is Lipschitz if the constant defined by $
L_f:=\underset{x_1\in X_1,x_2\in X_2}{\sup}\frac{d_2(f(x_1),f(x_2))}{d_1(x_1,x_2)} $
is finite.
The constant $L_f$ is called the Lipschitz constant. We denote by $L_R{_\infty}:=\max\{L_{R_1},\ldots L_{R_N}\}$ and similarly for the L-Lipschitz gradients $L_{\frac{\partial R{_\infty}}{\partial \pi}}:=\max\left\{L_{\frac{\partial R{_1}}{\partial \pi}},\ldots L_{\frac{\partial R{_N}}{\partial \pi}}\right\}$.

The results of the paper are built under the following assumptions:

\noindent\textbf{Assumption A.2.}
For any $\theta\in\Theta$, the functions $(R_{i,\theta})_{i\in\mathcal{N}}$ are bounded, measurable functions in the action inputs. 
\\
\noindent \textbf{Assumption A.3.}
The functions $\{R_i\}_{i\in\mathcal{N}}$ are Lipschitz and have L-Lipschitz continuous gradients in $\theta\in\Theta$ that is, for any $i\in\mathcal{N}$, there exist constants $L_{R_i}>0$ and $L_{\nabla_{\boldsymbol{\theta}}  R_i}>0$ s.th. for any $s,s'\in \mathcal{S}$ and $\forall\boldsymbol{a}\in \boldsymbol{\mathcal{A}}, \theta_a,\theta_b,\theta_c,\theta_d\in\Theta$ we have that:
\begin{align*}
\nonumber
&\left\|\nabla_{\boldsymbol{\theta}}  R_{i_{\theta_a}}(s';s,\boldsymbol{a})-\nabla_{\boldsymbol{\theta}}  R_{0_{\theta_b}}(s';s,\boldsymbol{a})\right\|
+\left\|R_{i_{\theta_c}}(s';s,\boldsymbol{a})- R_{i_{\theta_d}}(s';s,\boldsymbol{a})\right\|
\\&\leq L_{\nabla_{\boldsymbol{\theta}}  R_i}\left\|{\theta_a}-{\theta_b}\right\|+L_{R_i}\left\|{\theta_c}-{\theta_d}\right\|.
\end{align*}
\\
\noindent \textbf{Assumption A.4.}
The functions $\{R_i\}_{i\in\mathcal{N}}$ is continuously differentiable in the state and action inputs.
\\
\noindent \textbf{Assumption A.5.}
The set of policies $\{\pi\}_{i\in\mathcal{N},\boldsymbol{\eta}\in\boldsymbol{E}}$ is differentiable w.r.t. the policy parameter $\boldsymbol{\eta}$.

Assumption A.2. is rudimental and in general required in optimisation and stochastic approximation theory. Assumptions A.3. and A.4.  is typical in Q-learning proofs see pg 27 in \cite{szepesvari2005finite} (there in fact the transition function is also assumed to be Lipschitz),  assumption A7 in \cite{antos2008fitted}, in \cite{szepesvari2005finite} it is assumed that the transition function and reward function are smooth (see pg 21). In this setting, the assumptions are required to construct approximations of $\phi$ in terms of a differential equation. Assumption A.1. is fundamental to the structure of a state-based PG. In particular, it extends the notion of potentiality to the state input. The assumption is used in the proof of Theorem \ref{reduction_theorem}. Assumption A.5. is standard and required within the framework of policy gradient and actor-critic methods \cite{sutton2000policy,silver2014deterministic}.

\section{Proof of Theoretical Results}\label{sec:appendix_proof_results}
\subsection{Auxiliary Results}
\setcounter{lemma}{0}
\renewcommand\thelemma{\Alph{lemma}}
Let us denote by $\left(\mathcal{V},\|\|\right)$ any normed vector space.

\begin{lemma} \label{max_lemma}
For any
$f:\mathcal{V}\times\mathcal{V}\to\mathbb{R},g:\mathcal{V}\times\mathcal{V}\to\mathbb{R}$, we have that:
\begin{align}
\left\|\underset{a\in\mathcal{V}}{\sup}\:f(a)-\underset{a\in\mathcal{V}}{\sup}\: g(a)\right\| \leq \underset{a\in\mathcal{V}}{\sup}\: \left\|f(a)-g(a)\right\|.    \label{lemma_1_basic_max_ineq}
\end{align}
\end{lemma}
\begin{proof}
\begin{align}
f(a)&\leq \left\|f(a)-g(a)\right\|+g(a)\label{max_inequality_proof_start}
\\\implies\nonumber
\underset{a\in\mathcal{V}}{\sup}f(a)&\leq \underset{a\in\mathcal{V}}{\sup}\{\left\|f(a)-g(a)\right\|+g(a)\}
\\&\leq \underset{a\in\mathcal{V}}{\sup}\left\|f(a)-g(a)\right\|+\underset{a\in\mathcal{V}}{\sup}\;g(a). \label{max_inequality}
\end{align}
Deducting $\underset{a\in\mathcal{V}}{\sup}\;g(a)$ from both sides of (\ref{max_inequality}) yields:
\begin{align}
    \underset{a\in\mathcal{V}}{\sup}f(a)-\underset{a\in\mathcal{V}}{\sup}g(a)\leq \underset{a\in\mathcal{V}}{\sup}\left\|f(a)-g(a)\right\|.\label{max_inequality_result_last}
\end{align}
After reversing the roles of $f$ and $g$ and redoing steps (\ref{max_inequality_proof_start}) - (\ref{max_inequality}), we deduce the desired result since the RHS of (\ref{max_inequality_result_last}) is unchanged.
\end{proof}

The proof of the Theorem \ref{reduction_theorem} is established through the following results:
\begin{lemma}\label{trans_lemma}
For any c-SPG, the state transitivity conditions holds whenever the reward functions takes the following form:\begin{align*}
&R_i(s,(a^i,a^{-i}))=g(s,(a^i,a^{-i}))+k(s)h_i(a^i,a^{-i}),
\end{align*}
for any functions $g,k,h$ for which $k^{-1}$ exists. 
\end{lemma}
\begin{proof}
To prove the result, we show that the class of games can be rescaled accordingly. Indeed, we first note that $
k^{-1}(s)R_i(s,(a^i,a^{-i}))-k^{-1}(s')R_i(s',(a^i,a^{-i}))
=k^{-1}(s)g(s,(a^i,a^{-i}))-k^{-1}(s')g(s',(a^i,a^{-i}))$. Using the invertibility of $k$, we now consider a rescaled game \\$\mathcal{M}(s)=\langle \left(\mathcal{A}_i\right)_{i\in\mathcal{N}},\left(V_{i}(s)\right)_{i\in\mathcal{N}},\mathcal{N}\rangle$ where $V_i:=k^{-1}R_i$, then it is easy to see that for these games we have: $
V_i(s,(a^i,a^{-i}))-V_i(s',(a^i,a^{-i}))=L(s,(a^i,a^{-i}))-L(s',(a^i,a^{-i}))$
where $L:=k^{-1}g$ and hence the state transitivity assumption is satisfied.
\end{proof}

\begin{lemma}\label{dummy_coord_prop}
For any PG, there exists a function $B:\boldsymbol{\Pi}\times \mathcal{S}\to\mathbb{R}$ ($B\in\mathcal{H}$) such that the following holds for any $i\in\mathcal{N}, \forall (a^i_t,a^{-i}_t)\in\boldsymbol{\mathcal{A}}, \forall s \in\mathcal{S}$ $
R_i(s,a^i_t,a^{-i}_t)=\phi(s,a^i_t,a^{-i}_t)+F_{i}(s,a^{-i}_t)$ 
where $F_i$ satisfies $
F_{i}(s,a^{-i}_t)=F_{i}(s',a^{-i}_t)$.
\end{lemma}
The result generalises dummy-coordination separability known in PGs to a state-based setting \cite{slade1994does,ui2000shapley}.

\begin{proof}[Proof of Lemma \ref{dummy_coord_prop}]
To establish the forward implication, we make the following  observation which is straightforward:
\begin{align*}
&R_i(s,a^i_t,a^{-i}_t)-R_i(s,a'^i_t,a^{-i}_t)
\\&=\phi(s,a^i_t,a^{-i}_t)+F_{i}(s,a^{-i}_t)-\left(\phi(s,a^i_t,a^{-i}_t)+F_{i}(s,a^{-i}_t)\right)
\\&=\phi(s,a^i_t,a^{-i}_t)-\phi(s,a^i_t,a^{-i}_t)
\end{align*}
To prove the reverse assume that the game is a state-based potential game. Let us now define the function $T_i(s,a^i_t,a^{-i}_t):=R_i(s,a^i_t,a^{-i}_t)-\phi(s,a^i_t,a^{-i}_t)$, then we observe that: $
    R_i(s,a^i_t,a^{-i}_t)-R_i(s,a'^i_t,a^{-i}_t)=\phi(s,a^i_t,a^{-i}_t)-\phi(s,a'^i_t,a^{-i}_t) \iff R_i(s,a^i_t,a^{-i}_t)-\phi(s,a^i_t,a^{-i}_t)=R_i(s,a'^i_t,a^{-i}_t)-\phi(s,a'^i_t,a^{-i}_t)$
and hence $
    T_i(s,a^i_t,a^{-i}_t)=T_i(s,a'^i_t,a^{-i}_t)$
which implies that $T_i(s,a^i_t,a^{-i}_t)\equiv K_i(s,a^{-i}_t)$. In a similar way, writing $T_i(s,a^i_t,a^{-i}_t):=R_i(s,a^i_t,a^{-i}_t)-\phi(s,a^i_t,a^{-i}_t)$ and using the state transitive property, we deduce that $T_i(s',a^i_t,a^{-i}_t)=T_i(s,a^i_t,a^{-i}_t)$ which settles the proof.
\end{proof}

\subsection{Proof of Main Results}
\subsection*{Proof of Theorem \ref{reduction_theorem}}

\begin{proposition}\label{dpg_proposition}
There exists a function $B:\boldsymbol{\Pi}\times \mathcal{S}\to\mathbb{R}$ ($B\in\mathcal{H}$) and the following holds for any $i\in\mathcal{N}$
\begin{align}
\mathbb{E}_{s\sim P(\cdot|)}\left[ v^{\boldsymbol{\pi}}_i(s)-v^{\boldsymbol{\pi'}}_i(s)
\right]=\mathbb{E}_{s\sim P(\cdot|)}\left[B^{\boldsymbol{\pi}}(s)-B^{\boldsymbol{\pi'}}(s)\right].\label{potential_relation_proof}
\end{align}
\end{proposition}

\begin{proof}
% The Proposition is proven in two parts. Firstly, we tackle the finite horizon case ($T<\infty$), then thereafter the infinite horizon case ($T=\infty$).

We prove the proposition in two parts beginning with the finite case then extending to the infinite horizon case.

Hence, we first seek to show that for any joint strategy $(\pi_i,\pi_{-i})\in\boldsymbol{\Pi}$, define by $v_{i,k}$ the value function for the finite horizon game of length $k\in\mathbb{N}$ (i.e. \\$
v^{\boldsymbol{\pi}}_{i,k}(s):=\mathbb{E}_{s_t\sim P,\pi_i,\pi_{-i}}\left[\sum_{t=0}^k\gamma^tR_i(s_t,\boldsymbol{a}_t)|s\equiv s_0\right]$ for any $i\in\mathcal{N}$ and $k<\infty$). Then there exists a function $B_k:\boldsymbol{\Pi}\times \mathcal{S}\to\mathbb{R}$ such that the following holds for any $i\in\mathcal{N}$ and $\forall \pi_i,\pi'_i\in\Pi_i, \forall \pi_{-i}\in\Pi_{-i}$:
\begin{align}
& \mathbb{E}_{s\sim P(\cdot|)}\left[v^{\boldsymbol{\pi}}_{i,k}(s)-v^{\boldsymbol{\pi'}}_{i,k}(s)\right]
=\mathbb{E}_{s\sim P(\cdot|)}\left[B_k^{\boldsymbol{\pi}}(s)-B_k^{\boldsymbol{\pi'}}(s)\right].
\label{dpg_equality_finite}
\end{align}

For the finite horizon case, the result is proven by induction on the number of time steps until the end of the game. Unlike the infinite horizon case, for the finite horizon case the value function and policy have an explicit time dependence. 

We consider the case of the proposition at time $ T-1$ that is we evaluate the value functions at the penultimate time step. In the following, we employ the shorthands  $\boldsymbol{a}_k\equiv (a^i_k,a^{-i}_k)$ and by $\boldsymbol{a'}_k\equiv (a^{'i}_k,a^{-i}_k)$ for any $k\in\mathbb{N}$ and similarly $\boldsymbol{\pi}(\cdot)\equiv \prod_{j\in\mathcal{N}}\pi_j$ and $\boldsymbol{\pi'}(\cdot)\equiv \prod_{j\in\mathcal{N}/\{i\}}\pi_j(\cdot)\cdot\pi'_i(\cdot)$. We will also use the shorthands $F^{\boldsymbol{\pi}}\equiv F^{(\pi_i,\pi_{-i})}$ and $F^{\boldsymbol{\pi'}}\equiv F^{(\pi'_i,\pi_{-i})}$ given some function $F$.

In what follows and for the remainder of the script, we employ the following shorthands:
\begin{align*}
&\mathcal{P}^{\boldsymbol{a}}_{ss'}=:P(s';\boldsymbol{a},s), \quad\mathcal{P}^{\boldsymbol{\pi}}_{ss'}=:\int_{\boldsymbol{\mathcal{A}}}d\boldsymbol{a}_t\boldsymbol{\pi}(\boldsymbol{a}_t|s)\mathcal{P}^{\boldsymbol{a}_t}_{ss'}, \quad \mathcal{R}_i^{\boldsymbol{\pi}}(s_{t}):=\int_{\boldsymbol{\mathcal{A}}}d\boldsymbol{a}_t\boldsymbol{\pi}(\boldsymbol{a}_t|s_t)R_i(s_t,\boldsymbol{a}_t)
\\&\boldsymbol{\phi}^{\boldsymbol{\pi}}(s_t):=\int_{\boldsymbol{\mathcal{A}}}d\boldsymbol{a}_t\boldsymbol{\pi}(\boldsymbol{a}_t|s_t)\phi(s_t,\boldsymbol{a}_t), \quad \boldsymbol{F}_i^{\boldsymbol{\pi}}(s_t):=\int_{\boldsymbol{\mathcal{A}}}d\boldsymbol{a}_t\boldsymbol{\pi}(\boldsymbol{a}_t|s_t)F_i(s_t,\boldsymbol{a}_t), \quad \boldsymbol{v}_{i,k}^{\boldsymbol{\pi}}(s_t):=\int_{\boldsymbol{\mathcal{A}}}d\boldsymbol{a}_t\boldsymbol{\pi}(\boldsymbol{a}_t|s_t)v^{\boldsymbol{\pi}}_{i,k}(s_t,\boldsymbol{a}_t)
\end{align*}

In this case, we have that:

\begin{align*}\nonumber &\mathbb{E}_{s_{T-1}\sim P(\cdot|)}\left[v^{\boldsymbol{\pi}}_{i,T-1}(s_{T-1})-v^{\boldsymbol{\pi'}}_{i,T-1}(s_{T-1})\right]
\\&=\mathbb{E}_{s_{T-1}\sim P(\cdot|)}\left[\mathcal{R}_i^{\boldsymbol{\pi}}(s_{T})+\gamma\int_{\mathcal{S}}ds_T\int_{\boldsymbol{\mathcal{A}}} \mathcal{P}^{\boldsymbol{\pi}}_{s_Ts_{T-1}}v_i^{\boldsymbol{\pi}}(s_T)-\left( \mathcal{R}_i^{\boldsymbol{\pi'}}(s_{T})
+\gamma\int_{\mathcal{S}}ds_T\mathcal{P}^{\boldsymbol{\pi'}}_{s_Ts_{T-1}}v_i^{\boldsymbol{\pi'}}(s_T)\right)\right]
\\&\nonumber=\mathbb{E}_{s_{T-1}\sim P(\cdot|)}\Bigg[\boldsymbol{\phi}^{\boldsymbol{\pi}}(s_T)-\boldsymbol{\phi}^{\boldsymbol{\pi'}}(s_T)+\gamma\int_{\mathcal{S}}ds_T \mathcal{P}^{\boldsymbol{\pi}}_{s_Ts_{T-1}}v_i^{\boldsymbol{\pi}}(s_T)
-\gamma\int_{\mathcal{S}}ds\mathcal{P}^{\boldsymbol{\pi'}}_{s_Ts_{T-1}}v_i^{\boldsymbol{\pi'}}(s_T)\Bigg]
% \\&\begin{aligned}=\mathbb{E}_{s_{T-1}\sim P(\cdot|)}\Bigg[&\int_{\mathcal{A}_i}da^i_{T-1} \int_{\mathcal{A}_{-i}}\pi_i(da^i_{T-1};s_{T-1})\pi_{-i}(da^{-i}_{T-1},s_{T-1})
% \phi(s_{T-1},\boldsymbol{a}_{T-1})
% \\&-\int_{\mathcal{A}_i}da^i_{T-1} \int_{\mathcal{A}_{-i}}\pi'_i(da^{'i}_{T-1};s_{T-1})\pi_{-i}(da^{-i}_{T-1},s_{T-1})\phi(s_{T-1},\boldsymbol{a'}_{T-1}))
% \\&+\gamma\int_{\mathcal{S}}ds\int_{\mathcal{A}_i}da^i_{T-1} \int_{\mathcal{A}_{-i}}da^{-i}_{T-1}\pi_i(da^i_{T-1},s_{T-1})\pi_{-i}(da^{-i}_{T-1},s_{T-1})P(s_{T};s_{T-1},\boldsymbol{a}_{T-1})v_i^{\boldsymbol{\pi}}(x_T)
% \\&-
% \gamma\int_{\mathcal{S}}ds\int_{\mathcal{A}_i}da^i_{T-1} \int_{\mathcal{A}_{-i}}da^{-i}_{T-1}\pi'_i(da^{'i}_{T-1},s_{T-1})\pi_{-i}(da^{-i}_{T-1},s_{T-1})P(s_{T};s_{T-1},\boldsymbol{a'}_{T-1})v_i^{\pi^{i},\pi^{-i}}(x_T)\Bigg].
% \end{aligned}
\\&\begin{aligned}=\mathbb{E}_{s_{T-1}\sim P(\cdot|)}\Bigg[&\boldsymbol{\phi}^{\boldsymbol{\pi}}(s_{T-1})-\boldsymbol{\phi}^{\boldsymbol{\pi'}}(s_{T-1})+\gamma\mathbb{E}_{s_{T}\sim P(\cdot|)}\Bigg[\boldsymbol{v}_i^{\boldsymbol{\pi}}(s_T)
- \boldsymbol{v}_i^{\boldsymbol{\pi'}}(s_T)\Bigg]\Bigg].
\end{aligned}
\end{align*}
We now observe that for any $\pi_i\in\Pi_i$ and for any $\pi_{-i}\in\Pi_{-i}$ we have that $\forall i\in\mathcal{N}, \;v_i^{\pi^{i},\pi^{-i}}(s_T)=\mathbb{E}_{s_T\sim P(\cdot|)}\left[\mathcal{R}_i^{\boldsymbol{\pi}}(s_{T})\right]$, moreover we have that for any $\pi_i,\pi_i'\in\Pi_i$ and for any $i\in\mathcal{N}$, we have
\begin{align*}
&\mathbb{E}_{s_{T}\sim P(\cdot|)}\Bigg[\mathcal{R}_i^{\boldsymbol{\pi}}(s_{T})-\mathcal{R}_i^{\boldsymbol{\pi'}}(s_{T})\Bigg]
\\&=\int_{\mathcal{S}} ds_T\mathcal{P}^{\boldsymbol{\pi}}_{s_Ts_{T-1}} \mathcal{R}_i^{\boldsymbol{\pi}}(s_{T})
-\int_{\mathcal{S}} ds_T\mathcal{P}^{\boldsymbol{\pi'}}_{s_Ts_{T-1}} \mathcal{R}_i^{\boldsymbol{\pi'}}(s_{T})
\\&=\int_{\mathcal{S}} ds_T\mathcal{P}^{\boldsymbol{\pi}}_{s_Ts_{T-1}}\left[ \boldsymbol{\phi}^{\boldsymbol{\pi}}(s_{T})+\boldsymbol{F}_i(a^{-i}_T)\right]
-\int_{\mathcal{S}} ds_T\mathcal{P}^{\boldsymbol{\pi'}}_{s_Ts_{T-1}}\left[ \boldsymbol{\phi}^{\boldsymbol{\pi'}}(s_{T})+\boldsymbol{F}_i(a^{-i}_T)\right]
\\&=\int_{\mathcal{S}} ds_T\mathcal{P}^{\boldsymbol{\pi}}_{s_Ts_{T-1}} \boldsymbol{\phi}^{\boldsymbol{\pi}}(s_{T})-\int_{\mathcal{S}} ds_T\mathcal{P}^{\boldsymbol{\pi'}}_{s_Ts_{T-1}} \boldsymbol{\phi}^{\boldsymbol{\pi'}}(s_{T})+\int_{\mathcal{S}} ds_T\mathcal{P}^{\boldsymbol{\pi}}_{s_Ts_{T-1}}\boldsymbol{F}_i(a^{-i}_T)-\int_{\mathcal{S}} ds_T\mathcal{P}^{\boldsymbol{\pi'}}_{s_Ts_{T-1}}\boldsymbol{F}_i(a^{-i}_T)
\\&=\int_{\mathcal{S}} ds_T\mathcal{P}^{\boldsymbol{\pi}}_{s_Ts_{T-1}} \boldsymbol{\phi}^{\boldsymbol{\pi}}(s_{T})-\int_{\mathcal{S}} ds_T\mathcal{P}^{\boldsymbol{\pi'}}_{s_Ts_{T-1}} \boldsymbol{\phi}^{\boldsymbol{\pi'}}(s_{T})
\\&=\mathbb{E}_{s_{T}\sim P(\cdot|)}\left[\boldsymbol{\phi}^{\boldsymbol{\pi}}(s_{T})-\boldsymbol{\phi}^{\boldsymbol{\pi'}}(s_{T})\right]
\end{align*}
% \begin{align*}
% &\begin{aligned}
% \mathbb{E}_{s_{T}\sim P(\cdot|)}\Bigg[\int_{\mathcal{A}_i}\int_{\mathcal{A}_{-i}}&\pi_i(da^i_{s_{T}},s_{T})\pi_{-i}(da^{-i}_{s_{T}},s_{T})R_i(s_{T},a^i_{s_{T}},a^{-i}_{s_{T}})
% \\&-\int_{\mathcal{A}_i}\int_{\mathcal{A}_{-i}}\pi'_i(da^i_{s_{T}},s_{T})\pi_{-i}(da^{-i}_{s_{T}},s_{T})R_i(s_{T},a'^i_{s_{T}},a^{-i}_{s_{T}})\Bigg]
% \end{aligned}
% \\&\begin{aligned}
% =\mathbb{E}_{s_{T}\sim P(\cdot|)}\Bigg[\int_{\mathcal{A}_i}\int_{\mathcal{A}_{-i}}&\pi_i(da^i_{s_{T}},s_{T})\pi_{-i}(da^{-i}_{s_{T}},s_{T})\phi(s_{T},a^i_{s_{T}},a^{-i}_{s_{T}})
% \\&-\int_{\mathcal{A}_i}\int_{\mathcal{A}_{-i}}\pi'_i(da'^i_{s_{T}},s_{T})\pi_{-i}(da^{-i}_{s_{T}},s_{T})\phi(s_{T},a'^i_{s_{T}},a^{-i}_{s_{T}})\Bigg]
% \end{aligned}
% \\&\equiv \mathbb{E}_{s_{T}\sim P(\cdot|)}\left[\phi^{\boldsymbol{\pi}}(s_T)-\phi^{\pi'^{i},\pi^{-i}}(s_T)\right].
% \end{align*}
Since 
\begin{align*}
    &\int_{\mathcal{S}} ds_T\mathcal{P}^{\boldsymbol{\pi}}_{s_Ts_{T-1}}\boldsymbol{F}_i(a^{-i}_T)-\int_{\mathcal{S}} ds_T\mathcal{P}^{\boldsymbol{\pi'}}_{s_Ts_{T-1}}\boldsymbol{F}_i(a^{-i}_T)
\\&= \int_{\mathcal{A}_{-i}}\pi_i(da^i,s_{T-1})F_i(a^{-i})\int_{\mathcal{A}_i}\pi_{-i}(da^{-i},s_{T-1})\int_{\mathcal{S}} ds_T P(s_{T};s_{T-1},\boldsymbol{a}_T)
\\&\qquad\qquad\qquad- \int_{\mathcal{A}_{-i}}\pi'_i(da'^i,s_{T-1})F_i(a^{-i})\int_{\mathcal{A}_i}\pi_{-i}(da^{-i},s_{T-1}) \int_{\mathcal{S}} ds_T P(s_{T};s_{T-1},\boldsymbol{a'}_T)
\\&= \int_{\mathcal{A}_{-i}}\pi_{-i}(da^{-i},s_{T-1}) F_i(a^{-i})\left\{\int_{\mathcal{A}_i}\pi_i(da^i,s_{T-1})-\int_{\mathcal{A}_i}\pi'_i(da'^i,s_{T-1}) \right\}
\\&=0
\end{align*}
Hence, we find that
\begin{align*}\nonumber
&\mathbb{E}_{s_{T-1}\sim P(\cdot|)}\left[v^{\boldsymbol{\pi}}_{i,T-1}(s_{T-1})-v^{\boldsymbol{\pi'}}_{i,T-1}(s_{T-1})\right]
\\&=\mathbb{E}_{s_{T-1}\sim P(\cdot|)}\left[\boldsymbol{\phi}^{\boldsymbol{\pi}}(s_{T-1})-\boldsymbol{\phi}^{\boldsymbol{\pi'}}(s_{T-1})+\gamma\int_{\mathcal{S}} ds_T\mathcal{P}^{\boldsymbol{\pi}}_{s_Ts_{T-1}} \boldsymbol{\phi}^{\boldsymbol{\pi}}(s_{T})-\gamma\int_{\mathcal{S}} ds_T\mathcal{P}^{\boldsymbol{\pi'}}_{s_Ts_{T-1}} \boldsymbol{\phi}^{\boldsymbol{\pi'}}(s_{T})\right]
\\&\nonumber=\mathbb{E}_{s_{T-1}\sim P(\cdot|)}\left[\boldsymbol{\phi}^{\boldsymbol{\pi}}(s_{T-1})-\boldsymbol{\phi}^{\boldsymbol{\pi'}}(s_{T})+\gamma\mathbb{E}_{s_{T}\sim P(s_{T}|\cdot)}\left[\boldsymbol{\phi}^{\boldsymbol{\pi}}(s_{T-1})-\boldsymbol{\phi}^{\boldsymbol{\pi'}}(s_{T})\right]\right]
\\&=\mathbb{E}_{s_{T-1}\sim P(\cdot|)}\Bigg[
B_{T-1}^{\boldsymbol{\pi}}(s_{T-1})-B_{T-1}^{\boldsymbol{\pi'}}(s_{T-1})\Bigg],
\end{align*}
using the iterated law of expectations in the last line and where
\begin{align}
B^{\boldsymbol{\pi}}_T(s):=\mathbb{E}_{s_t\sim P,\pi_i,\pi_{-i}}\left[\sum_{t=0}^T\gamma^t\phi(s_t,\boldsymbol{a}_t)|s\equiv s_0]\right] \label{expression_for_B}.
\end{align} 

Hence, we have succeeded in proving that the expression (\ref{potential_relation_proof}) holds for $T-k$ when $k=1$.

Our next goal is to prove that the expression holds for any $0<k\leq T$.

\noindent Note that for any $T \geq k> 0$, we can write (\ref{expression_for_B}) as 
\begin{align*}
B^{\boldsymbol{\pi}}_{T-k}(s)=\mathbb{E}_{\pi_i,\pi_{-i}}\left[
\phi(s,\boldsymbol{a}_{k})+\gamma\int_{\mathcal{S}}ds' P(s';s,\boldsymbol{a}_{k})B_{T-(k+1)}^{\pi^i,\pi^{-i}}(s')\cdot 1_{k\leq T}\right].
\end{align*}

Now we consider the case when we evaluate the expression (\ref{potential_relation_proof}) for any $0<k\leq T$. Our inductive hypothesis is the expression holds for some $0<k\leq T$, that is for a $0<k\leq T$ we have that:
\begin{align}
 &\mathbb{E}_{s_{T-k}\sim P(\cdot|)}\left[v^{\boldsymbol{\pi}}_{i,k}(s_{T-k})-v^{\boldsymbol{\pi'}}_{i,k}(s_{T-k})\right]
=\mathbb{E}_{s_{T-k}\sim P(\cdot|)}\left[B^{\boldsymbol{\pi}}_{k}(s_{T-k})-B^{\boldsymbol{\pi'}}_k(s_{T-k})\right].\label{inductive_hyp}
\end{align}

It is easy to see that given \eqref{inductive_hyp} and Lemma \ref{dummy_coord_prop}, it must be the case that:
\begin{align}
 &\mathbb{E}_{s_{T-k}\sim P(\cdot|)}\left[v^{\boldsymbol{\pi}}_{i,k}(s_{T-k})\right]
=\mathbb{E}_{s_{T-k}\sim P(\cdot|)}\left[B^{\boldsymbol{\pi}}_{k}(s_{T-k})+G^{\pi_{-i}}_{i,k}(s_{T-k})\right].\label{inductive_hyp_dummy}
\end{align}
where $G^{\pi_{-i}}_{i,k}(s):=\mathbb{E}_{P,\pi_{-i}}\left[\sum_{t=0}^k\gamma^tF_{-i}(s,a_t^{-i})\right]$. 

Moreover, we recall that $F_{-i}$ satisfies the condition $F_{-i}(s,a^{-i}_t)=F_{-i}(s',a^{-i}_t)$, hence $G^{\pi_{-i}}_{i,k}(s)=G^{\pi_{-i}}_{i,k}(s')$ so from now on we drop the dependence on $s$ and write $G^{\pi_{-i}}_{i,k}$.

It remains to show that the expression holds for $k+1$ time steps prior to the end of the horizon. The result can be obtained using the dynamic programming principle and the base case ($k=1$) result, indeed we have that
\begin{align*}\nonumber
&\mathbb{E}_{s_{T-(k+1)}\sim P(\cdot|)}\left[v^{\boldsymbol{\pi}}_{i,k+1}(s_{T-(k+1)})-v^{\boldsymbol{\pi'}}_{i,k+1}(s_{T-(k+1)})\right]
\\&\begin{aligned}=\mathbb{E}_{s_{T-(k+1)}\sim P(\cdot|)}\Bigg[\mathcal{R}_i^{\boldsymbol{\pi}}(s_{T-(k+1)})&+\gamma\int_{\mathcal{S}}ds_{T-k}\mathcal{P}^{\boldsymbol{\pi}}_{s_{T-(k+1)}s_{T-k}}\boldsymbol{v}_{i,k}^{\boldsymbol{\pi}}(s_{T-k})
\\&-
\mathcal{R}_i^{\boldsymbol{\pi'}}(s_{T-(k+1)})-\gamma\int_{\mathcal{S}}ds_{T-k}\mathcal{P}^{\boldsymbol{\pi'}}_{s_{T-(k+1)}s_{T-k}}\boldsymbol{v}_{i,k}^{\boldsymbol{\pi'}}(s_{T-k})\Bigg]
\end{aligned}
\\&=\mathbb{E}_{s_{T-(k+1)}\sim P(\cdot|)}\left[\mathcal{R}_i^{\boldsymbol{\pi}}(s_{T-(k+1)})-\mathcal{R}_i^{\boldsymbol{\pi'}}(s_{T-(k+1)})\right]
\\&+\gamma\mathbb{E}_{s_{T-(k+1)}\sim P(\cdot|)}\Bigg[\int_{\mathcal{S}}ds\mathcal{P}^{\boldsymbol{\pi}}_{s_{T-(k+1)}s_{T-k}}\boldsymbol{v}_{i,k}^{\boldsymbol{\pi}}(s_{T-k})-\mathcal{P}^{\boldsymbol{\pi'}}_{s_{T-(k+1)}s_{T-k}}\boldsymbol{v}_{i,k}^{\boldsymbol{\pi'}}(s_{T-k})\Bigg].
\end{align*}
Studying the terms under the first expression, we observe that by construction, we have that:
\begin{align}
&\mathbb{E}_{s_{T-(k+1)}\sim P(\cdot|)}\Bigg[\mathcal{R}_i^{\boldsymbol{\pi}}(s_{T-(k+1)})-
\mathcal{R}_i^{\boldsymbol{\pi'}}(s_{T-(k+1)})\Bigg]
=\mathbb{E}_{s_{T-(k+1)}\sim P(\cdot|)}\left[\boldsymbol{\phi}^{\boldsymbol{\pi}}(s_{T-(k+1)})-\boldsymbol{\phi}^{\boldsymbol{\pi}'}(s_{T-(k+1)})\right].    \label{potential_condition_proof_result_cpt1}
\end{align}
We now study the terms within the second expectation. 

Using (\ref{inductive_hyp}) (i.e. the inductive hypothesis), we find that:
\begin{align*}
&\int_{\mathcal{S}} ds_{T-k}\mathcal{P}^{\boldsymbol{\pi}}_{s_{T-k}s_{T-k-1}}\boldsymbol{v}_{i,k}^{\boldsymbol{\pi}}(s_{T-k})
-\int_{\mathcal{S}} ds_{T-k}\mathcal{P}^{\boldsymbol{\pi'}}_{s_{T-k}s_{T-k-1}}\boldsymbol{v}_{i,k}^{\boldsymbol{\pi'}}(s_{T-k})
\\&=\int_{\mathcal{S}} ds_{T-k}\mathcal{P}^{\boldsymbol{\pi}}_{s_{T-k}s_{T-k-1}}\left[ \boldsymbol{B}_k^{\boldsymbol{\pi}}(s_{T-k})+\boldsymbol{G}_{i,k}^{\pi_{-i}}\right]
-\int_{\mathcal{S}} ds_{T-k}\mathcal{P}^{\boldsymbol{\pi'}}_{s_{T-k}s_{T-k-1}}\left[ \boldsymbol{B}_k^{\boldsymbol{\pi'}}(s_{T-k})+\boldsymbol{G}_{i,k}^{\pi_{-i}}\right]
\\&\begin{aligned}=\int_{\mathcal{S}} ds_{T-k}\mathcal{P}^{\boldsymbol{\pi}}_{s_{T-k}s_{T-k-1}} \boldsymbol{B}_k^{\boldsymbol{\pi}}(s_{T-k})&-\int_{\mathcal{S}} ds_{T-k}\mathcal{P}^{\boldsymbol{\pi'}}_{s_{T-k}s_{T-k-1}} \boldsymbol{B}_k^{\boldsymbol{\pi'}}(s_{T-k})\\&+\int_{\mathcal{S}} ds_{T-k}\mathcal{P}^{\boldsymbol{\pi}}_{s_{T-k}s_{T-k-1}}\boldsymbol{G}_{i,k}^{\pi_{-i}}-\int_{\mathcal{S}} ds_{T-k}\mathcal{P}^{\boldsymbol{\pi'}}_{s_{T-k}s_{T-k-1}}\boldsymbol{G}_{i,k}^{\pi_{-i}}
\end{aligned}
\end{align*}
We now observe that
\begin{align*}
&\int_{\mathcal{S}} ds_{T-k}\mathcal{P}^{\boldsymbol{\pi}}_{s_{T-k}s_{T-k-1}}\boldsymbol{G}_{i,k}^{\pi_{-i}}-\int_{\mathcal{S}} ds_{T-k}\mathcal{P}^{\boldsymbol{\pi'}}_{s_{T-k}s_{T-k-1}}\boldsymbol{G}_{i,k}^{\pi_{-i}}
\\&=\int_{\mathcal{S}}ds_{T-k}\int_{\boldsymbol{\mathcal{A}}} \Bigg[\pi_i(da^i_{T-(k+1)},s_{T-(k+1)})-\pi'_i(da^{i}_{T-(k+1)},s_{T-(k+1)})\Bigg]\\
&\qquad\qquad\qquad\cdot P(s_{T-k};s_{T-(k+1)},\boldsymbol{a}_{T-(k+1)})\pi_{-i}(da^{-i}_{T-(k+1)},s_{T-(k+1)})G^{\pi_{-i}}_{i,k}
\\&=\int_{\mathcal{S}}ds_{T-k}\int_{\mathcal{A}_{-i}}\pi_{-i}(da^{-i}_{T-(k+1)},s_{T-(k+1)})
\\&\qquad\qquad\qquad\qquad\cdot\left( P(s_{T-k};s_{T-(k+1)},\pi_i,a^{-i}_{T-(k+1)})-P(s_{T-k};s_{T-(k+1)},\pi'_i,a^{-i}_{T-(k+1)})\right)G^{\pi_{-i}}_{i,k}
\\&=0
\end{align*}
We now find that:
\begin{align}\nonumber
&\int_{\mathcal{S}} ds_{T-k}\mathcal{P}^{\boldsymbol{\pi}}_{s_{T-k}s_{T-k-1}}\boldsymbol{v}_{i,k}^{\boldsymbol{\pi}}(s_{T-k})
-\int_{\mathcal{S}} ds_{T-k}\mathcal{P}^{\boldsymbol{\pi'}}_{s_{T-k}s_{T-k-1}}\boldsymbol{v}_{i,k}^{\boldsymbol{\pi'}}(s_{T-k})
\\&=\int_{\mathcal{S}} ds_{T-k}\mathcal{P}^{\boldsymbol{\pi}}_{s_{T-k}s_{T-k-1}} \boldsymbol{B}_k^{\boldsymbol{\pi}}(s_{T-k})-\int_{\mathcal{S}} ds_{T-k}\mathcal{P}^{\boldsymbol{\pi'}}_{s_{T-k}s_{T-k-1}} \boldsymbol{B}_k^{\boldsymbol{\pi'}}(s_{T-k})
\label{potential_condition_proof_result_cpt2}
\end{align}

Now combining (\ref{potential_condition_proof_result_cpt1}) and (\ref{potential_condition_proof_result_cpt2}) leads to the fact that:
\begin{align*}\nonumber
&\mathbb{E}_{s_{T-(k+1)}\sim P(\cdot|)}\left[v^{\boldsymbol{\pi}}_{i,k+1}(s_{T-(k+1)})-v^{\boldsymbol{\pi'}}_{i,k+1}(s_{T-(k+1)})\right]
\\&\begin{aligned}=\mathbb{E}_{s_{T-(k+1)}\sim P(\cdot|)}\Bigg[\int_{\mathcal{S}} ds_{T-k}\mathcal{P}^{\boldsymbol{\pi}}_{s_{T-k}s_{T-k-1}} \boldsymbol{B}_k^{\boldsymbol{\pi}}(s_{T-k})&-\int_{\mathcal{S}} ds_{T-k}\mathcal{P}^{\boldsymbol{\pi'}}_{s_{T-k}s_{T-k-1}} \boldsymbol{B}_k^{\boldsymbol{\pi'}}(s_{T-k})\Bigg]
\\&+\mathbb{E}_{s_{T-(k+1)}\sim P(\cdot|)}\left[\boldsymbol{\phi}^{\boldsymbol{\pi}}(s_{T-(k+1)})-\boldsymbol{\phi}^{\boldsymbol{\pi}'}(s_{T-(k+1)})\right]
\end{aligned},
\end{align*}
from which we immediately deduce that
\begin{align*}\nonumber
\mathbb{E}_{s_{T-(k+1)}\sim P(\cdot|)}\left[v^{\boldsymbol{\pi}}_{i,k+1}(s_{T-(k+1)})-v^{\boldsymbol{\pi'}}_{i,k+1}(s_{T-(k+1)})\right]=\mathbb{E}_{s_{T-(k+1)}\sim P(\cdot|)}\left[B^{\boldsymbol{\pi}}_{k+1}(s_{T-(k+1)})-B^{\boldsymbol{\pi'}}_{k+1}(s_{T-(k+1)})\right],
\end{align*}
where
\begin{align*}
B^{\boldsymbol{\pi}}_{k}(s)=\mathbb{E}_{\pi_i,\pi_{-i}}\left[
\phi(s_{k},\boldsymbol{a}_{k})+\gamma\int_{\mathcal{S}}ds' P(s';s,\boldsymbol{a}_{k})B_{k-1}^{\pi^i,\pi^{-i}}(s')\right],
\end{align*}
from which we deduce the result.
\end{proof}
Thus far we have established the relation (\ref{dpg_equality_finite}) holds only for the finite horizon case. We now extend the coverage to the infinite horizon case in which we can recover the use of stationary strategies. Before doing so, we require the following results:

\begin{lemma}\label{convergence_of_B}
For any $t' <\infty$, define by $B^{\boldsymbol{\pi}}_{t'}:\mathbb{N}\times\Pi_i\times\Pi_{-i}\times \mathcal{S}\to \mathbb{R}$ the following function:
\begin{align*}
B^{\boldsymbol{\pi}}_{t'}(s):=\mathbb{E}_{s_t\sim P,\pi_i,\pi_{-i}}\left[\sum_{t=0}^T\gamma^t\phi(s_t,\boldsymbol{a}_t)|s\equiv s_0\right].
\end{align*} 
then $\exists B^{\boldsymbol{\pi}}:\Pi_i\times\Pi_{-i}\times \mathcal{S}\to \mathbb{R}$ s.t. $\forall s\in \mathcal{S}$ and for any $\pi_i\in\Pi_i,\pi_{-i}\in\Pi_{-i}$,
\begin{align*}
    \underset{t\to\infty}{\lim}B^{\boldsymbol{\pi}}_{t}(s)=B^{\boldsymbol{\pi}}(s),
\end{align*}
where for any $\pi_i\in\Pi_i,\pi_{-i}\in\Pi_{-i}$, the function $B^{\boldsymbol{\pi}}$ is given by:\\
$
    B^{\boldsymbol{\pi}}(s):=\mathbb{E}_{s_t\sim P,\pi_i,\pi_{-i}}\left[\sum_{t=0}^\infty\gamma^t\phi(s_t,\boldsymbol{a}_t)|s\equiv s_0\right]
$.
\end{lemma}

\begin{proof}
We prove the result by showing that the sequence $B^{\boldsymbol{\pi}}_{n},B^{\boldsymbol{\pi}}_{n+1},\ldots$ converges uniformly, that is the sequence is a Cauchy sequence. In particular, we show that $\forall \epsilon>0$, $\exists T'>0$ s.th. $\forall t',t''>T'$ and for any $\pi_i\in\Pi_i,\pi_{-i}\in\Pi_{-i}$
\begin{align*}
    \left\|B^{\boldsymbol{\pi}}_{t'}-B^{\boldsymbol{\pi}}_{t''}\right\|<\epsilon.
\end{align*}
Firstly, we deduce that the function $\phi$ is bounded since each $R_i$ is bounded also (c.f. (\ref{line_integral_potential_equation})). Now w.log., consider the case when $t'\geq t''$. We begin by observing the fact that
\begin{align*}
&B^{\boldsymbol{\pi}}_{t'}(s)-B^{\boldsymbol{\pi}}_{t''}(s)
\\&=\mathbb{E}_{s_t\sim P(\cdot;s_{t-1},a_{t-1}),\pi_i,\pi_{-i}}\left[\sum_{t=0}^{t'}\gamma^t\phi_t(s_t,\boldsymbol{a}_t)-\sum_{t=0}^{t''}\gamma^t\phi_t(s_t,\boldsymbol{a}_t)\right]
\\&=\mathbb{E}_{s_t\sim P(\cdot;s_{t-1},a_{t-1}),\pi_i,\pi_{-i}}\left[\sum_{t=t''}^{t'}\gamma^t\phi_t(s_t,\boldsymbol{a}_t)\right].
\end{align*}
Hence, we find that

\begin{align*}
&\; \left| B^{\boldsymbol{\pi}}_{t'}(s)-B^{\boldsymbol{\pi}}_{t''}(s)\right|
\\&=\left|\mathbb{E}_{s_t\sim P(\cdot;s_{t-1},a_{t-1}),\pi_i,\pi_{-i}}\left[\sum_{t=t''}^{t'}\gamma^t\phi_t(s_t,\boldsymbol{a}_t)\right]\right|
\\&
\leq \sum_{t=t''}^{t'}\gamma^t\|\phi\|_{\infty}\leq |\gamma|\frac{\left|\gamma^{t''}-\gamma^{t'}\right|}{1-\gamma}\|\phi\|_{\infty}
\\&\leq |\gamma^{t''}|\frac{\left|1-\gamma^{t'-t''}\right|}{1-\gamma}\|\phi\|_{\infty}
\\&\leq \frac{|\gamma^{t''}|}{1-\gamma}\|\phi\|_{\infty}
=e^{t''\ln\gamma}\frac{\|\phi\|_{\infty}}{1-\gamma}
\\&= e^{-t''|\ln\gamma|}\left(\frac{\|\phi\|_{\infty}}{1-\gamma}\right)\leq e^{-T'|\ln\gamma|}\left(\frac{\|\phi\|_{\infty}}{1-\gamma}\right),
\end{align*}
using Cauchy-Schwarz and since $t'\geq t''> T'$ and $\gamma \in [0,1[$.  
The inequality of the proposition is true whenever $T'$ is chosen to satisfy
\begin{align*}
T'\geq \left|\ln{(\epsilon)}(\ln{(\gamma)}\left(\frac{\|\phi\|_{\infty}}{1-\gamma}\right)^{-1}\right|,    
\end{align*}
hence the result is proven.
\end{proof}
We are now in a position to extend the dynamic potential property (\ref{dpg_equality_finite}) to the infinite horizon case:
\begin{proof}
The result is proven by contradiction.

\noindent To this end, let us firstly assume there exists a constant $c\neq 0$ s.th.
\begin{align*}\nonumber
&\mathbb{E}_{s\sim P(\cdot|)}\left[v^{\boldsymbol{\pi}}_i(s)-v^{\boldsymbol{\pi'}}_i(s)\right]
-\mathbb{E}_{s\sim P(\cdot|)}\left[B^{\boldsymbol{\pi}}_i(s)-B^{\boldsymbol{\pi'}}_i(s)\right] =c.
\end{align*}

Let us now define the following quantities for any $s\in \mathcal{S}$ and for each $\pi_i\in\Pi_i$ and $\pi_{-i}\in\Pi_{-i}$ and $\forall\theta\in \Theta, \forall i\in\mathcal{N}$:
\begin{align*}
v_{i,T'}^{\boldsymbol{\pi}}(s)
:=\sum_{t=1}^{T'}\int_{\mathcal{S}}ds_{j+1}\mu(s_0)\pi_i(a^i_0,s_0)\pi_{-i}(a^{-i}_0,s_0)\prod_{j=1}^{t-1}\gamma^jP(s_{j+1}|s_j,a^i_j,a^{-i}_j)&
\\\cdot\pi_i(a^i_j|s_j)\pi_{-i}(a^{-i}_j|s_j)R_i(s_j,a^i_j,a^{-i}_j)&,     \end{align*}
and 
\begin{align*}
&B_{T'}^{\boldsymbol{\pi}}(s)
\\&\begin{aligned}:=\sum_{t=1}^{T'}\int_{\mathcal{S}}ds_{j+1}\mu(s_0)\pi_i(a^i_0,s_0)\pi_{-i}(a^{-i}_0,s_0)\prod_{j=1}^{t-1}\gamma^jP(s_{j+1}|s_j,a^i_j,a^{-i}_j)\pi_i(a^i_j|s_j)\pi_{-i}(a^{-i}_j|s_j)&
\\\cdot\phi(s_j,a^i_j,a^{-i}_j)&,
\end{aligned}\end{align*}
so that the quantity $v_{i,T'}^{\boldsymbol{\pi}}(s)$ measures the expected cumulative return until the point $T'<\infty$.

Hence, we straightforwardly deduce that
\begin{align*}\nonumber
&v_{i}^{\boldsymbol{\pi}}(s)\equiv v_{i,\infty}^{\boldsymbol{\pi}}(s)
\\&\begin{aligned}=
v_{i,T'}^{\boldsymbol{\pi}}(s)
+\gamma^{T'}\int_{\mathcal{S}}ds_{j+1}\mu(s_0)\pi_i(a^i_0,s_0)\pi_{-i}(a^{-i}_0,s_0)\prod_{j=1}^{T'}P(s_{j+1}|s_j,a^i_j,a^{-i}_j)\pi_i(a^i_j|s_j)&
\\\cdot\pi_{-i}(a^{-i}_j|s_j)v_{i}^{\boldsymbol{\pi}}(s_{T'}).&     
\end{aligned}
\end{align*}
Our first task is to establish that the quantity $\left|\underset{t\to\infty}{\lim}\mathbb{E}_{s\sim P(\cdot|)}\left[B^{\boldsymbol{\pi}}_{i,t}(s)-B^{\boldsymbol{\pi'}}_{i,t}(s)\right]\right|$ is in fact, well-defined for any $s\in \mathcal{S}$ and $\forall i \in\mathcal{N}$.

This is true since by (\ref{inductive_hyp}) for any $t>0$ we have that
\begin{align}
&\left|\mathbb{E}_{s\sim P(\cdot|)}\left[B^{\boldsymbol{\pi}}_{i,t}(s)-B^{\boldsymbol{\pi'}}_{i,t}(s)\right]\right|=
\left|\mathbb{E}_{s\sim P(\cdot|)}\left[v^{\boldsymbol{\pi}}_{i,t}(s)-v^{\boldsymbol{\pi'}}_{i,t}(s)\right]\right|, \label{potential_condition_proof}
\end{align}
and hence we have that
\begin{align*}\nonumber
&\left|\mathbb{E}_{s\sim P(\cdot|)}\left[B^{\boldsymbol{\pi}}(s)-B^{\boldsymbol{\pi'}}(s)\right]\right|
<\infty.
\end{align*}

To see this, we firstly observe that by the boundedness of $R_i$, $\exists c>0$ s.th. $\forall t\in \mathbb{N}, \forall i\in\mathcal{N}$ and for any $\pi_i\in\Pi_i,\pi_{-i}\in\Pi_{-i}$
\begin{align*}
    \left|v^{\boldsymbol{\pi}}_{i,t}(s)-v^{\boldsymbol{\pi'}}_{i,t}(s)\right|<c.
\end{align*}
This is true since for any $k<\infty$ we have
\begin{align*}
&\qquad\qquad\qquad \qquad v_{i,k}^{\boldsymbol{\pi}}(s)-v_{i,k}^{\boldsymbol{\pi'}}(s)
\\&=\mathbb{E}_{s_t\sim P,\pi_i,\pi_{-i}}\left[\sum_{t=0}^{k}\gamma^tR_i(s_t,\boldsymbol{a}_t)\right]-\mathbb{E}_{s_t\sim P,\pi'_i,\pi_{-i}}\left[\sum_{t=0}^{k}\gamma^tR_i(s_t,\boldsymbol{a}_t)\right]
\\&\leq\left|\mathbb{E}_{s_t\sim P,\pi_i,\pi_{-i}}\left[\sum_{t=0}^{k}\gamma^tR_i(s_t,\boldsymbol{a}_t)\right]-\mathbb{E}_{s_t\sim P,\pi'_i,\pi_{-i}}\left[\sum_{t=0}^{k}\gamma^tR_i(s_t,\boldsymbol{a}_t)\right]\right|
\\&\leq2\sum_{t=0}^{k}\gamma^t\|R_i\|_{\infty} 
=2\frac{1-\gamma^k}{1-\gamma}\|R_i\|_{\infty}.
\end{align*}

Therefore, by the bounded convergence theorem we have that
\begin{align}
    \underset{t\to\infty}{\lim}\left|\mathbb{E}_{s\sim P(\cdot|)}\left[v^{\boldsymbol{\pi}}_{i,t}(s)-v^{\boldsymbol{\pi'}}_{i,t}(s)\right]\right|<\infty.\label{boundedness_v_limit}
\end{align}
Now, using (\ref{potential_condition_proof}), we deduce that for any $\epsilon>0$, the following statement holds:
\begin{align*}\nonumber
&\left|\mathbb{E}_{s\sim P(\cdot|)}\left[B^{\boldsymbol{\pi}}_{i,t}(s)-B^{\boldsymbol{\pi'}}_{i,t}(s)\right]\right|<
\left|\mathbb{E}_{s\sim P(\cdot|)}\left[v^{\boldsymbol{\pi}}_{i,t}(s)-v^{\boldsymbol{\pi'}}_{i,t}(s)\right]\right|+\epsilon,
\end{align*}
which after taking the limit as $t\to\infty$ and using (\ref{boundedness_v_limit}), Lemma \ref{convergence_of_B} and the dominated convergence theorem, we find that
\begin{align*}\nonumber
&\underset{t\to\infty}{\lim}\left|\mathbb{E}_{s\sim P(\cdot|)}\left[B^{\boldsymbol{\pi}}_{i,t}(s)-B^{\boldsymbol{\pi'}}_{i,t}(s)\right]\right|<\infty.
\end{align*}
Next we observe that:

\begin{align*}\nonumber
c&=\mathbb{E}_{s\sim P(\cdot|)}\left[\left(v_{i}^{\boldsymbol{\pi}}-v_{i}^{\boldsymbol{\pi'}}\right)(s)\right]-\mathbb{E}_{s\sim P(\cdot|)}\left[\left(B^{\boldsymbol{\pi}}-B^{\boldsymbol{\pi'}}\right)(s)\right]
\\&\nonumber
=\mathbb{E}_{s\sim P(\cdot|)}\left[\left(v_{i,T'}^{\boldsymbol{\pi}}-v_{i,T'}^{\boldsymbol{\pi'}}\right)(s)\right]-\mathbb{E}_{s\sim P(\cdot|)}\left[\left(B_{T'}^{\boldsymbol{\pi}}-B_{T'}^{\boldsymbol{\pi'}}\right)(s)\right]
\\&\begin{aligned}
+\gamma^{T'}\mathbb{E}_{s_{T'}\sim P(\cdot|)}\Bigg[\int_{\mathcal{S}}ds_{j+1}\mu(s_0)\pi_i(a^i_0,s_0)\pi_{-i}(a^{-i}_0,s_0)\prod_{j=1}^{T'}P(s_{j+1}|s_j,a^i_j,a^{-i}_j)\pi_i(a^i_j|s_j)\pi_{-i}(a^{-i}_j|s_j)
\\\cdot
\left(v_{i}^{\boldsymbol{\pi}}(s_{T'})-B^{\boldsymbol{\pi}}(s_{T'})\right)&\nonumber
\end{aligned}
\\&\begin{aligned}
+\int_{\mathcal{S}}ds_{j+1}\mu(s_0)\pi'_i(a'^i_0,s_0)\pi_{-i}(a^{-i}_0,s_0)\prod_{j=1}^{T'}P(s_{j+1}|s_j,a'^i_j,a^{-i}_j)\pi_i(a'^i_j|s_j)\pi_{-i}(a^{-i}_j|s_j)
\\\cdot
\left(v_{i}^{\boldsymbol{\pi'}}(s_{T'})-B^{\boldsymbol{\pi'}}(s_{T'})\right)\Bigg].&
\end{aligned}
\end{align*}
Considering the last expectation and its coefficient and denoting it by $\kappa$, we observe the following bound:
\begin{align*}
|\kappa|\leq 2\gamma^{T'}\left(\|v_{i}\|+\|B\|\right).
\end{align*}
Since we can choose $T'$ freely and $\gamma \in ]0,1[$, we can choose $T'$ to be sufficiently large so that
\begin{align*}
\gamma^{T'}\left(\|v_{i}\|+\|B\|\right)<\frac{1}{4}|c|.
\end{align*}
This then implies that
\begin{align*}
\left|\mathbb{E}_{s\sim P(\cdot|)}\left[\left(v_{i,T'}^{\boldsymbol{\pi}}-v_{i,T'}^{\boldsymbol{\pi'}}\right)(s)-\left(B_{T'}^{\boldsymbol{\pi}}-B_{T'}^{\boldsymbol{\pi'}}\right)(s)\right]\right|>\frac{1}{2}c,    
\end{align*}
which is a contradiction since we have proven that for any finite $T'$ it is the case that
\begin{align*}
\mathbb{E}_{s\sim P(\cdot|)}\left[\left(v_{i,T'}^{\boldsymbol{\pi}}-v_{i,T'}^{\boldsymbol{\pi'}}\right)(s)-\left(B_{T'}^{\boldsymbol{\pi}}-B_{T'}^{\boldsymbol{\pi'}}\right)(s)\right]=0,
\end{align*}
and hence we deduce the thesis.
\end{proof}
\begin{proof}[Proof of Lemma \ref{difference_constant_lemma}]
The result is proven after a straightforward extension of the static case (Lemma 2.7. in \cite{monderer1996potential}).
\end{proof}
\begin{proposition}\label{reduction_prop}
There exists a function $B:\mathcal{S}\times \boldsymbol{\Pi}\to\mathbb{R}$ such that  $\forall s\in S$ we have that 
\begin{align*}
\boldsymbol{\pi}\in \underset{{\boldsymbol{\pi'}}\in\boldsymbol{\Pi}}{\arg\sup}\; B^{{\boldsymbol{\pi'}}}(s)\implies \boldsymbol{\pi}\in NE\{\mathcal{G}\}.
\end{align*}
\end{proposition}
\begin{proof}[Proof of Prop. \ref{reduction_prop}]
We do the proof by contradiction. Let $\boldsymbol{\pi}=(\pi_1,\ldots,\pi_N)\in \underset{\boldsymbol{\pi'}\in\boldsymbol{\Pi}}{\arg\sup}\; v^{\boldsymbol{\pi'}}(s)$. Let us now therefore assume that $\boldsymbol{\pi}\notin NE\{\mathcal{G}\}$, hence there exists some other strategy profile $\boldsymbol{\pi'}=(\pi_1,\ldots,\pi_i',\ldots,\pi_N)$ which contains at least one profitable deviation by one of the agents so that $\pi_i'\neq \pi_i$ for $i\in\mathcal{N}$ i.e. $v^{\boldsymbol{\pi'}}_i(s)> v^{\boldsymbol{\pi}}_i(s)$ (using the preservation of signs of integration). Prop. \ref{dpg_proposition} however implies that $v_i^{\boldsymbol{\pi'}}(s)-v_i^{\boldsymbol{\pi}}(s)>0$ which is a contradiction since $\boldsymbol{\pi}$ is a maximum of $B$.  
\end{proof}
\begin{proof}[Proof of Theorem \ref{reduction_theorem}]
Combining Prop. \ref{reduction_prop} with Prop. \ref{dpg_proposition} proves Theorem \ref{reduction_theorem}.
\end{proof}
\begin{proof}[Proof of Prop. \ref{reparam_pg_ode}]
% We consider a set of parameterised policy classes $\left(\Pi_i(\boldsymbol{\eta})\right)_{\boldsymbol{\eta}\in\mathrm{T},i\in\mathcal{N}}$. 
Since the functions $(R_{i})_{i\in\mathcal{N}}$ are differentiable in the action inputs,  we first we note the following
$ 
R_i(s,a^i,a^{-i})- R_i(s,a'^i,a^{-i})
=\int_{a^i}^{a'^i}\frac{\partial R_i(s,a,a^{-i})}{\partial a}da
$
and $
\phi(s,a^i,a^{-i})- \phi(s,a'^i,a^{-i})
=\int_{a^i}^{a'^i}\frac{\partial \phi(s,a,a^{-i})}{\partial a}da
$.
We then deduce that $
 \frac{\partial R_i(s,a,a^{-i})}{\partial a}=   \frac{\partial \phi(s,a,a^{-i})}{\partial a}$.
Considering actions sampled from stochastic policies, we find that \begin{align*}\nonumber
&\mathbb{E}_{\pi_i(\boldsymbol{\eta}_i),\pi_{-i}(\boldsymbol{\eta}_{-i})}\left[R_i(s,a^i,a^{-i})\right]- \mathbb{E}_{\pi_i(\boldsymbol{\eta}'_{i}),\pi_{-i}(\boldsymbol{\eta}_{-i})}\left[R_i(s,a'^i,a^{-i})\right]
\\&=\mathbb{E}_{\pi_i(\boldsymbol{\eta}_{i}),\pi_{-i}(\boldsymbol{\eta}_{-i})}\left[\phi(s,a^i,a^{-i})\right]- \mathbb{E}_{\pi'_i(\boldsymbol{\eta}'_{i}),\pi_{-i}(\boldsymbol{\eta}_{-i})}\left[\phi(s,a'^i,a^{-i})\right].
\end{align*}
Now suppose $a=h(s,\boldsymbol{\eta}^i)$ then
\begin{align}
&\int_{a^i}^{a'^i}\frac{\partial R_i(s,a,a^{-i})}{\partial a}da=\int_{\boldsymbol{\eta}_i=h^{-    1}(a^i)}^{\boldsymbol{\eta}'_i=h^{-1}(a'^i)}\frac{\partial R_i(s,h(s,\boldsymbol{\eta}^i),a^{-i})}{\partial a}\frac{dh}{d\boldsymbol{\eta}}d\boldsymbol{\eta} \label{deterministic_difference_integral}
\end{align}
Similarly we find that for $a^{i}\sim \pi_{i,\epsilon}(\cdot,|s,\eta)$ we have that
\begin{align}
&\mathbb{E}_{\pi_i(\boldsymbol{\eta}_i),\pi_{-i}(\boldsymbol{\eta}_{-i})}\left[R_i(s,a^i,a^{-i})\right]- \mathbb{E}_{\pi_i(\boldsymbol{\eta}'_{i}),\pi_{-i}(\boldsymbol{\eta}_{-i})}\left[R_i(s,a'^i,a^{-i})\right]
\nonumber\\&=\int_{\mathcal{A}_{-i}}\pi_{-i}(da^{-i},\boldsymbol{\eta}_{-i})\int_{\boldsymbol{\eta}_i(a^i)}^{\boldsymbol{\eta}'_i(a'^i)}\int_{\mathcal{A}_i}\frac{\partial}{\partial \boldsymbol{\eta}_i}\pi_i(da^i,s;\boldsymbol{\eta}_i)\frac{\partial R_i\left(s,\pi_{i,\epsilon}(\cdot,|s,\eta),a^{-i}\right)}{\partial a^i}d\boldsymbol{\eta}_i
\\&\nonumber=\int_{\mathcal{A}_{-i}}\pi_{-i}(da^{-i},\boldsymbol{\eta}_{-i})\int_{\boldsymbol{\eta}_i(a^i)}^{\boldsymbol{\eta}'_i(a'^i)}\int_{\mathcal{A}_i}\frac{\partial}{\partial \boldsymbol{\eta}_i}\pi_i(da^i,s;\boldsymbol{\eta}_i)\frac{\partial R_i(s,a^i,a^{-i})}{\partial a^i}d\boldsymbol{\eta}_i
\\&=\int_{\boldsymbol{\eta}_i(a^i)}^{\boldsymbol{\eta}'_i(a'^i)}\int_{\mathcal{A}_{-i}}\pi_{-i}(da^{-i},\boldsymbol{\eta}_{-i})\int_{\mathcal{A}_i}\frac{\partial}{\partial \boldsymbol{\eta}_i}\pi_i(da^i,s;\boldsymbol{\eta}_i)\frac{\partial R_i(s,a^i,a^{-i})}{\partial a^i}d\boldsymbol{\eta}_i
\label{derivative_reward_param}\end{align}
By the same reasoning as above we find that 
\begin{align}
&\nonumber\mathbb{E}_{\pi_i(\boldsymbol{\eta}_i),\pi_{-i}(\boldsymbol{\eta}_{-i})}\left[\phi(s,a^i,a^{-i})\right]- \mathbb{E}_{\pi_i(\boldsymbol{\eta}'_{i}),\pi_{-i}(\boldsymbol{\eta}_{-i})}\left[\phi(s,a'^i,a^{-i})\right]
\\&=\int_{\boldsymbol{\eta}_i(a^i)}^{\boldsymbol{\eta}'_i(a'^i)}\int_{\mathcal{A}_{-i}}\pi_{-i}(da^{-i},\boldsymbol{\eta}_{-i})\int_{\mathcal{A}_i}\frac{\partial}{\partial \boldsymbol{\eta}_i}\pi_i(da^i,s;\boldsymbol{\eta}_i)\frac{\partial \phi(s,a^i,a^{-i})}{\partial a^i}d\boldsymbol{\eta}_i    \label{derivative_phi_param}
\end{align}

Putting \eqref{derivative_reward_param} and \eqref{derivative_phi_param} together, we deduce that
\begin{align*}
&\int_{\boldsymbol{\eta}_i(a^i)}^{\boldsymbol{\eta}'_i(a'^i)}\int_{\mathcal{A}_{-i}}\pi_{-i}(da^{-i},\boldsymbol{\eta}_{-i})\int_{\mathcal{A}_i}\frac{\partial}{\partial \boldsymbol{\eta}_i}\pi_i(da^i,s;\boldsymbol{\eta}_i)\frac{\partial R_i(s,a^i,a^{-i})}{\partial a^i}d\boldsymbol{\eta}_i
\\&=
\int_{\boldsymbol{\eta}_i(a^i)}^{\boldsymbol{\eta}'_i(a'^i)}\int_{\mathcal{A}_{-i}}\pi_{-i}(da^{-i},\boldsymbol{\eta}_{-i})\int_{\mathcal{A}_i}\frac{\partial}{\partial \boldsymbol{\eta}_i}\pi_i(da^i,s;\boldsymbol{\eta}_i)\frac{\partial \phi(s,a^i,a^{-i})}{\partial a^i}d\boldsymbol{\eta}_i
\end{align*}

Which implies that
\begin{align*}
&\int_{\mathcal{A}_{-i}}\int_{\mathcal{A}_i}\pi_{-i}(da^{-i},\boldsymbol{\eta}_{-i})\frac{\partial}{\partial \boldsymbol{\eta}_i}\pi_i(da^i,s;\boldsymbol{\eta}_i)\frac{\partial R_i(s,a^i,a^{-i})}{\partial a^i}
\\&=
\int_{\mathcal{A}_{-i}}\int_{\mathcal{A}_i}\pi_{-i}(da^{-i},\boldsymbol{\eta}_{-i})\frac{\partial}{\partial \boldsymbol{\eta}_i}\pi_i(da^i,s;\boldsymbol{\eta}_i)\frac{\partial \phi(s,a^i,a^{-i})}{\partial a^i}
\end{align*}
Moreover we also find that
\begin{align*}
&
\int_{\mathcal{A}_{-i}}\pi_{-i}(da^{-i},\boldsymbol{\eta}_{-i})\int_{\mathcal{A}_i}\frac{\partial}{\partial \boldsymbol{\eta}_i}\pi_i(da^i,s;\boldsymbol{\eta}_i)\frac{\partial R_i(s,a^i,a^{-i})}{\partial a^i}d\boldsymbol{\eta}_i
\\&=
\int_{\mathcal{A}_i}\int_{\mathcal{A}_{-i}}\pi_{-i}(da^{-i}|s;\boldsymbol{\eta}^{-i})\pi_{i}(da^{i}|s;\boldsymbol{\eta}^{i})\frac{\partial}{\partial \boldsymbol{\eta}^i}\ln{[\pi_i(da^i|s;\boldsymbol{\eta}^i)]}\frac{\partial R_i(s,a^i,a^{-i})}{\partial a^i}
\\&=\mathbb{E}_{\boldsymbol{\pi}(\boldsymbol{\boldsymbol{\eta}})}\left[\frac{\partial}{\partial \boldsymbol{\eta}^i}\ln{[\pi_i(a^i|s;\boldsymbol{\eta}^i)]}\frac{\partial R_i(s,a^i,a^{-i})}{\partial a^i}\right].
\end{align*}
Hence, using the linearity of the expectation and the derivative we arrive at
\begin{align}
&\mathbb{E}_{(a_i,a_{-i})\sim(\pi_i,\pi_{-i})}\left[\frac{\partial}{\partial \boldsymbol{\eta}^i}\ln{[\pi_i(a^i|s;\boldsymbol{\eta}^i)]}\left(\frac{\partial R_i(s,a^i,a^{-i})}{\partial a^i}-\frac{\partial \phi(s,a^i,a^{-i})}{\partial a^i}\right)\right]=\boldsymbol{0}
\end{align}

% If we now suppose that each agent's policy is drawn from a parameterised policy space.....
% then $\forall s\in \mathcal{S}, \forall a^i\in\mathcal{A}_i, \forall\boldsymbol{\pi}\in\boldsymbol{\Pi}_{ps}$: $\mathbb{E}_{(a^i,a^{-i})\sim (\pi_i,\pi_{-i})}\left[\frac{\partial f_i}{\partial \pi_i}\frac{\partial}{\partial f_i}\ln{[f_i(a^i|s;\boldsymbol{\eta}_i)]}\left(R_i(s,(a^i,a^{-i})-\phi(s,(a^i,a^{-i})\right)\right]=\boldsymbol{0}$.
% Since $\frac{\partial f}{\partial \pi_i}\not\equiv 0$ then since the above holds for any parameterisation including that induced by the linear map, by the iterated law of expectations we can write the following condition:
% \begin{align}
% \mathbb{E}_{(a^i,a^{-i})\sim(\pi_i(\boldsymbol{\eta}^i),\pi_{-i}(\boldsymbol{\eta}^{-i}))}\left[\frac{\partial}{\partial \boldsymbol{\eta}^i}\ln{[\pi_i(a^i|s;\boldsymbol{\eta}^i)]}\left(R_i(s,(a^i,a^{-i}))-\phi(s,(a^i,a^{-i}))\right)\right]=\boldsymbol{0}, 
% \end{align}

In a similar way we observe that for any c-SPG in which the state transitive assumption holds, we have that $
R_i(s',a^i,a^{-i})- R_i(s,a^i,a^{-i})
=\int_{s'}^{s}\frac{\partial R_i(s,a,a^{-i})}{\partial s}ds$
and $
\phi(s',a^i,a^{-i})- \phi(s',a^i,a^{-i})
=\int_{s'}^{s}\frac{\partial \phi(s,a,a^{-i})}{\partial s}ds$.
We then find that $
 \frac{\partial R_i(s,a,a^{-i})}{\partial s}=   \frac{\partial \phi(s,a,a^{-i})}{\partial s}$.
By identical reasoning as above we deduce that 
\begin{align}
&\mathbb{E}_{(a_i,a_{-i})\sim(\pi_i,\pi_{-i})}\left[\frac{\partial}{\partial \boldsymbol{\eta}^i}\ln{[\pi_i(a^i|s;\boldsymbol{\eta}^i)]}\left(\frac{\partial R_i(s,a^i,a^{-i})}{\partial s}-\frac{\partial \phi(s,a^i,a^{-i})}{\partial s}\right)\right]=\boldsymbol{0}.
\end{align}
Putting the two statements together leads to the expression:
\begin{align}
&\mathbb{E}_{(a_i,a_{-i})\sim(\pi_i,\pi_{-i})}\left[\frac{\partial}{\partial \boldsymbol{\eta}^i}\ln{[\pi_i(a^i|s;\boldsymbol{\eta}^i)]}D[R_i,\phi](s,a^i,a^{-i})\right]=\boldsymbol{0}
\end{align}
which concludes the proof.
\end{proof}
% \begin{lemma}\label{expected_potentiality_lemma}
% For any c-SPG and for any $\pi_i,\pi_i'\in\Pi_i$, $\pi_{-i}\in\Pi_{-i}$ the condition holds in expectation over policies, that is:

% \begin{align*}
% &\mathbb{E}_{(\pi_i,\pi_{-i})}\left[R_i(s_t,(a^i_{t},a^{-i}_{t}))\right]- \mathbb{E}_{(\pi'_i,\pi_{-i})}\left[R_i(s_t,(a'^i_{t},a^{-i}_{t}))\right]
% \\&=\mathbb{E}_{(\pi_i,\pi_{-i})}\left[\phi(s_t,(a^i_{t},a^{-i}_{t}))\right]- \mathbb{E}_{(\pi'_i,\pi_{-i})}\left[\phi(s_t,(a'^i_{t},a^{-i}_{t}))\right]
% \end{align*}
% Moreover the conditions are equivalent in pure strategies.
% \end{lemma}
\begin{lemma}\label{expected_potentiality_lemma}
For any c-SPG the following expression holds $\forall (\boldsymbol{\eta}^i,\boldsymbol{\eta}^{-i}) \in \boldsymbol{E}^{ps}, \forall s\in \mathcal{S}$ 
\begin{align}
&\mathbb{E}_{\pi_i(\boldsymbol{\eta}_i),\pi_{-i}(\boldsymbol{\eta}_{-i})}\left[R_i\right]- \mathbb{E}_{\pi_i(\boldsymbol{\eta}'_{i}),\pi_{-i}(\boldsymbol{\eta}_{-i})}\left[R_i\right]=\mathbb{E}_{\pi_i(\boldsymbol{\eta}_{i}),\pi_{-i}(\boldsymbol{\eta}_{-i})}\left[\phi\right]- \mathbb{E}_{\pi'_i(\boldsymbol{\eta}'_{i}),\pi_{-i}(\boldsymbol{\eta}_{-i})}\left[\phi\right]. \label{parameterisation_potential_condition}
\end{align}
\end{lemma}
\begin{proof}[Proof of Lemma \ref{expected_potentiality_lemma}]
The forward implication is straightforward. 

Indeed, assume that \eqref{potential_condition_static} holds, that is:
\begin{align}
R_i(s_t,(a^i_{t},a^{-i}_{t}))- R_i(s_t,(a'^i_{t},a^{-i}_{t}))=\phi(s_t,(a^i_{t},a^{-i}_{t}))- \phi(s_t,(a'^i_{t},a^{-i}_{t}))\label{pot_game_static_notes}
\end{align}
Define by $\Delta F[s_t](a^i_{t},a'^i_t),a^{-i}_{t}):= F(s_t,(a^i_{t},a^{-i}_{t}))- F(s_t,(a'^i_{t},a^{-i}_{t}))$ for any $a'^i_{t}\in \mathcal{A}_i$ then for any $F:\mathcal{S}\times\mathcal{A}_i\times\mathcal{A}_{-i}\to \mathbb{R}$ ($F\in\mathcal{H}$) we have that 
\begin{align*}
\\&\int_{a_i\in\mathcal{A}_i}\int_{a'_i\in\mathcal{A}_i}\int_{a_{-i}\in\mathcal{A}_{-i}}\pi_{i}(da^i|\cdot)\pi_{i}(da'^i|\cdot)\pi_{-i}(da^{-i}|\cdot)\Delta F[s_t](a^i_{t},a'^i_t),a^{-i}_{t})
\\&=
\int_{a_i\in\mathcal{A}_i}\int_{a_{-i}\in\mathcal{A}_{-i}}\pi_{i}(da^i|\cdot)\pi_{-i}(da^{-i}|\cdot)F(a^i_{t},a'^i_t),a^{-i}_{t})- \int_{a'_i\in\mathcal{A}_i}\int_{a_{-i}\in\mathcal{A}_{-i}}\pi_{i}(da'^i|\cdot)\pi_{-i}(da^{-i}|\cdot)F(s_t,(a'^i_{t},a^{-i}_{t}))
\\&=
\mathbb{E}_{(\pi_i,\pi_{-i})}\left[F(s_t,(a^i_{t},a^{-i}_{t}))\right]- \mathbb{E}_{(\pi'_i,\pi_{-i})}\left[F(s_t,(a'^i_{t},a^{-i}_{t}))\right]
\end{align*}
This immediately suggests that we can get the result by multiplying \eqref{pot_game_static_notes} by\\ $\int_{a_i\in\mathcal{A}_i}\int_{a'_i\in\mathcal{A}_i}\int_{a_{-i}\in\mathcal{A}_{-i}}\pi_{i}(da^i|\cdot)\pi_{i}(da'^i|\cdot)\pi_{-i}(da^{-i}|\cdot)$.

For the reverse (in pure strategies) we first consider the case in which the pure strategy is a linear map from some parameterisation. We now readily verify that the reverse holds indeed:
\begin{align*}
&\mathbb{E}_{\pi_i(\boldsymbol{\eta}_i),\pi_{-i}(\boldsymbol{\eta}_{-i})}\left[F(s,a^i,a^{-i})\right]-\mathbb{E}_{\pi_i(\boldsymbol{\eta}'_i),\pi_{-i}(\boldsymbol{\eta}_{-i})}\left[F(s,a^i,a^{-i})\right]
\\&=\int_{a'\in \mathcal{A}^i}\int_{a''\in \mathcal{A}^{-i}}\pi_i(da'|s;\boldsymbol{\eta}_i)\pi_{-i}(da''|s;\boldsymbol{\eta}_{-i})F_i(s,(a',a''))
\\&\qquad-\int_{a'\in \mathcal{A}^i}\int_{a''\in \mathcal{A}^{-i}}\pi_i(da'|s;\boldsymbol{\eta}'_i)\pi_{-i}(da''|s;\boldsymbol{\eta}_{-i})F_i(s,(a',a''))
\\&
=\int_{a'\in \mathcal{A}^i}\int_{a''\in \mathcal{A}^{-i}}da'da''\delta_i(a'-a^i(s;\boldsymbol{\eta}_i))\delta_{-i}(a''-a^{-i}(s;\boldsymbol{\eta}_{-i}))F_i(s,(a',a''))
\\&\qquad-\int_{a'\in \mathcal{A}^i}\int_{a''\in \mathcal{A}^{-i}}da'da''\delta_i(a'-a^i(s;\boldsymbol{\eta}'_i))\delta_{-i}(a''-a^{-i}(s;\boldsymbol{\eta}_{-i}))F_i(s,(a',a''))
\\&
=\Delta F_i(s,(a^i(s;\boldsymbol{\eta}_i),a^{-i}(s;\boldsymbol{\eta}_{-i})))
\end{align*}
which proves the statement in the linear case.

For the general case, consider a strategy which is defined by a map $h:\mathcal{S}\times E\to \mathcal{A}$. We  note that for any $\beta\in\mathbb{R}/\{0\}$ and $x\in X\subset \mathbb{R}$ we have that
\begin{align}
\int_\mathbb{R}\delta(\beta x)=\frac{1}{|\beta|}\int_{\mathbb{R}}\delta(\beta x)dx    \label{simple_delta_equality}
\end{align}
This is true since for any $\beta\in\mathbb{R}/\{0\}$ we can construct the delta function in the following way:
\begin{align*}
\delta(\beta x)=\underset{m\to\infty}{\lim}
\frac{1}{|m|\sqrt{\pi}}e^{-\left(\beta x/m\right)^2}
\end{align*}
Now define $n:=m\beta^{-1}$, then
\begin{align*}
&\underset{m\to\infty}{\lim}
\frac{1}{|m|\sqrt{\pi}}e^{-\left(\beta x/m\right)^2}
\\&=\underset{n\to\infty}{\lim}
\frac{1}{|\beta n|\sqrt{\pi}}e^{-\left(x/n\right)^2}
\\&=\frac{1}{|\beta|}\underset{n\to\infty}{\lim}
\frac{1}{|n|\sqrt{\pi}}e^{-\left( x/n\right)^2}=
\frac{1}{|\beta|}\delta(x)
\end{align*}

In the following, we use the coarea formula (for geometric measures) \cite{simon1983lectures,nicolaescu2011coarea} which says that for any open set $X\subset \mathbb{R}^n$ and for any Lipschitz function $f:X\to\mathbb{R}$ on $X$ and for any $L^1$ function $g$ the following expression holds:
\begin{align*}
\int_X g(x)|\nabla f(x)|dx=\int_{\mathbb{R}}\left(\int_{f^{-1}(s)}g(x)dH_{n-1}(x)\right)ds    
\end{align*}
where $H_{n-1}$ is the $(n-1)-$dimensional Hausdorff measure.

Let us now define $k(s,\boldsymbol{\eta}, a):= h^{-1}(s,\boldsymbol{\eta})-a$.

Now
\begin{align*}
\int_{a\in\mathcal{A}}\pi_{\epsilon=0}(a;s,\boldsymbol{\eta})F_i(s,\boldsymbol{a})da
=\int_{a\in\mathcal{A}}\delta(k(s, a,h(\boldsymbol{\eta})))F_i(s,\boldsymbol{a})da
\end{align*}
By Taylor's theorem, expanding about the point $y$ where $y$ is defined by $k(s,\boldsymbol{\eta}, y)=0$ implies that
\begin{align*}
\int_{a\in\mathcal{A}}\delta(k(s, a,\boldsymbol{\eta}))F_i(s,\boldsymbol{a})da\approx\int_{a\in\mathcal{A}}\delta\left(k'(s, y,\boldsymbol{\eta})(a-y)\right)F_i(s,\boldsymbol{a})da
\end{align*}
Moreover
\begin{align*}
&\int_{a\in\mathcal{A}}\delta(k'(s, y,\boldsymbol{\eta}))(a-y)F_i(s,\boldsymbol{a})da
\\&=\int_{y\in k^{-1}(0)}\left(\int^{y+\epsilon}_{y-\epsilon}\delta\left(k'(s, y,\boldsymbol{\eta})(a-y)\right)F_i(s,\boldsymbol{a})da\right)d\sigma(y)
\end{align*}
where $\sigma$ is a Minkowski content measure.

Define $x:=a- y$ so $dx:=da$ 
\begin{align*}
&\int_{y\in k^{-1}(0)}\left(\int^{y+\epsilon}_{y-\epsilon}\delta\left(k'(s, y,g(\boldsymbol{\eta}))(a-y)\right)F_i(s,\boldsymbol{a})da\right)d\sigma(y)
\\&=\int_{y\in k^{-1}(0)}\left(\int^{x+\epsilon}_{x-\epsilon}\delta\left(k'(s, y,\boldsymbol{\eta})x\right)F_i(s,x+y)dx\right)d\sigma(y)
\\&=\int_{y\in k^{-1}(0)}\left(\int^{x+\epsilon}_{x-\epsilon}\delta(x)\frac{F(s,x+y)}{|k'(s, y,\boldsymbol{\eta})|}dx\right)d\sigma(y)
\\&=\int_{y\in k^{-1}(0)}\frac{F_i(s,y)}{|k'(s, y,\boldsymbol{\eta})|}d\sigma(y)
\\&=\int_{y\in k^{-1}(0)}F_i(s,y)d\sigma(y)
\\&=F_i(s,a)
\end{align*}
where we have used \eqref{simple_delta_equality} in the second step.

Hence we complete the proof by noting that
\begin{align*}
&\int_{a\in\mathcal{A}}\pi_{\epsilon=0}(a;s,\boldsymbol{\eta}')F_i(s,\boldsymbol{a})da-\int_{a\in\mathcal{A}}\pi_{\epsilon=0}(a;s,\boldsymbol{\eta}'')F_i(s,\boldsymbol{a})da
\\&=F_i(s,a(\eta))-F_i(s,a(\eta'))
\end{align*}
as required.
\end{proof}

\begin{proof}[Proof of Lemma \ref{bridge_lemma}]

Recall $\Delta F(s_t,(a^i_{t},a'^i_t),a^{-i}_{t}):= F(s_t,(a'^i_t,a^{-i}_{t}))- F(s_t,(a^i_{t},a^{-i}_{t}))$ define also by \\$\Delta F(s_t,(\pi_i,\pi'_i),\pi_{-i})):=\mathbb{E}_{(\pi_i,\pi_{-i})}\left[F(s_t,(a^i_{t},a^{-i}_{t}))\right]- \mathbb{E}_{(\pi'_i,\pi_{-i})}\left[F(s_t,(a'^i_{t},a^{-i}_{t}))\right]$. We wish to bridge the two cases by proving the following:
\begin{align}
\left|\Delta F(s_t,(a^i_{t},a'^i_t),a^{-i}_{t})-\Delta F(s_t,(\pi_{i,\epsilon},\pi'_{i,\epsilon}),a^{-i})\right|\leq c\bar{\sigma}^2_{\epsilon}.
\end{align}
where $\bar{\sigma}_\epsilon=\max\{\sigma_\epsilon,\sigma'_\epsilon\}$ and $\sigma_\epsilon,\sigma'_\epsilon$ are the variances of the policies $\pi_\epsilon$ and $\pi'_\epsilon$ respectively. 
Indeed,
\begin{align*}
&\quad \left|\mathbb{E}\left[F(s,\pi_{i,\epsilon=0}(\cdot,|s,\eta),a^{-i})-F(s,\pi_{i,\epsilon}(\cdot,|s,\eta),a^{-i})\right]\right|
\\&\leq \mathbb{E}\left[\left|F(s,\pi_{i,\epsilon=0}(\cdot,|s,\eta),a^{-i})-F(s,\pi_{i,\epsilon}(\cdot,|s,\eta),a^{-i})\right|\right]
\\&= \mathbb{E}\Big[\left(1_{\left|F(s,\pi_{i,\epsilon=0}(\cdot,|s,\eta),a^{-i})-F(s,\pi_{i,\epsilon}(\cdot,|s,\eta),a^{-i})\right|>\gamma}+1_{\left|F(s,\pi_{i,\epsilon=0}(\cdot,|s,\eta),a^{-i})-F(s,\pi_{i,\epsilon}(\cdot,|s,\eta),a^{-i})\right|\leq \gamma}\right)
\\&\quad\cdot\left(F(s,\pi_{i,\epsilon=0}(\cdot,|s,\eta),a^{-i})-F(s,\pi_{i,\epsilon}(\cdot,|s,\eta),a^{-i})\right)\Big]
\end{align*}
Now since $F$ is bounded and continuous, we deduce that 
\begin{align*}
&\mathbb{E}\left[\left(1_{\left|F(s,\pi_{i,\epsilon=0}(\cdot,|s,\eta),a^{-i})-F(s,\pi_{i,\epsilon}(\cdot,|s,\eta),a^{-i})\right|>\gamma}\right)\left(F(s,\pi_{i,\epsilon=0}(\cdot,|s,\eta),a^{-i})-F(s,\pi_{i,\epsilon}(\cdot,|s,\eta),a^{-i})\right)\right]\\& \leq \|F\|_\infty\mathbb{E}\left[\left(1_{\left|F(s,\pi_{i,\epsilon=0}(\cdot,|s,\eta),a^{-i})-F(s,\pi_{i,\epsilon}(\cdot,|s,\eta),a^{-i})\right|>\gamma}\right)\right]
\\&=\|F\|_\infty\mathbb{P}\left(\left|F(s,\pi_{i,\epsilon=0}(\cdot,|s,\eta),a^{-i})-F(s,\pi_{i,\epsilon}(\cdot,|s,\eta),a^{-i})\right|>\gamma\right)
\end{align*}
using the properties of the indicator function. Now by the continuity and boundedness of $F$ we deduce that there exists $\delta>0$ such that $a-b>\delta$ whenever $F(a)-F(b)>\gamma$ applying this result we then find that
\begin{align*}\nonumber
&\mathbb{E}\Big[\left(1_{\left|F(s,\pi_{i,\epsilon=0}(\cdot,|s,\eta),a^{-i})-F(s,\pi_{i,\epsilon}(\cdot,|s,\eta),a^{-i})\right|>\gamma}+1_{\left|F(s,\pi_{i,\epsilon=0}(\cdot,|s,\eta),a^{-i})-F(s,\pi_{i,\epsilon}(\cdot,|s,\eta),a^{-i})\right|\leq \gamma}\right)
\\&\cdot\left(F(s,\pi_{i,\epsilon=0}(\cdot,|s,\eta),a^{-i})-F(s,\pi_{i,\epsilon}(\cdot,|s,\eta),a^{-i})\right)\Big]
\\&\leq \|F\|_\infty\mathbb{P}\left(\left|F(s,\pi_{i,\epsilon=0}(\cdot,|s,\eta),a^{-i})-F(s,\pi_{i,\epsilon}(\cdot,|s,\eta),a^{-i})\right|>\gamma\right)+\gamma
\\&\leq \|F\|_\infty\mathbb{P}\left(\left|\pi_{i,\epsilon=0}(\cdot,|s,\eta)-\pi_{i,\epsilon}(\cdot,|s,\eta)\right|>\delta\right)+\gamma
\\&\leq \delta^{-2}\|F\|_\infty\sigma_{\epsilon}^2+\gamma
\end{align*}
where we have used Tschebeyshev's inequality in the last line.

Now since $\gamma$ is arbitrary we deduce that
\begin{align}
&\quad \left|\mathbb{E}\left[F(s,\pi_{i,\epsilon=0}(\cdot,|s,\eta),a^{-i})-F(s,\pi_{i,\epsilon}(\cdot,|s,\eta),a^{-i})\right]\right|\leq \delta^{-2}\|F\|_\infty\sigma_{\epsilon}^2
\end{align}

Moreover
\begin{align*}\nonumber
&\left| \Delta F(s_t,(a^i_{t},a'^i_t),a^{-i}_{t})-\Delta F(s_t,(\pi_i,\pi'_i),\pi_{-i}))\right|
\\&\nonumber= \left| F(s_t,(a'^i_t,a^{-i}_{t}))- F(s_t,(a^i_{t},a^{-i}_{t}))-\left(\mathbb{E}_{(\pi_i,\pi_{-i})}\left[F(s_t,(a^i_{t},a^{-i}_{t}))\right]- \mathbb{E}_{(\pi'_i,\pi_{-i})}\left[F(s_t,(a'^i_{t},a^{-i}_{t}))\right]\right)\right|
\\&\nonumber\leq 
\left|F(s_t,(a^i_t,a^{-i}_{t}))-\mathbb{E}_{(\pi_i,\pi_{-i})}\left[F(s_t,(a^i_{t},a^{-i}_{t}))\right]\right| +\left|F(s_t,(a'^i_{t},a^{-i}_{t}))- \mathbb{E}_{(\pi'_i,\pi_{-i})}\left[F(s_t,(a'^i_{t},a^{-i}_{t}))\right]\right|
\\&\leq c\|F\|_\infty\bar{\sigma}_{\epsilon}^2
\end{align*}

We deduce the last statement by applying the result to the sequence of $\epsilon/n$ and by the sandwich theorem. 
\end{proof}

% \section{Proofs of Theoretical Results} 
% in Section \ref{sec:analysis_of_spot-q}}
\begin{lemma}\label{dynamic_potential_form}
The function $B$ is given by the following expression for $s\in \mathcal{S}, \forall \boldsymbol{\pi}\in\boldsymbol{\Pi}$:
\begin{align*}
B^{\boldsymbol{\pi}}(s)-B^{\boldsymbol{\pi'}}(s)
=\mathbb{E}_{s_t\sim P}\left[\sum_{t=0}^\infty\sum_{i\in\mathcal{N}}\gamma^t\int_0^1\boldsymbol{\gamma}'(z)\frac{\partial R_i}{\partial \pi_i}(s_t,\boldsymbol{\gamma}(z))\Big|s=s_0\right],
\end{align*}
where $\boldsymbol{\gamma}(z)$ is a continuous differentiable path in $\boldsymbol{\Pi}$ connecting two strategy profiles $\boldsymbol{\pi}\in\boldsymbol{\Pi}$ and $\boldsymbol{\pi'}\in\boldsymbol{\Pi}$.
\end{lemma}
\begin{proof}
We note that from \eqref{parameterisation_potential_condition}, using the gradient theorem of vector calculus, it is straightforward to deduce that the potential function $\phi$ can be computed from the reward functions $(R_{i})_{i\in\mathcal{N}}$ via the following expression \cite{monderer1996potential}:
\begin{align} 
\phi^{\boldsymbol{\pi}}(s)=\phi^{\boldsymbol{\pi'}}(s)+\sum_{i\in\mathcal{N}}\int_0^1\boldsymbol{\gamma}'(z)\frac{\partial R_i}{\partial \pi_i}(s_t,\boldsymbol{\gamma}(z)), \label{line_integral_potential_equation} 
\end{align}
where $\boldsymbol{\gamma}(z)$ is a continuous differentiable path in $\boldsymbol{\Pi}$ connecting two strategy profiles $\boldsymbol{\pi}\in\boldsymbol{\Pi}$ and $\boldsymbol{\pi'}\in\boldsymbol{\Pi}$.

We then deduce the result (in the finite case) after inserting (\ref{line_integral_potential_equation}) into (\ref{expression_for_B}).
\end{proof}

\begin{proof}[Proof of Prop. \ref{SL_convergence}]
Recall the following definitions:
\begin{align}
F_i(s,\boldsymbol{\eta},\rho):=\int_{\mathcal{A}}\boldsymbol{\pi}_\epsilon(\boldsymbol{da},\boldsymbol{\eta},s)\frac{\partial}{\partial \eta_i}\pi_{i,\epsilon}(a^i,\eta_i,s)\nabla F_\rho(s,\boldsymbol{a})    
\end{align}
and
\begin{align}
U(s,\boldsymbol{\eta},\rho):=\int_{\mathcal{A}}\boldsymbol{\pi}_\epsilon(\boldsymbol{da},\boldsymbol{\eta},s)\frac{\partial}{\partial \eta_i}\pi_{i,\epsilon}(a^i,\eta_i,s)\nabla R_i(s,\boldsymbol{a}),   
\end{align}
where we have used the shorthand: $\boldsymbol{\pi}_\epsilon(\boldsymbol{a},\boldsymbol{\eta},s):=\pi_i(a,s,\eta_i)\pi_{-i}(a_i,s,\eta_{-i})$

% It is a simple matter to see that the proof of Prop. \ref{Lipschitz_potential} contains the proof of Lemma \ref{Lipschitz_potential_itself}.
Since $U$ and $F_i$ are locally Lipschitz continuous each can have at most polynomial growth. By the H\"{o}lder inequality we find that:
\begin{align}
&\int_{\Omega}\sum_{i\in\mathcal{N}}\left|F_i(s,\boldsymbol{\eta},\rho)-u(s,\boldsymbol{\eta})\right|^2d\nu_1(s,\boldsymbol{\eta})\nonumber
\\&=\int_{\Omega}\sum_{i\in\mathcal{N}}\left|\int_{\mathcal{A}}\boldsymbol{\pi}_\epsilon(\boldsymbol{da},\boldsymbol{\eta},s)\frac{\partial}{\partial \eta_i}\pi_{i,\epsilon}(a^i,\eta_i,s)\nabla F_\rho(s,\boldsymbol{a})-\int_{\mathcal{A}}\boldsymbol{\pi}_\epsilon(\boldsymbol{da},\boldsymbol{\eta},s)\frac{\partial}{\partial \eta_i}\pi_{i,\epsilon}(a^i,\eta_i,s)\nabla R_i(s,\boldsymbol{a})\right|^2d\nu_1(s,\boldsymbol{\eta})\nonumber
\\&\leq\int_{\Omega}\sum_{i\in\mathcal{N}}\left|\int_{\mathcal{A}}\boldsymbol{\pi}_\epsilon(\boldsymbol{da},\boldsymbol{\eta},s)\left(\frac{\partial}{\partial \eta_i}\pi_{i,\epsilon}(a^i,\eta_i,s)\nabla F_\rho(s,\boldsymbol{a})-\frac{\partial}{\partial \eta_i}\pi_{i,\epsilon}(a^i,\eta_i,s)\nabla R_i(s,\boldsymbol{a})\right)\right|^2d\nu_1(s,\boldsymbol{\eta})\nonumber
\\&\leq \int_{\Omega}\left(\sum_{i\in\mathcal{N}}\underset{\boldsymbol{a}\in\mathcal{A}}{\max}\left|\frac{\partial}{\partial \eta_i}\pi_{i,\epsilon}(a^i,\eta_i,s)\nabla F_\rho(s,\boldsymbol{a})\right|^{q_i}+\underset{\boldsymbol{a}\in\mathcal{A}}{\max}N\left|\frac{\partial}{\partial \eta_i}\pi_{i,\epsilon}(a^i,\eta_i,s)\nabla R_i(s,\boldsymbol{a})\right|^r\right)
\\&\cdot\left(\sum_{i\in\mathcal{N}}\left|\int_{\mathcal{A}}\boldsymbol{\pi}_\epsilon(\boldsymbol{da},\boldsymbol{\eta},s)\left(\frac{\partial}{\partial \eta_i}\pi_{i,\epsilon}(a^i,\eta_i,s)\nabla F_\rho(s,\boldsymbol{a})-\frac{\partial}{\partial \eta_i}\pi_{i,\epsilon}(a^i,\eta_i,s)\nabla R_i(s,\boldsymbol{a})\right)\right|^2\right)d\nu_1(s,\boldsymbol{\eta})\nonumber
\\&\leq N\int_{\Omega}\left(\sum_{i\in\mathcal{N}}\left|\frac{\partial}{\partial \eta_i}\pi_{i,\epsilon}(a^i,\eta_i,s)\right|^{q_i\wedge r}\underset{\boldsymbol{a}\in\mathcal{A}}{\max}\left(\left|\nabla F_\rho(s,\boldsymbol{a})\right|^{q_i}+\left|\nabla R_i(s,\boldsymbol{a})\right|^r\right)\right)\nonumber
\\&\cdot\left(\sum_{i\in\mathcal{N}}\left|\frac{\partial}{\partial \eta_i}\pi_{i,\epsilon}(a^i,\eta_i,s)\right|^2\left|\int_{\mathcal{A}}\boldsymbol{\pi}_\epsilon(\boldsymbol{da},\boldsymbol{\eta},s)\left(\nabla F_\rho(s,\boldsymbol{a})-\nabla R_i(s,\boldsymbol{a})\right)\right|^2\right)d\nu_1(s,\boldsymbol{\eta})\nonumber
\\&\leq N\left(\int_{\Omega}\left(\sum_{i\in\mathcal{N}}\left|\frac{\partial}{\partial \eta_i}\pi_{i,\epsilon}(a^i,\eta_i,s)\right|^{q_i\wedge r}\underset{\boldsymbol{a}\in\mathcal{A}}{\max}\left(\left|\nabla F_\rho(s,\boldsymbol{a})\right|^{q_i}+\left|\nabla R_i(s,\boldsymbol{a})\right|^r\right)\right)^{r_1}d\nu_1(s,\boldsymbol{\eta})\right)^{1/r_1}\nonumber
\\& \qquad\cdot \left(\int_{\Omega}\left(\sum_{i\in\mathcal{N}}\left|\frac{\partial}{\partial \eta_i}\pi_{i,\epsilon}(a^i,\eta_i,s)\right|^2\left|\int_{\mathcal{A}}\boldsymbol{\pi}_\epsilon(\boldsymbol{da},\boldsymbol{\eta},s)\left(\nabla F_\rho(s,\boldsymbol{a})-\nabla R_i(s,\boldsymbol{a})\right)\right|^2\right)^{r_2}d\nu_1(s,\boldsymbol{\eta})\right)^{1/r_2}\nonumber
\\&\leq cN\left(\int_{\Omega}\left(\sum_{i\in\mathcal{N}}\left|\frac{\partial}{\partial \eta_i}\pi_{i,\epsilon}(a^i,\eta_i,s)\right|^{q_i\wedge r}\underset{\boldsymbol{a}\in\mathcal{A}}{\max}\left(\left|\nabla F_\rho(s,\boldsymbol{a})-\nabla R_i(s,\boldsymbol{a})\right|^{q_i}+\left|\nabla R_i(s,\boldsymbol{a})\right|^{r\wedge q_i}\right)\right)^{r_1}d\nu_1(s,\boldsymbol{\eta})\right)^{1/r_1}\nonumber
\\& \qquad\cdot \left(\int_{\Omega}\left(\sum_{i\in\mathcal{N}}\left|\frac{\partial}{\partial \eta_i}\pi_{i,\epsilon}(a^i,\eta_i,s)\right|^2\left|\int_{\mathcal{A}}\boldsymbol{\pi}_\epsilon(\boldsymbol{da},\boldsymbol{\eta},s)\left(\nabla F_\rho(s,\boldsymbol{a})-\nabla R_i(s,\boldsymbol{a})\right)\right|^2\right)^{r_2}d\nu_1(s,\boldsymbol{\eta})\right)^{1/r_2}\nonumber
\\&\leq cN\left(\int_{\Omega}\left(\sum_{i\in\mathcal{N}}\left|\frac{\partial}{\partial \eta_i}\pi_{i,\epsilon}(a^i,\eta_i,s)\right|^{q_i\wedge r}\underset{\boldsymbol{a}\in\mathcal{A}}{\max}\left(\left|\nabla \left[F_\rho(s,\boldsymbol{a})- R_i(s,\boldsymbol{a})\right]\right|^{q_i}+\left|\nabla R_i(s,\boldsymbol{a})\right|^{r\wedge q_i}\right)\right)^{r_1}d\nu_1(s,\boldsymbol{\eta})\right)^{1/r_1}\nonumber
\\& \qquad\cdot \left(\int_{\Omega}\left(\sum_{i\in\mathcal{N}}\left|\frac{\partial}{\partial \eta_i}\pi_{i,\epsilon}(a^i,\eta_i,s)\right|^2\left|\int_{\mathcal{A}}\boldsymbol{\pi}_\epsilon(\boldsymbol{da},\boldsymbol{\eta},s)\left(\nabla\left[ F_\rho(s,\boldsymbol{a})- R_i(s,\boldsymbol{a})\right]\right)\right|^2\right)^{r_2}d\nu_1(s,\boldsymbol{\eta})\right)^{1/r_2}\label{gradient_optimisation_pen_step}
\end{align}
Now for any $l\in\mathbb{R}_{>0}$ we have that
\begin{align}\nonumber
&\left|\nabla\left[ F_\rho(s,\boldsymbol{a})- R_i(s,\boldsymbol{a})\right]\right|^l
\\&\leq \left|\nabla\left[ F_\rho(s,\boldsymbol{a})-\phi(s,\boldsymbol{a})+\phi(s,\boldsymbol{a})- R_i(s,\boldsymbol{a})\right]\right|^l\nonumber
\\&\leq \left|\nabla\left[ F_\rho(s,\boldsymbol{a})-\phi(s,\boldsymbol{a})\right]\right|^l+\left|\nabla\left[\phi(s,\boldsymbol{a})- R_i(s,\boldsymbol{a})\right]\right|^l\nonumber
\\&\leq \left|\nabla\left[ F_\rho(s,\boldsymbol{a})-\phi(s,\boldsymbol{a})\right]\right|^l\label{subs_r_diff_phi_diff}
\end{align}
using the potentiality property.

Inserting \eqref{subs_r_diff_phi_diff} into \eqref{gradient_optimisation_pen_step} yields

\begin{align*}
&\int_{\Omega}\sum_{i\in\mathcal{N}}\left|F_i(s,\boldsymbol{a},\rho)-u(s,\boldsymbol{a})\right|^2d\nu_1(s,\boldsymbol{\eta})
\\&\leq cN\left(\int_{\Omega}\left(\sum_{i\in\mathcal{N}}\left|\frac{\partial}{\partial \eta_i}\pi_{i,\epsilon}(a^i,\eta_i,s)\right|^{q_i\wedge r}\underset{\boldsymbol{a}\in\mathcal{A}}{\max}\left(\left|\nabla \left[F_\rho(s,\boldsymbol{a})- \phi(s,\boldsymbol{a})\right]\right|^{q_i}+\left|\nabla R_i(s,\boldsymbol{a})\right|^{r\wedge q_i}\right)\right)^{r_1}d\nu_1(s,\boldsymbol{\eta})\right)^{1/r_1}\nonumber
\\& \qquad\cdot \left(\int_{\Omega}\left(\sum_{i\in\mathcal{N}}\left|\frac{\partial}{\partial \eta_i}\pi_{i,\epsilon}(a^i,\eta_i,s)\right|^2\left|\int_{\mathcal{A}}\boldsymbol{\pi}_\epsilon(\boldsymbol{da},\boldsymbol{\eta},s)\left(\nabla\left[ F_\rho(s,\boldsymbol{a})- \phi(s,\boldsymbol{a})\right]\right)\right|^2\right)^{r_2}d\nu_1(s,\boldsymbol{\eta})\right)^{1/r_2}
\\&\leq cN^2\left(\underset{i\in\mathcal{N}}{\max}\epsilon^{q_i}+\underset{i\in\mathcal{N}}{\max}\|\bar{R}_{\nabla}\|^{m\wedge q_i}_{\infty}\right)\epsilon^2
\end{align*}
where the last line follows from the boundedness of $\nabla R_i\leq \bar{R}_{\nabla}$ and $\frac{\partial \pi_{i,\epsilon}}{\partial \eta_i}$ and Lemma 6 in \cite{bertsekas2000gradient}.
\end{proof}

\begin{proof}[Proof of Prop. \ref{Lipschitz_potential_itself}]
We first show that the Bellman operator is a contraction. Indeed, for any bounded $F,F'\in\mathcal{H}$ we have 
\begin{align*}
     \left\|\left[T_{\phi} F\right]-\left[T_{\phi} F'\right]\right\|\leq \gamma d
\end{align*}
where $d:=\|F-F'\|_\infty$ using the fact that $T_{\phi}$ is monotonic.  

We now observe that 
\begin{align*}
    \left\|\left[T^k_{\phi} F\right]-\left[T^k_{\phi_{\epsilon}} F'\right]\right\|&\leq \left\|\phi-\phi_{\epsilon}\right\|_\infty+ \gamma\left\|T^{k-1}_{\phi}F-T^{k-1}_{\phi_{\epsilon}}F'\right\|_\infty
    \\&\leq \epsilon +\sum_{j=0}^{m-1}\gamma^j\left\|\phi-\phi_{\epsilon}\right\|_\infty +\gamma^m\left\|T^{k-m}_{\phi}F-T^{k-m}_{\phi_{\epsilon}}F'\right\|_\infty
    \\&\leq \epsilon +\sum_{j=0}^{k-1}\gamma^j\left\|\phi-\phi_{\epsilon}\right\|_\infty +\gamma^kd
        \\&\leq \epsilon\left(1+\frac{1-\gamma^k}{1-\gamma}\right) +\gamma^kd
\end{align*}
and hence $
    \underset{k\to\infty}{\lim}\left\|\left[T^k_{\phi} F\right]-\left[T^k_{\phi_{\epsilon}} F'\right]\right\|\leq c\epsilon$ 
for some $c>0$ from which we deduce the result.
\end{proof}

\end{document}